\documentclass{lmcs} 
\pdfoutput=1

\usepackage{lastpage}
\lmcsdoi{15}{1}{26}
\lmcsheading{}{\pageref{LastPage}}{}{}%
{Aug.~10,~2018}{Feb.~20,~2020}{}


\usepackage{hyperref}

\usepackage{amssymb}
\usepackage{amsmath}
\usepackage{stmaryrd}
\usepackage{tikz-cd}
\usepackage{amsthm}
\usepackage{mathpartir}

\usepackage{xcolor}
\usepackage{xspace}
\usepackage{listings}
\usepackage{url}
\usepackage{doi}

\newif\ifsubmit%
\submittrue%
\ifsubmit%
\newcommand{\EZ}[1]{{#1}}
\newcommand{\FD}[1]{{#1}}

\newcommand{\EZComm}[1]{}
\newcommand{\DAComm}[1]{}
\newcommand{\FDComm}[1]{}
\else
\newcommand{\FD}[1]{ {\color{red}#1} }
\newcommand{\EZ}[1]{ {\color{blue}#1} }

\newcommand{\EZComm}[1]{{\scriptsize\color{blue}[\textbf{Elena: }#1]}}
\newcommand{\DAComm}[1]{{\scriptsize\color{magenta}[\textbf{Davide: }#1]}}
\newcommand{\FDComm}[1]{{\scriptsize\color{red} [\textbf{Francesco: }#1]}}
\fi

\newif\ifrev%
\revfalse%
\ifrev%
\newcommand{\Rev}[1]{{\color{orange} #1}}
\else
\newcommand{\Rev}[1]{#1}
\fi

\newcommand{\Space}{\hskip 0.7em}
\newcommand{\BigSpace}{\hskip 1.5em}

\newcommand{\refToFigure}[1]{Figure~\ref{fig:#1}}
\newcommand{\refToSection}[1]{Section~\ref{sect:#1}}

\newcommand{\refToTheorem}[1]{Theorem~\ref{theo:#1}}
\newcommand{\refToCorollary}[1]{Corollary~\ref{cor:#1}}
\newcommand{\refToLemma}[1]{Lemma~\ref{lem:#1}}
\newcommand{\refToDefinition}[1]{Definition~\ref{def:#1}}
\newcommand{\refToProposition}[1]{Proposition~\ref{prop:#1}}
\newcommand{\refToFact}[1]{Fact~\ref{fact:#1}}

\DeclareMathOperator*{\argmin}{arg\,min}

\newcommand{\fun}[3]{{#1}:{#2}\rightarrow{#3}}
\newcommand{\Id}[1]{\textsf{id}_{#1}}
\newcommand{\dom}{\mathit{dom}}

\newcommand{\Tuple}[1]{\left({#1}\right)}
\newcommand{\Pair}[2]{\Tuple{{#1},\,{#2}}}
\newcommand{\Triple}[3]{\Tuple{{#1},\,{#2},\,{#3}}}

\newcommand{\N}{\mathbb{N}}
\newcommand{\Z}{\mathbb{Z}}

\newcommand{\judg}{\mathit{j}}
\newcommand{\prem}{\textit{Pr}}
\newcommand{\cons}{\textit{c}}
\newcommand{\is}{\mathcal{I}}
\newcommand{\cois}{\is^{\textsf{co}}}
\newcommand{\universe}{\mathcal{U}}
\newcommand{\myrule}{\Rule{\prem}{\cons}}
\newcommand{\Rule}[2]{
	\displaystyle%
	\frac{#1}{#2}
}
\newcommand{\CoRule}[2]{
	\displaystyle%
	\genfrac{}{}{1pt}{}{#1}{#2}
}
\newcommand{\CoAxiom}[1]{
	\CoRule{  }{#1}
}
\newcommand{\Op}[1]{{\textit{F}_{#1}}}
\newcommand{\IterOp}[2]{\textit{F}^{#2}_{#1}}
\newcommand{\Ind}[1]{\textit{Ind}(#1)} 
\newcommand{\CoInd}[1]{\textit{CoInd}(#1)} 
\newcommand{\Generated}[2]{\textit{Gen}(#1,#2)} 
\newcommand{\coaxioms}{\gamma}

\newcommand{\DefSet}{\mathcal{D}}
\newcommand{\Spec}{\mathcal{S}}

\newcommand{\CoIndPrinciple}{\textsc{\small (CoInd)}\xspace}
\newcommand{\IndPrinciple}{\textsc{\small (Ind)}\xspace}
\newcommand{\SpecAllPos}{\Spec^\textit{allPos}}

\newcommand{\EList}{\Lambda}
\newcommand{\Cons}[2]{{#1}{::}{#2}}
\newcommand{\LSet}{\mathbb{L}}
\newcommand{\LInfSet}{\LSet^\infty}
\newcommand{\Bool}{\mathbb{B}}
\newcommand{\True}{\textsf{T}}
\newcommand{\False}{\textsf{F}}

\newcommand{\xs}{\mathit{xs}}
\newcommand{\elemsName}{\textit{elems}}
\newcommand{\memberName}{\textit{member}}
\newcommand{\allPosName}{\textit{allPos}}
\newcommand{\maxElemName}{\textit{maxElem}}

\newcommand{\elems}[2]{\elemsName{\Pair{#1}{#2}}}
\newcommand{\member}[3]{\memberName{\Triple{#1}{#2}{#3}}}
\newcommand{\memberPred}[2]{\memberName{\Pair{#1}{#2}}}
\newcommand{\allPos}[2]{\allPosName{\Pair{#1}{#2}}}
\newcommand{\allPosPred}[1]{\allPosName{\Tuple{#1}}}

\newcommand{\maxElem}[2]{\maxElemName{\Pair{#1}{#2}}}

\newcommand{\Sub}[1]{\textsf{Sub}(#1)} 
\newcommand{\GroupGen}[1]{\left\langle{#1}\right\rangle}

\newcommand{\node}{\textit{v}}
\newcommand{\anode}{\textit{u}}
\newcommand{\nodeset}{\mathcal{N}}
\newcommand{\Visit}[2]{#1{\stackrel{\star}{\rightarrow}}#2}
\newcommand{\adj}{\textit{adj}}
\newcommand{\Nodes}{\textit{V}}
\newcommand{\Labels}{\mathcal{L}}
\newcommand{\T}{\mathcal{T}}
\newcommand{\CTree}[1]{\T^{\textsf{ci}}(#1)} 

\newcommand{\Tree}[2][{}]{\T_{#1}{(#2)}}

\newcommand{\Path}{\textsf{P}}
\newcommand{\Paths}[1]{\textsf{path}{(#1)}}
\newcommand{\dsub}{\textit{dsub}}
\newcommand{\chl}{\textit{chl}}
\newcommand{\dist}[3]{\textit{dist}{\Triple{#1}{#2}{#3}}}
\newcommand{\minPath}[4]{\textit{spath}{\Tuple{{#1},\,{#2},\,{#3},\,{#4}}}}
\newcommand{\Edges}{\textit{E}}
\newcommand{\TOrder}{\triangleleft}
\newcommand{\ApproxTOrder}[1]{\TOrder_{#1}}
\newcommand{\TLub}{\bigvee}
\newcommand{\ApproxEq}[1]{\bowtie_{#1}}

\newcommand{\String}[1]{{#1}^\star}
\newcommand{\EString}{\varepsilon}
\newcommand{\Len}[1]{\left | {#1} \right |}

\newcommand{\nonterminal}{\textit{A}}
\newcommand{\firstset}{\mathcal{F}}
\newcommand{\first}[1]{\textit{first}(#1)} 
\newcommand{\First}[2]{\textit{first}(#1,#2)} 
\newcommand{\produzioneinline}[2]{#1::=#2}
\newcommand{\simb}{\sigma} 

\newcommand{\isin}{\ensuremath{\mathit{is\_in}}}
\newcommand{\isinZero}{\ensuremath{\mathit{is\_in0}}}
\newcommand{\pathZero}{\ensuremath{\mathit{path0}}}
\newcommand{\tree}{\ensuremath{\mathit{tree}}}
\newcommand{\add}{\ensuremath{\mathit{add}}}
\newcommand{\sem}[1]{\ensuremath{\left\llbracket{#1}\right\rrbracket}}

\newcommand{\infv}{\ensuremath{v_\infty}}
\newcommand{\eval}[2]{\ensuremath{{#1}\Rightarrow{#2}}}
\newcommand{\subs}[3]{\ensuremath{{#1}[{#2}\leftarrow{#3}]}}
\newcommand{\RuleName}[3]{
\displaystyle%
\mbox{({#1})}\,\frac{#2}{#3}
}
\newcommand{\CoAxiomName}[2]{
\displaystyle%
\mbox{({#1})}\, \CoAxiom{#2}
}

\newcommand{\order}{\sqsubseteq}

\newcommand{\LowSet}[1]{{\downarrow}\,{#1}}
\newcommand{\UpSet}[1]{{\uparrow}\,{#1}}
\newcommand{\lub}{\bigsqcup}
\newcommand{\glb}{\bigsqcap}
\newcommand{\function}{\mathit{F}}
\newcommand{\afunction}{\mathit{G}}

\newcommand{\lattice}{\mathit{L}}

\newcommand{\meet}{\sqcap}
\newcommand{\join}{\sqcup}

\newcommand{\Pre}[1]{\textsf{pre}(#1)} 
\newcommand{\Post}[1]{\textsf{post}(#1)} 
\newcommand{\Fix}[1]{\textsf{fix}(#1)} 
\newcommand{\lfp}{\mu}
\newcommand{\gfp}{\nu}

\newcommand{\Iterate}[2]{\mathit{I}_{{#1}, {#2}} }

\newcommand{\CSys}{\mathit{C}}
\newcommand{\KSys}{\mathit{K}}
\newcommand{\closure}{\nabla}
\renewcommand{\ker}{\Delta}
\newcommand{\bound}{\beta}
\newcommand{\Restricted}[2]{{#1_{{\sqcap}#2}}}
\newcommand{\Extended}[2]{{#1_{{\sqcup}#2}}}

\newcommand{\FJ}{\textsc{FJ}\xspace}
\newcommand{\coFJ}{\textsc{coFJ}\xspace}

\newcommand{\LambdaCalculus}{$\lambda$-calculus\xspace}

\newcommand{\Restrict}[1]{_{\mid #1}} 

\theoremstyle{plain} 

\begin{document}

\title{Coaxioms: flexible coinductive definitions by inference systems}

\author[F.~Dagnino]{Francesco Dagnino}	
\address{DIBRIS, University of Genova}	
\email{francesco.dagnino@dibris.unige.it}  

%





\begin{abstract}
  We introduce a generalized notion of inference system to support more flexible interpretations of recursive definitions.
Besides axioms and inference rules with the usual meaning,  we allow also  \emph{coaxioms}, which are, intuitively, axioms which can only be applied ``at infinite depth'' in a proof tree.
Coaxioms allow us to interpret recursive definitions as fixed points which are not necessarily the least, nor the greatest one, \FD{whose existence is guaranteed by a smooth extension of classical results. }
This notion nicely subsumes standard inference systems and their inductive and coinductive interpretation, thus allowing formal reasoning in cases where the inductive and coinductive interpretation do not provide the intended meaning, \FD{but are rather} mixed together.

This is a corrected version of the paper (\url{https://arxiv.org/abs/1808.02943v4}) published originally on 12 March 2019.

\end{abstract}

\maketitle


\section{Introduction}\label{sect:intro}

Recursive definitions are everywhere in computer science.
They allow very compact and intuitive definition of several types of objects: data types, predicates and functions.
Furthermore, they are also essential in programming languages, especially for declarative paradigms, to write non-trivial programs.

Assigning a {formal} semantics to recursive definitions  is not an easy task;
usually, a recursive definition is associated with a monotone function on a partially ordered set, or, more generally, with a functor on a category, and the semantics is defined to be  a \emph{fixed point} of such function/functor~\cite{JacobsRutten97}.
However, in general, a monotone function (a functor) has several fixed points, hence the problem is how to choose the right one, that is, the fixed point that matches the intended semantics.

The most {widely known} semantics for recursive definitions  is the \emph{inductive} one~\cite{Aczel77}, which corresponds to the least fixed point/{initial algebra}.
This interpretation works perfectly in all cases where we can reach a base case in a finite number of steps, this is the case, for instance, when we deal with well-founded (algebraic) objects (such as natural numbers, finite lists, finite trees etc.). 

Nevertheless, in some cases the inductive interpretation is not {appropriate}.
This is the case, for instance, when we deal with circular, or more generally non-well-founded (coalgebraic) objects (graphs, infinite lists, infinite trees, etc.), where clearly we are not guaranteed to reach a base case.
Here a possibility is to choose the dual to induction: the \emph{coinductive} semantics~\cite{Aczel88, Rutten00, Jacobs16}, corresponding to the greatest fixed point/{final coalgebra}.

Therefore we have two strongly opposite options to interpret recursive definitions: the inductive (least) semantics or the coinductive (greatest) semantics.
However, as we will see, there are cases where neither of these two dual solutions is suitable, hence the need of more flexibility to choose the desired fixed point. 

On the programming language side, the {most widely adopted} semantics for recursive definition is again the inductive one.
However, the support for coinductive semantics have been provided both in logic programming~\cite{SimonMBG06, SimonBMG07,KomendantskayaJ15} and in functional programming~\cite{Hagino87,BirdWadler88}.
In both cases, the same dichotomy described above emerges, hence, recently, some operational models to support  more flexible definitions on non-well-founded structures have been proposed: in the logic  paradigm~\cite{Ancona13}, in the functional paradigm~\cite{JeanninKS13, JeanninKS17} and in the {object-oriented} paradigm~\cite{AnconaZucca12, AnconaZucca13}
%

In this paper, we propose a framework to interpret recursive definitions as fixed points that are not necessarily the least, nor the greatest one.
More precisely, we will extend the standard and well-known framework of \emph{inference systems}, where recursive definitions are represented as sets of \emph{rules}, as we will formally define in the next section.

In order to illustrate the complexity of interpreting recursive definitions especially in presence of non-well-founded structures, and to introduce the idea behind our proposal, let us consider some examples on lists of integers.
In the following, $l$ will range over finite or infinite lists and  $x, y, z$ {over} integers, $\EList$ is the empty list and $\Cons{-}{-}$ is the list constructor.
We start with the simple predicate $\memberPred{x}{l}$ stating that the element $x$ belongs to $l$, {defined as follows}:
\begin{mathpar}
\Rule{}{ \memberPred{x}{\Cons{x}{l}} }
\and
\Rule{
	\memberPred{x}{l}
}{ \memberPred{x}{\Cons{y}{l}} } x\ne y
\end{mathpar}
The standard way to interpret an inference system is the inductive one, {which} consists of the set of judgements having a finite derivation.
For the above definition, the inductive interpretation works perfectly in all cases, also for infinite lists.
Intuitively, this is due to the fact that in all cases, in order to establish that $\memberPred{x}{l}$ holds, we have to {find} $x$ in $l$, and{,} if $x$ actually belongs to $l$, we find it in finitely many steps.

Let us consider another example: the predicate $\allPosPred{l}$ {stating that} $l$ contains only strictly positive natural numbers.
\begin{mathpar}
\Rule{}{ \allPosPred{\EList} }
\and
\Rule{
	\allPosPred{l}
}{ \allPosPred{\Cons{x}{l}} } x>0
\end{mathpar}
Here the inductive interpretation still works well on finite lists, but fails on infinite lists, since, intuitively, to establish whether $l$ contains only positive  elements, we need to inspect the whole list, and this cannot be done with a finite derivation for an infinite list.

Therefore, we have to switch to the coinductive interpretation, considering as semantics the set of judgements having an arbitrary (finite or infinite) derivation.
This is indeed the correct way to get the expected semantics also on infinite lists.

We now consider a slight variation of these two examples.
Let $\Bool = \{\True, \False\}$ be the set of truth values,  consider the judgements $\member{x}{l}{b}$ and $\allPos{l}{b}$ with $b \in \Bool$ such that
\begin{itemize}
\item $\member{x}{l}{\True}$ holds iff $\memberPred{x}{l}$ holds, and otherwise $\member{x}{l}{\False}$ holds
\item $\allPos{l}{\True}$ holds iff $\allPosPred{l}$ holds, and otherwise $\allPos{l}{\False}$ holds
\end{itemize}
We can define these judgements by means of the following inference systems
\begin{small}
\begin{mathpar}
{\Rule{}{ \member{x}{\EList}{\False} }  }
\and
\Rule{}{ \member{x}{\Cons{x}{l}}{\True} }
\and
\Rule{ \member{x}{l}{b} }{ \member{x}{\Cons{y}{l}}{b} } x \ne y
\\
\Rule{}{ \allPos{\EList}{\True} }
\and
\Rule{}{ \allPos{\Cons{x}{l}}{\False} } x\le 0
\and
\Rule{ \allPos{l}{b} }{ \allPos{\Cons{x}{l}}{b} }x>0
\end{mathpar}
\end{small}
For both definitions, neither the inductive interpretation, nor the coinductive one works well on infinite lists.
For the judgement $\member{x}{l}{b}$, with the inductive interpretation we cannot derive any judgement of shape $\member{x}{l}{\False}$ where $l$ is an infinite list and $x$ does not belong to $l$, while with the coinductive interpretation we get both $\member{x}{l}{\False}$ and $\member{x}{l}{\True}$.
For the judgement $\allPos{l}{b}$, with the inductive interpretation we cannot derive any judgement of shape $\allPos{l}{\True}$ where $l$ is an infinite list containing only positive elements, while with the coinductive interpretation we get both $\allPos{l}{\True}$ and $\allPos{l}{\False}$.

We consider now a last example, defining the  predicate $\maxElem{l}{x}$ stating that $x$ is the maximum of the list $l$.
The definition is given by the following inference system
\begin{mathpar}
\Rule{}{ \maxElem{\Cons{x}{\EList}}{x} }
\and
\Rule{
	\maxElem{l}{y}
}{ \maxElem{\Cons{x}{l}}{z} }\ z = \max\{x, y\}
\end{mathpar}
The inductive interpretation works well on finite list{s}, but does not allow to derive any judgement {on} infinite lists, again, because, to compute a maximum, we need to inspect the whole list.
The coinductive interpretation still works well on finite lists, but, again,  we can derive too many judgements regarding infinite lists: for instance{,} if $l$ is the infinite list of 1s, we can derive both $\maxElem{l}{1}$, which is correct, and $\maxElem{l}{2}$, that is clearly wrong, since $2$ does not belong to $l$.

All these examples point out that the inductive interpretation {cannot} properly deal with non-well-founded structure, while the coinductive one allows {the derivation of} too many judgements.
Hence we need a way to {``}filter out'' some infinite derivations, in order to restrict the coinductive interpretation to the intended semantics.
We make this possible by introducing  \emph{coaxioms}.

Coaxioms are special rules that need to be {provided} together with {standard rules} in order to control {their} semantics.
Intuitively, they are axioms that can be only applied ``at infinite depth'' in a derivation.
From a model-theoretic perspective, coaxioms allow one to choose as interpretation a fixed point that is not necessarily  either the least or the greatest one.
For instance{,} in the last three examples{,} the intended semantics is always a fixed point that lies between the least, that is undefined on infinite lists, and the greatest one, that is undetermined on them.
In addition, we will also show that inductive and coinductive interpretations are particular cases of our extension, {thus} proving that it is actually a generalization.
Another important feature is that in this framework we can interpret also inference systems where judgements that should be defined inductively and coinductively are mixed together in the same definition.

The notion of coaxiom has been inspired by  some of the operational models mentioned above~\cite{AnconaZucca12, AnconaZucca13, Ancona13} and, in our intention, this generalization of inference systems will serve as an abstract framework for a better understanding of these operational models, allowing formal reasoning on them.

The rest of the paper is organized as follows.
In \refToSection{coaxioms} we will recall some basic concepts regarding inference systems and we will introduce  \emph{inference systems with coaxioms}, informally explaining their semantics with a bunch of examples.
The fixed point semantics for inference systems with coaxioms is formally defined in \refToSection{coaxioms-model}.
{There} we will present closure and kernel systems, which are well-known notions on the power-set, in the more general setting of complete lattices, getting the definition of \emph{bounded fixed point}, that will represent the semantics induced by coaxioms for an inference system.
In \refToSection{coaxioms-trees} we will first formalize the notion of proof tree, which is the object representing a derivation,  then we will  introduce several equivalent semantics based {on proof trees, that is, proof-theoretic semantics.}
Particularly interesting are the two characterizations exploiting the new concept of \emph{approximated proof tree}, that will allow us to provide the semantics in terms of sequences of well-founded trees, without considering non-well-founded derivations.
Proof techniques for coaxioms to prove both completeness and soundness of definitions will be  discussed in \refToSection{coaxioms-reasoning}.
In particular{,} we will introduce the \emph{bounded coinduction principle}{,} that is a generalization of the standard coinduction principle, aimed to show the completeness of a definition expressed in terms of an inference system with coaxioms.
In \refToSection{coaxioms-examples}, we will illustrate weaknesses and strengths of our framework, using various, more involved, examples.
A straightforward and further extension of the framework is presented in \refToSection{corules}, where we {introduce}  \emph{corules}.
Finally, in \refToSection{related} related work is summarized and \refToSection{conclu} conclude{s} the work.

This paper is extracted from {my} master thesis~\cite{Dagnino17}{, and}  presents in more detail the work we have done in~\cite{AnconaDZ17esop}.
Notably, here we discuss closures and kernels from a more general point of view (see \refToSection{ck}), in order to better frame the bounded fixed point in lattice theory.
Furthermore, thanks to a more formal treatment of proof trees, we introduce an additional proof-theoretic characterization, using approximated proof trees (see \refToTheorem{approx-sequence}).
We also present another example of application of coaxioms to graphs (see \refToSection{graph-example}).
With respect to~\cite{Dagnino17}, here we briefly present a straightforward further extension of the framework, considering also \emph{corules} (see \refToSection{corules}).
We only briefly mention corules, because there are still few and quite involved examples where they seem to be really needed (one can be found in~\cite{AnconaDZ18}), hence restricting ourselves to coaxioms allows us to provide a clearer presentation.


\section{Inference systems with coaxioms}\label{sect:coaxioms}

First of all, we recall some standard notions about inference systems~\cite{Aczel77,Sangiorgi11}.
In the following, assume a set $\universe$, called \emph{universe}, whose elements are named \emph{judgements}.
An \emph{inference system} $\is$ consists of a set of \emph{inference rules}, which are pairs $\frac{\prem}{\cons}$, with $\prem\subseteq\universe$ the set of \emph{premises}, $\cons\in\universe$ the \emph{consequence} (a.k.a. \emph{conclusion}).

The intuitive interpretation of a rule is that if the premises in $\prem$ hold then the consequence $\cons$ should hold as well.
In particular, an \emph{axiom} is (the consequence of) a rule with empty set of premises, which necessarily holds.

{A \emph{proof tree}\footnote{See \refToSection{coaxioms-trees} for a more formal account of proof trees.} is a tree whose nodes are (labeled by) judgements in $\universe$, and there is a node $\cons$ with set of children $\prem$ only if there exists a rule $\frac{\prem}{\cons} \in \is$.
We say that a judgement $\judg \in \universe$ \emph{has a proof tree} if it is the root of some proof tree.
The \emph{inductive interpretation} of $\is$ is the set of judgements having a well-founded proof tree, while the \emph{coinductive interpretation} of $\is$ is the set of judgements having an arbitrary (well-founded or not) proof tree.
It can be \EZ{shown} that these definitions are equivalent to standard fixed point semantics, \EZ{which} we will discuss later. }

\subsection{A gentle introduction}
We introduce now our generalization of inference systems.
An \emph{inference system with coaxioms} is a pair $\Pair{\is}{\coaxioms}$ where $\is$ is an inference system and $\coaxioms \subseteq \universe$ is a set of \emph{coaxioms}{.}
A coaxiom $\cons \in \coaxioms$ will be written as $\CoAxiom{\cons}$, very much like an axiom, and, analogously to an axiom, it can be {used} as initial assumption to derive other judgements.
However, coaxioms will be used in a special way, that is, intuitively they can be used only ``at infinite depth'' in a derivation.
This will allow us to impose an initial assumption also to infinite proof trees, that otherwise  are not required to have {a base case}.
We will make precise this notion in next sections, now we will present some examples to illustrate  how to use coaxioms  to govern the semantics of an inference system.

As we are used to doing for rules,  we will express  sets of coaxioms by means of  \emph{meta-coaxioms} with side conditions.

Let us start with an introductory example concerning graphs, that are a widely used non-well-founded data type.
Consider a graph $\Pair{\Nodes}{\adj}$ where $\Nodes$ is the set of nodes and $\fun{\adj}{\Nodes}{\wp(\Nodes)}$ is the adjacency function, that is, for each node $\node \in \Nodes$, $\adj(\node)$ is the set of nodes adjacent to $\node$.
We want to define the judgement $\Visit{\node}{\nodeset}$ stating that {$\nodeset$ is the set of nodes reachable  from $\node$.}

We define this judgement with the following (meta-)rule and (meta-)coaxiom: 
\begin{mathpar}
\Rule{
	\Visit{\node_1}{\nodeset_1} \Space \ldots \Space \Visit{\node_k}{\nodeset_k}
}{ \Visit{\node}{\{\node\} \cup \nodeset_1 \cup \ldots \cup \nodeset_k} }
\adj(\node) = \{\node_1, \ldots, \node_k\}
\and
\CoAxiom{\Visit{\node}{\emptyset}} \node \in \Nodes
\end{mathpar}
{To be more concrete, we consider the graph drawn in \refToFigure{concrete-graph}, whose corresponding rules are reported in the same figure. }
\begin{figure}
\begin{center}
\begin{tikzcd}
a \ar[r, bend left] & b \ar[l, bend left] & c
\end{tikzcd}
\end{center}
\vspace{1ex}
\hrule
\vspace{1ex}
\begin{mathpar}
\Rule{\Visit{b}{\nodeset}}{\Visit{a}{\{a\}\cup\nodeset}}
\and
\Rule{\Visit{a}{\nodeset}}{\Visit{b}{\{b\}\cup\nodeset}}
\and
\Rule{}{\Visit{c}{\{c\}}}
\and
\CoAxiom{\Visit{a}{\emptyset}}
\and
\CoAxiom{\Visit{b}{\emptyset}}
\and
\CoAxiom{\Visit{c}{\emptyset}}
\end{mathpar}
\caption{Concrete graph example}\label{fig:concrete-graph}
\end{figure}

Let us ignore for a moment coaxioms and reason about the standard interpretations.
It is clear that, if we interpret the system inductively, we will only prove the judgement $\Visit{c}{\{c\}}$, because it is the only axiom and other rules do not depend on it.
In other words, the judgement $\Visit{\node}{\nodeset}$, like other judgements on graphs,  cannot be defined inductively by structural recursion, since the structure is not well-founded.
In particular, the problem are cycles, where the proof may be trapped, continuously unfolding the structure of the graph without ever reaching a base case.
{Usual implementations of visits on graphs rely on imperative features and correct this issue by marking already visited nodes.
In this way, they avoid visiting twice the same node, actually breaking cycles. }

On the other hand, if we interpret the meta-rules coinductively (excluding again the coaxioms), then we get the correct judgements $\Visit{a}{\{a,b\}}$ and $\Visit{b}{\{a,b\}}$, but we also get the wrong judgements $\Visit{a}{\{a,b,c\}}$ and $\Visit{b}{\{a,b,c\}}$, as shown by the following derivations
\begin{small}
\begin{mathpar}
\Rule{
	\Rule{
		\Rule{
			\vdots
		}{ \Visit{a}{\{a, b\}} }
	}{ \Visit{b}{\{a, b\}} }
}{ \Visit{a}{\{a, b\}} }
\and
\Rule{
	\Rule{
		\Rule{
			\vdots
		}{ \Visit{b}{\{a, b\}} }
	}{ \Visit{a}{\{a, b\}} }
}{ \Visit{b}{\{a, b\}} }
\and
\Rule{
	\Rule{
		\Rule{
			\vdots
		}{ \Visit{a}{\{a, b, c\}} }
	}{ \Visit{b}{\{a, b, c\}} }
}{ \Visit{a}{\{a, b, c\}} }
\and
\Rule{
	\Rule{
		\Rule{
			\vdots
		}{ \Visit{b}{\{a, b, c\}} }
	}{ \Visit{a}{\{a, b, c\}} }
}{ \Visit{b}{\{a, b, c\}} }
\end{mathpar}
\end{small}

{As already said, coaxioms allow us to impose additional conditions on proof trees to be considered valid:
the semantics of an inference system with coaxioms $\Pair{\is}{\coaxioms}$ is the set of the judgements having
\begin{enumerate}
\item an arbitrary (well-founded or not) proof tree $t$ in $\is$ such that
\item each node of $t$ has a well-founded proof tree in $\Extended{\is}{\coaxioms}$
\end{enumerate}
where $\Extended{\is}{\coaxioms}$ is the inference system obtained enriching $\is$ by judgements in $\coaxioms$ considered as axioms.
Hence, we can use coinduction thanks to 1, but we use coaxioms to restrict its usage, by \emph{filtering out} undesired proof trees. }

Note that for nodes in $t$, which are roots of a well-founded subtree, {the second condition always holds} (a well-founded proof tree in $\is$ is a well-founded proof tree in $\Extended{\is}{\coaxioms}$ as well), hence it is only significant for nodes which are roots of an infinite path in the proof tree.

For instance, in the example {in \refToFigure{concrete-graph}},  the judgement $\Visit{a}{\{a,b\}}$ has an infinite proof tree in $\is$ where each node has a finite proof tree in $\Extended{\is}{\coaxioms}$, as shown below.
\begin{small}
\[
\begin{array}{c|c} 
\multicolumn{1}{l|}{\is}  & \multicolumn{1}{l}{\Space \Extended{\is}{\coaxioms}} \\
\Rule{
	\Rule{
		\Rule{
			\vdots
		}{ \Visit{a}{\{a, b\}} }
	}{ \Visit{b}{\{a, b\}} }
}{ \Visit{a}{\{a, b\}} }
\Space & \Space
\Rule{
	\Rule{
		\Rule{ }{ \Visit{a}{\emptyset} }
	}{ \Visit{b}{\{b\}} }
}{ \Visit{a}{\{a, b\}} }
\BigSpace
\Rule{
	\Rule{
		\Rule{ }{ \Visit{b}{\emptyset} }
	}{ \Visit{a}{\{a\}} }
}{ \Visit{b}{\{a, b\}} }
\end{array}
\]
\end{small}
{On the other hand, the judgement $\Visit{a}{\{a, b, c\}}$ has no finite proof tree in $\Extended{\is}{\coaxioms}$, because $c$ is not reachable from $a$; hence such judgement is  not derivable in $\Pair{\is}{\coaxioms}$, as expected. }

\Rev{We mentioned before an alternative view of the condition imposed by coaxioms on infinite proof trees: in an infinite derivation coaxioms can only be used ``at infinite depth''.
The formal counterpart of this sentence is that the infinite proof tree can be approximated, for all $n\geq 0$, by a well-founded proof tree in $\Extended{\is}{\coaxioms}$ where coaxioms can only be used at depth greater than $n$.
Hence, in a sense, the infinite proof tree is obtained by ``pushing'' coaxioms to infinity.

For instance, the infinite proof tree in $\is$ for the judgement $\Visit{a}{\{a,b\}}$ shown above can be approximated, for any $n\ge 0$, by a finite proof tree in $\Extended{\is}{\coaxioms}$ where coaxioms are used at depth greater than $n$, as shown below.
\begin{small}
\begin{mathpar}
\Rule{
	\Rule{
		\Rule{ }{ \Visit{a}{\emptyset} }
	}{ \Visit{b}{\{b\}} }
}{ \Visit{a}{\{a, b\}} }
\and
\Rule{
	\Rule{
		\Rule{
			\Rule{ }{ \Visit{b}{\emptyset} }
		}{ \Visit{a}{\{a\}} }
	}{ \Visit{b}{\{a, b\}} }
}{ \Visit{a}{\{a, b\}} }
\and
\Rule{
	\Rule{
		\Rule{
			\Rule{
				\Rule{ }{ \Visit{a}{\emptyset} }
			}{ \Visit{b}{\{b\}} }
		}{ \Visit{a}{\{a, b\}} }
	}{ \Visit{b}{\{a, b\}} }
}{ \Visit{a}{\{a, b\}} }
\and
\cdots
\end{mathpar}
\end{small}
This does not hold, instead, for $\Visit{a}{\{a, b, c\}}$, since it has no finite derivation using coaxioms.
}

As a second example, we consider the definition of the \textit{first} sets in a grammar. Let us represent a context-free grammar by its set of terminals $T$, its set of non-terminals $N$, and all the productions $\produzioneinline{\nonterminal}{\beta_1\mid\ldots\mid\beta_n}$ with left-hand side $\nonterminal$, for each non-terminal $\nonterminal$.
Recall that, for each $\alpha\in {(T\cup N)}^{+}$, we can define the set
$\first{\alpha} = \{ \simb \mid \simb\in T, \alpha{\rightarrow^\star}\simb\beta\}$.
Informally, $\first{\alpha}$ is the set of the initial terminal symbols of the strings which can be derived from a string $\alpha$ in $0$ or more steps.

We defines the judgement $\First{\alpha}{\firstset}$ by the following inference system with coaxioms, where $\firstset\subseteq T$.
\begin{small}
\begin{mathpar}
\Rule{}{ \First{\simb\alpha}{\{\sigma\}} }\sigma\in T
\and
\Rule{
	\First{\nonterminal}{\firstset}
}{ \First{\nonterminal\alpha}{\firstset} }
\begin{array}{l}
\nonterminal\in N\\
\nonterminal{\not\rightarrow^\star}\epsilon
\end{array}
\and
\Rule{
	\First{\nonterminal}{\firstset}
	\\
	\First{\alpha}{\firstset'}
}{ \First{\nonterminal\alpha}{\firstset\cup\firstset'} }
\begin{array}{l}
\nonterminal\in N\\
\nonterminal{\rightarrow^\star}\epsilon
\end{array}
\\
\Rule{}{ \First{\epsilon}{\emptyset} }
\and
\Rule{
	\First{\beta_1}{\firstset_1}
	\\ \cdots \\ 
	\First{\beta_n}{\firstset_n}
}{ \First{\nonterminal}{\firstset_1\cup\dots\cup\firstset_n} }
\produzioneinline{\nonterminal}{\beta_1\mid\ldots\mid\beta_n}
\and
\CoAxiom{ \First{\nonterminal}{\emptyset} }\nonterminal\in N 
\end{mathpar}
\end{small}

The rules of the inference system correspond to the natural recursive definition of \textit{first}.
Note, in particular, that in a string of shape $\nonterminal\alpha$, if the non-terminal $\nonterminal$ is \emph{nullable}, that is, we can derive from it the empty string, then the \textit{first} set for $\nonterminal\alpha$ should also include the \textit{first} set for $\alpha$.

As in the previous example on graphs, the problem with this recursive definition is that, since the non-terminals in a grammar can mutually refer to each other, the function defined by the inductive interpretation can be undefined, since it may never reach a base case.
That is, a naive top-down implementation might not terminate.
For this reason, \textit{first} sets are typically computed by an imperative bottom-up algorithm, or the top-down implementation is corrected by marking already encountered non-terminals, analogously to what is done for visiting graphs.
Again as in the previous example, the coinductive interpretation may fail to be a function, whereas, with the coaxioms, we get the expected result.

Let us now consider some examples of judgements concerning lists.
We consider arbitrary (finite or infinite) lists of integers {and} denote {by} $\LInfSet$ the set of such lists.
We first consider the judgement $\maxElem{l}{x}$, with $l \in \LInfSet$ and $x \in \Z$,  stating that $x$ is the maximum element that {occurs in}  $l$.
This judgement has a natural definition by structural recursion we have discussed in \refToSection{intro} where we have shown that neither inductive nor coinductive interpretations are able to capture the expected semantics.
Therefore, in the following definition we have {added a coaxiom} to the inference system from \refToSection{intro} in order to restrict the coinductive interpretation{.}
\begin{mathpar}
\Rule{}{ \maxElem{\Cons{x}{\EList}}{x} }
\and
\Rule{ \maxElem{l}{y} }{ \maxElem{\Cons{x}{l}}{z} } z = \max\{x, y\}
\and
\CoAxiom{ \maxElem{\Cons{x}{l}}{x} }
\end{mathpar}
Recall that the problem with the coinductive interpretation is that it accepts all judgements $\maxElem{l}{x}$ where $x$ is an upper bound of $l$, even if it does not {occur} in $l$.
The coaxiom, thanks to the way it is used, imposes that $\maxElem{l}{x}$ may hold only if $x$ appears somewhere in the list, hence undesired proofs are filtered out.

A similar example is given by the judgement $\elems{l}{\xs}$ where $l \in \LInfSet$ and $\xs \subseteq \Z$, stating that $\xs$ is the carrier of the list $l$, that is, the set of all elements appearing in $l$.
This judgement can be defined using coaxioms as follows{:}
\begin{mathpar}
\Rule{}{ \elems{\EList}{\emptyset} }
\and
\Rule{ \elems{l}{\xs} }{ \elems{\Cons{x}{l}}{\{x\} \cup \xs} }
\and
\CoAxiom{ \elems{l}{\emptyset} }
\end{mathpar}
If we ignore the coaxiom {and} interpret the system coinductively{, then} we can prove $\elems{l}{\xs}$ for any {superset} $\xs$ of the carrier of $l$ if $l$ is infinite.
The coaxioms again {allow} us to filter out undesired derivations.
For instance, for $l$ the infinite list of {1s}, any judgement $\elems{l}{\xs}$ with $1 \in\xs$ can be derived.
Indeed, for any such judgement we can construct an infinite proof tree which is a chain of applications of the last {meta-rule}.
With the coaxioms, we only consider the infinite trees where the node $\elems{l}{\xs}$ has a finite proof tree in the inference system enriched by the coaxioms.
This is only true for $\xs=\{1\}$.

Using coaxioms, we can get the right semantics also for other examples on lists discussed in \refToSection{intro}, in particular definitions for predicates $\member{x}{l}{b}$ and $\allPos{l}{b}$ are reported below.
\begin{small}
\begin{mathpar}
\Rule{}{ \member{x}{\EList}{\False} }
\and
\Rule{}{ \member{x}{\Cons{x}{l}}{\True} }
\and
\Rule{ \member{x}{l}{b} }{ \member{x}{\Cons{y}{l}}{b} } x \ne y
\and
\CoAxiom{ \member{x}{l}{\False} }
\\
\Rule{}{ \allPos{\EList}{\True} }
\and
\Rule{}{ \allPos{\Cons{x}{l}}{\False} } x\le 0
\and
\Rule{ \allPos{l}{b} }{ \allPos{\Cons{x}{l}}{b} }x>0
\and
\CoAxiom{ \allPos{l}{\True} }
\end{mathpar}
\end{small}
In \refToSection{intro} we said that the standard coinductive interpretation allows us to prove too many judgements.
For instance, if $l$ is the infinite list of 1s, hence $l = \Cons{1}{l}$\footnote{It is well-known that an infinite term can be represented by a set of recursive equations, see, e.g.,~\cite{AdamekMV06b}.}, the following are valid infinite derivations, obtained repeatedly applying the only rule with non-empty premises
\begin{small}
\begin{mathpar}
\Rule{
	\Rule{
		\vdots
	}{ \member{2}{l}{\True} }
}{\member{2}{l}{\True} }
\and
\Rule{
	\Rule{
		\vdots
	}{ \member{2}{l}{\False} }
}{\member{2}{l}{\False} }
\and
\Rule{
	\Rule{
		\vdots
	}{ \allPos{l}{\True} }
}{ \allPos{l}{\True} }
\and
\Rule{
	\Rule{
		\vdots
	}{ \allPos{l}{\False} }
}{ \allPos{l}{\False} }
\end{mathpar}
\end{small}
In the semantics induced by coaxioms, only the second and the third proof trees are valid, since their nodes  are derivable {by a finite proof tree using also coaxioms}, while this fact is not true for the other derivations.

\subsection{Semantics}\label{sect:sem-coaxioms}

We now define the model-theoretic semantics for inference systems with coaxioms.
First of all we recall some notions for standard inference systems.
Consider an inference system $\is$;
the \emph{(one step) inference operator} $\fun{\Op{\is}}{\wp(\universe)}{\wp(\universe)}$ associated with $\is$  is defined by
\[
\Op{\is}(S)= \left\{\cons \in \universe  \mid \prem\subseteq S, \Rule{\prem}{\cons}\in\is \right\}
\]
That is, $\Op{\is}(S)$ is the set of judgements that can be inferred (in one step) from the judgements in $S$ using the inference rules in $\is$.
Note that this set always includes axioms.

A set $S$ is \emph{closed} if $\Op{\is}(S)\subseteq S$, and \emph{consistent} if $S\subseteq\Op{\is}(S)$.
That is, no new judgements can be inferred from a closed set, and all judgements in a consistent set can be inferred from the set itself.

The \emph{inductive interpretation} of $\is$, denoted $\Ind{\is}$, is the smallest closed set, that is, the intersection of all closed sets, and the \emph{coinductive interpretation} of $\is$, denoted $\CoInd{\is}$, is the largest consistent set, that is, the union of all consistent sets.
 Both interpretations are well-defined and can be equivalently expressed as the least (respectively, the greatest) fixed point of the inference operator.
These definitions can be \EZ{shown} to be equivalent to the proof-theoretic characterizations introduced \EZ{before}, see~\cite{LeroyGrall09,Dagnino17}.

For particular inference systems, we can also compute $\Ind{\is}$ and $\CoInd{\is}$ iteratively, see e.g.~\cite{Sangiorgi11}.
More precisely, if all rules in $\is$ have a finite set of premises, then
$\Ind{\is}    =\bigcup \{ \IterOp{\is}{n}(\emptyset) \mid n \in \N \}$,
and, if for each judgement $\cons$ there is a finite set of rules having $\cons$ as conclusion, then
$\CoInd{\is}=\bigcap \{ \IterOp{\is}{n}(\universe) \mid n \in  \N\}$.
This happens because, under the former condition, $\Op{\is}$  is upward continuous, and, under the latter condition, $\Op{\is}$ is downward continuous (see page~\pageref{page:continuity} for a formal definition of upward/downward continuity).

Given an inference systems with coaxioms $\Pair{\is}{\coaxioms}$, we can construct the \emph{interpretation generated by coaxioms}, denoted by $\Generated{\is}{\coaxioms}$, by the following two steps:
\begin{enumerate}
\item First,  we consider the inference system ${\Extended{\is}{\coaxioms}}$ obtained enriching $\is$ by judgements in $\coaxioms$ considered as axioms, and we take its inductive interpretation $\Ind{{\Extended{\is}{\coaxioms}}}$.
\item Then, we take the coinductive interpretation of the inference system obtained from $\is$ by {keeping only rules with consequence in} $\Ind{{\Extended{\is}{\coaxioms}}}$, that is, we define
\[
\Generated{\is}{\coaxioms}=\CoInd{\Restricted{\is}{\Ind{{\Extended{\is}{\coaxioms}}}}}
\]
\end{enumerate}
where $\Restricted{\is}{S}$, with $\is$ inference system and $S\subseteq\universe$, denotes the inference system obtained from $\is$ by keeping only rules with
consequence in $S$, that is, $\Restricted{\is}{S} = \{\frac{\prem}{\cons} \in \is \mid \cons \in S\}$.

If we consider again the example of the graph in \refToFigure{concrete-graph}, since the universe is finite, every monotone function is continuous, hence we can compute fixed points iteratively.
Therefore, in the first phase, we obtain the following judgements (the number at the beginning of each line indicates  the iteration step):
\begin{small}
\begin{enumerate}[(1),leftmargin=1cm]
    \item $\Visit{a}{\emptyset}$, $\Visit{b}{\emptyset}$, $\Visit{c}{\emptyset}$, $\Visit{c}{\{c\}}$
    \item $\Visit{a}{\emptyset}$, $\Visit{b}{\emptyset}$, $\Visit{c}{\emptyset}$, $\Visit{c}{\{c\}}$, $\Visit{a}{\{a\}}$, $\Visit{b}{\{b\}}$
    \item $\Visit{a}{\emptyset}$, $\Visit{b}{\emptyset}$, $\Visit{c}{\emptyset}$, $\Visit{c}{\{c\}}$, $\Visit{a}{\{a\}}$, $\Visit{b}{\{b\}}$, $\Visit{a}{\{a,b\}},\Visit{b}{\{a,b\}}$
\end{enumerate}
\end{small}
The last set is closed, hence it is $\Ind{{\Extended{\is}{\coaxioms}}}$.

For the second phase, first of all we have to construct the inference system $\Restricted{\is}{\Ind{\Extended{\is}{\coaxioms}}}$, whose rules are those of $\is$ (in \refToFigure{concrete-graph})  with conclusion in \EZ{$\Ind{\Extended{\is}{\coaxioms}}$}, computed above.
Hence, they are the following:
\begin{small}
\begin{mathpar}
\Rule{\Visit{b}{\nodeset}}{\Visit{a}{\{a\} \cup \nodeset}}\ c \notin \nodeset
\and
\Rule{\Visit{a}{\nodeset}}{\Visit{b}{\{b\} \cup \nodeset}}\ c \notin \nodeset
\and
\Rule{}{\Visit{c}{\{c\}}}
\end{mathpar}
\end{small}
These rules have to be interpreted coinductively, hence
each iteration of the inference operator removes judgements which cannot be inferred from the set obtained from the previous iteration step, that is, we get:
\begin{small}
\begin{enumerate}[(1),leftmargin=1cm]
    \item $\Visit{c}{\{c\}}$, $\Visit{a}{\{a\}}$, $\Visit{b}{\{b\}}$, $\Visit{a}{\{a,b\}},\Visit{b}{\{a,b\}}$
    \item $\Visit{c}{\{c\}}$, $\Visit{a}{\{a,b\}}$, $\Visit{b}{\{a,b\}}$
\end{enumerate}
\end{small}
This last set is consistent, hence it is $\Generated{\is}{\coaxioms}$, and it is indeed the expected result.

In next sections, we will study properties of $\Generated{\is}{\coaxioms}$ in a more formal way, notably, we will show that it is actually a fixed point of the inference operator $\Op{\is}$ as expected (see \refToSection{coaxioms-model}).
Such a fixed point will be constructed by taking the greatest consistent subset of the smallest closed superset of the set of coaxioms.
Then, we will also prove that such semantics corresponds to the proof-theoretic characterization informally introduced at the beginning of  this section (see \refToSection{coaxioms-trees}).



\section{Fixed point semantics for coaxioms}\label{sect:coaxioms-model}

\Rev{As mentioned in \refToSection{sem-coaxioms}, we can associate with an inference system $\is$ a monotone function $\Op{\is}$ on the power-set lattice. We always require the semantics of $\is$ to be a fixed point of $\Op{\is}$,  hence we aim to show that this property indeed holds for $\Generated{\is}{\coaxioms}$.

In this section, we will develop the theory needed for this result and some important consequences.
In order to construct the fixed point we need, we work in the general framework of \emph{lattice theory}~\cite{Nation98,DaveyPriestley02}, so that we can highlight only the essential structure.
More precisely, in \refToSection{ck} we discuss closure and kernel operators,
presenting almost standard results, for which, however, we have not found a complete enough discussion in literature~\cite{Nation98,DaveyPriestley02}.
Then, in \refToSection{bfp} and \refToSection{correspondence},
we define the \emph{bounded fixed point}, showing it corresponds to the interpretation generated by coaxioms and it subsumes both inductive and coinductive interpretations.   }

Let us start by recalling some basic definitions about lattices.
{A \emph{complete lattice} is a partially ordered set $\Pair{\lattice}{\order}$ where all subsets $A \subseteq \lattice$ have a least upper bound (a.k.a. \emph{join}), denoted  by $\lub A$.
In particular, in $\lattice$  there are both a top element $\top = \lub \lattice$ and a bottom element $\bot = \lub \emptyset$.
Furthermore, it can be proved that in $\lattice$  all subsets $A \subseteq \lattice$ have also a greatest lower bound (a.k.a. \emph{meet}), denoted by $\glb A$.
In the following, for all $x, y \in \lattice$,  we will write $x \join y$ for the binary join and $x \meet y$ for the binary meet, that is, respectively, $\lub \{x, y\}$ and $\glb\{x, y\}$, respectively. }

Given a function $\fun{\function}{\lattice}{\lattice}$ and an element $x \in \lattice$, we say that
\begin{itemize}
\item $x$ is a \emph{pre-fixed point} if $\function(x) \order x$
\item $x$ is a \emph{post-fixed point} if $x \order \function(x)$
\item $x$ is a \emph{fixed point} if $x = \function(x)$
\end{itemize}
We will denote by $\Pre{\function}$, $\Post{\function}$ and $\Fix{\function}${,} respectively{,} the sets of pre-fixed, post-fixed and fixed points of $\function$.

We also say that $\function$ is \emph{monotone} if{,} for all $x, y \in \lattice$, if $x \order y$ then $\function(x) \order \function(y)$.
Monotone functions over a complete lattice are particularly interesting since, thanks to the Knaster-Tarski theorem~\cite{Tarski55, LassezKS82}, we know that they have both the least and the greatest fixed point, {that} we denote by $\lfp\function$ and $\gfp\function${,} respectively.

In the following we will assume a complete lattice $\Pair{\lattice}{\order}$ and a monotone function $\fun{\function}{\lattice}{\lattice}$.

\subsection{Closures and kernels}\label{sect:ck}
We start by  introducing some notions which are slight generalizations of concepts that can be found in~\cite{AbramskiJung94, Nation98}.

\begin{defi}\label{def:ck-system}
Let $\Pair{\lattice}{\order}$  be a complete lattice. Then
\begin{enumerate}
\item a subset $\CSys \subseteq \lattice$ is a \emph{closure system} if{,} for any subset $X \subseteq \CSys$, $\glb X \in \CSys$
\item a subset $\KSys\subseteq \lattice$ is a \emph{kernel system} if{,} for any subset $X \subseteq \KSys$, $\lub X \in \KSys$
\end{enumerate}
\end{defi}

\noindent
Note that{,} with the usual convention that $\lub \emptyset = \bot$ and $\glb \emptyset = \top$, we have that{,} for all closure systems $\CSys \subseteq \lattice$, $\top \in \CSys$, and{,} for {all} kernel system{s} $\KSys \subseteq \lattice$, $\bot \in \KSys$.

This definition provides a general order-theoretic account of a kind of structures that are very frequent in mathematics, in particular considering the complete lattice carried by the power-set.
For instance, given a group $G$, the set $\Sub{G} \subseteq \wp(G)$ of all subgroups of $G$ is closed under arbitrary intersections, that is, under the meet operation in the  power-set lattice $\Pair{\wp(G)}{\subseteq}$.
{Hence, $\Sub{G}$ is a closure system in the complete lattice $\Pair{\wp(G)}{\subseteq}$, according to the above definition. }
It is easy to see that this fact  also holds for any algebraic structure.
Another example comes from topology.
Indeed, given a topological space $\Pair{X}{\tau}$, by definition $\tau \subseteq \wp(X)$ and is closed under arbitrary unions, hence $\tau$ is a kernel system with respect to the complete lattice $\Pair{\wp(X)}{\subseteq}$.
Moreover{,} the set of closed {sets} in the topological space $\Pair{X}{\tau}$, that is, the set $\{X\setminus A \mid A \in \tau\}$, is closed under arbitrary intersections, hence it is a closure system in $\Pair{\wp(X)}{\subseteq}$.
Actually this is a general fact: if $\KSys \subseteq \wp(X)$ is a kernel system, then $\{X\setminus A \mid A \in \KSys\}$ is a closure system.
Also the converse is true.

It is quite easy to check that the following proposition holds

\begin{prop}\label{prop:pre-post-ck}
Let $\Pair{\lattice}{\order}$ be a complete lattice and $\fun{\function}{\lattice}{\lattice}$ a monotone function.
Then
\begin{enumerate}
\item $\Pre{\function}$ is a closure system
\item $\Post{\function}$ is a kernel system
\end{enumerate}
\end{prop}
\begin{proof}
We only prove 1, since 2 can be obtained by duality.\\
Let $A\subseteq \Pre{\function}$ be a set of pre-fixed points of $\function$.
We have that{,} for all $x \in A$, $\glb A \order x$ (by definition of greatest lower bound),
 then $\function(\glb A) \order \function(x) \order x$ (since $\function$  is monotone and $x$ is pre-fixed),
  hence, finally, $\function(\glb A) \order \glb A$ (by definition of greatest lower bound).
\end{proof}

This observation provides us with a canonical way to associate a closure and a kernel system with a monotone function.
Let us introduce another notion.

\begin{defi}\label{def:ck-op}
Let $\Pair{\lattice}{\order}$ be a complete lattice. Then
\begin{enumerate}
\item A monotone function $\fun{\closure}{\lattice}{\lattice}$ is a \emph{closure operator}  if it satisfies the following conditions{:}
\begin{itemize}
\item for all $x \in \lattice$, $x \order \closure(x)$
\item for all $x \in \lattice$, $\closure(\closure(x)) = \closure(x)$
\end{itemize}
\item A monotone function $\fun{\ker}{\lattice}{\lattice}$ is a \emph{kernel operator}  if it satisfies the following conditions{:}
\begin{itemize}
\item for all $x \in \lattice$, $\ker(x) \order x$
\item for all $x \in \lattice$, $\ker(\ker(x)) = \ker(x)$
\end{itemize}
\end{enumerate}
\end{defi}

\noindent
Note that, since a closure operator $\fun{\closure}{\lattice}{\lattice}$  is a monotone function, we can associate with it both a closure and a kernel system, $\Pre{\closure}$ and $\Post{\closure}$.
However, by the first condition of the definition of closure operator, we get that $\Post{\closure} = \lattice$, hence it is  not interesting, and $\Pre{\closure}=\Fix{\closure}$.
Dually, for a kernel operator $\fun{\ker}{\lattice}{\lattice}$, only $\Post{\ker} = \Fix{\ker}$ is interesting, because $\Pre{\ker} = \lattice$.
Therefore, we can say that every closure operator naturally induces a closure system   and every kernel operator naturally induces a kernel system.

The next result shows how we can build, in a canonical way, a closure/kernel operator {from a closure/kernel system}.

\begin{thm}\label{theo:ck-sys-op}
Let $\Pair{\lattice}{\order}$ be a complete lattice.  Then
\begin{enumerate}
\item given a closure system $\CSys \subseteq \lattice$ the function
\[
\closure_\CSys(x) = \glb \{y \in \CSys \mid x \order y\}
\]
is a closure operator such that $\Fix{\closure_\CSys} = \CSys$
\item given a kernel system $\KSys \subseteq \lattice$ the function
\[
\ker_\KSys (x) = \lub \{y \in \KSys \mid y \order x\}
\]
is a kernel operator such that $\Fix{\ker_\KSys} = \KSys$
\end{enumerate}
\end{thm}
\begin{proof}
We prove only point 1, {the} point 2 can be obtained by duality.
We first prove that $\closure_\CSys$ is monotone.
Consider $x, y \in \lattice$ such that $x\order y$, hence $\{z \in {\CSys} \mid y \order z\}\subseteq \{z \in \CSys \mid x \order z\}$, thus  $\closure_\CSys (x) \order \closure_\CSys(y)$, since the greatest lower bound is a monotone operator.\footnote{We are considering the function $A\mapsto \glb A$ from $\wp(\lattice)$ to $\lattice$} \\
The fact that $x \order \closure_\CSys(x)$ for all $x \in \lattice$ follows from the fact that $x$ is a lower bound of the set $\{y \in \CSys \mid x \order y \}$.\\
Finally, note that by definition{,} for all $x \in \lattice$, $\closure_\CSys(x) \in \CSys$, because $\CSys$ is a closure system, hence in order to show that $\closure_\CSys(\closure_\CSys(x)) = \closure_\CSys(x)$ {it} is enough to show that{,} for all $z \in \CSys$, $\closure_\CSys(z) = z$, namely, $\CSys \subseteq \Fix{\closure_\CSys}$.
So consider $z \in \CSys$, we have already shown that $z \order \closure_\CSys(z)$, thus we have only to show  the other inequality.
Since $z\in\CSys$,  $z \in \{y \in \CSys \mid z \order y\}$, and this implies that $\closure_\CSys(z) \order z$.

This shows that $\closure_\CSys$ is a closure operator.
Actually we have also proved that $\CSys \subseteq \Fix{\closure_\CSys}$.
Therefore, to conclude the proof it remains to show that $\Fix{\closure_\CSys} \subseteq  \CSys$, but this is trivial, since if $z = \closure_\CSys(z)$, then $z \in \CSys$ by definition.
\end{proof}

The above theorem, considered for instance for closure systems, states that each closure system induces a closure operator having as (pre-)fixed points exactly the elements in the closure system. 
Actually, we can say even more: each closure system induces a unique closure operator, that is, each closure operator is uniquely determined by its (pre-)fixed points. 

\begin{thm}\label{theo:ck-op-unique}
Let $\Pair{\lattice}{\order}$ be a complete lattice. Then
\begin{itemize}
\item if $\fun{\closure}{\lattice}{\lattice}$ is a closure operator then $\closure_{\Fix{\closure}} = \closure$
\item if $\fun{\ker}{\lattice}{\lattice}$ is a kernel operator, then $\ker_{\Fix{\ker}} = \ker$.
\end{itemize}
\end{thm}
\begin{proof}
We prove only point 1, the point 2 can be obtained by duality. \\
We have to show that $\closure(x) = \closure_{\Fix{\closure}}(x)$ for all $x \in \lattice$.
By definition, $\closure_{\Fix{\closure}} = \glb A$ with $A=\{y \in \Fix{\closure} \mid x\order y\}$, hence, since $x \order \closure(x)$, $\closure(x) \in A$.
In order to conclude the proof we have to show that $\closure(x)$ is the least element of $A$.
To this aim, consider $y = \closure(y) \in A$ and prove that it is above $\closure(x)$.
Note that $x\order y$, hence, by monotonicity of $\closure$, $\closure(x) \order \closure(y) = y$, as needed.
\end{proof}

In other words{,} the above theorem {states} that to define a closure or kernel operator {it} is enough to specify a closure or a kernel system.
{F}or instance, the closure system $\Sub{G}$, where $G$ is a group, induces the closure operator $\fun{\GroupGen{-}}{\wp(G)}{\wp(G)}$, that computes for any set $X \subseteq G$  the subgroup generated by $X$.
For a topological space $\Pair{X}{\tau}$ we have that the topology $\tau$ induces a kernel operator that, for any set $A\subseteq X$, computes {its} interior, and the set of closed sets $\{X\setminus A \mid A \in \tau\}$ induces {the} topological closure operator.

\subsection{The bounded fixed point}\label{sect:bfp}
Let us now consider a monotone function $\fun{\function}{\lattice}{\lattice}$.
As we have seen, we can associate with $\function$ both a closure and a kernel system, $\Pre{\function}$ and $\Post{\function}$ respectively.
Thanks to \refToTheorem{ck-sys-op} and \refToTheorem{ck-op-unique} we know that these systems induce a unique closure and kernel operator respectively, defined below
\begin{align*}
\closure_\function(x) = \closure_{\Pre{\function}} &= \glb \{y \in \Pre{\function} \mid x \order y\}\\
\ker_\function(x) = \ker_{\Post{\function}} 	         &= \lub \{y \in \Post{\function} \mid y \order x \}
\end{align*}
We call $\closure_\function$ the \emph{closure} of $\function$ and $\ker_\function$ the \emph{kernel} of $\function$.
Intuitively, $\closure_\function(x)$ is the best pre-fixed approximation of $x$ (the least pre-fixed point above $x$), while $\ker_\function(x)$ is the best post-fixed approximation of $x$ (the greatest post-fixed point below $x$).
In this part of the section we will study some properties of these operators related to fixed points constructions.

First of all, we note that from the definitions of the closure and the kernel of $\function$ we can immediately derive a generalization of both the induction and the coinduction principles.
Given $\coaxioms, \bound \in \lattice$, for all $x \in \lattice$ we have
\begin{itemize}[align=left]
\item[\IndPrinciple] if $\function(x) \order x$ ($x$ pre-fixed) and $\coaxioms \order x$, then $\closure_\function(\coaxioms) \order x$
\item[\CoIndPrinciple] if $x \order \function(x)$ ($x$ post-fixed) and $x \order \bound$, then $x \order \ker_\function(\bound)$
\end{itemize}
These two principles are a generalization of standard induction and coinduction principles, because we can retrieve them through particular choices for $\coaxioms$ and $\bound$.
Indeed, if $\coaxioms = \bot$, the condition  $\coaxioms \order x$ is trivially always true, and we have $\closure_\function(\bot) = \glb \Pre{\function} = \lfp \function$ by Knaster-Tarski  fixed point theorem~\cite{Tarski55}, hence \IndPrinciple allows us to conclude $\lfp\function \order x$ like standard induction, requiring the same hypothesis.
Dually, if $\bound = \top$, again the condition $x\order \bound$ is trivially always true, and $\ker_\function(\top) = \lub \Post{\function} = \gfp \function$, again by Knaster-Tarski, hence \CoIndPrinciple allows us to conclude $x \order \gfp \function$ like standard coinduction, requiring the same hypothesis.

We now prove a result ensuring us that under suitable hypotheses we can use the closure and the kernel of a monotone function to build {a} fixed point of that function.
We will denote by $\LowSet{x}$ and $\UpSet{x}$ respectively the set of lower bounds of $x$ and the set of upper bounds of $x$.

\begin{prop}\label{prop:ck-fp}
Let $\coaxioms, \bound \in \lattice$.  Then
\begin{enumerate}
\item if $\bound$ is a pre-fixed point of $\function$, then $\ker_\function(\bound)$ is a fixed point
\item if $\coaxioms$ is  post-fixed point of $\function$, then $\closure_\function(\coaxioms)$ is a fixed point
\end{enumerate}
\end{prop}
\begin{proof}
We will prove only point 1, the point 2 can by obtained by duality. \\
Note that $\LowSet{\bound}$ is a complete lattice and   the function $\fun{\function}{\LowSet{\bound}}{\LowSet{\bound}}$ (obtained by restricting $\function$ to $\LowSet{\bound}$)  is well-defined and monotone, since $\bound$ is a pre-fixed point.
Therefore, $\ker_\function(\bound)$ is the join of all post-fixed point{s} of $\function$ in the complete lattice $\LowSet{\bound}$, hence by Knaster-Tarski  it is a fixed point.
\end{proof}

Therefore, we now know that if $\bound$ is pre-fixed $\ker_\function(\bound)$ is the greatest fixed point below $\bound$, and, if $\coaxioms$ is post-fixed, then $\closure_\function(\coaxioms)$ is the least fixed point above $\coaxioms$.

We are now able to define the \emph{bounded fixed point}.

\begin{defi}[Bounded fixed point]\label{def:bfp}
Let $\coaxioms \in \lattice$.
The \emph{bounded fixed point of $\function$ generated by $\coaxioms$}, denoted by $\Generated{\function}{\coaxioms}$,  is the greatest fixed point of $\function$ below the closure of $\coaxioms$, that is, $\Generated{\function}{\coaxioms} = \ker_\function(\closure_\function(\coaxioms))$.
\end{defi}

The bounded fixed point is well-defined since, thanks to \refToProposition{ck-fp}, there exists the greatest fixed point below $\bound$, provided that the bound $\bound$  is a pre-fixed point.
Since in general $\coaxioms$ might not be pre-fixed, we need to construct a pre-fixed point from $\coaxioms$, and this is done by the closure operator $\closure_\function$.
Note that the first step of this construction \emph{cannot} be expressed as the least fixed point of $\function$ on the complete lattice $\UpSet{\coaxioms}$, since in general $\function$ may fail to be well-defined (e.g., if $\function$ is the function which maps any element to $\bot \order  \coaxioms$ with $\coaxioms \ne \bot$).
Indeed, $\closure_\function(\coaxioms)$ is \emph{not} a fixed point in general, but only a pre-fixed point: we need the two steps to obtain a fixed point.

Note also that the definition of bounded fixed point is asymmetric, that is, we take the greatest fixed point bounded from above by a least pre-fixed point, rather than the other way round.
This is motivated by the intuition, explained in \refToSection{coaxioms}, that we essentially need a greatest fixed point, since we want to deal with {non-}well-founded structures, but we want to ``constrain'' in some way such greatest fixed point.
Investigating the dual construction ($\closure_\function(\ker_\function(\coaxioms))$) is a matter of further work.

The following proposition states some immediate properties of the bounded fixed point{.}
\begin{prop}\label{prop:bfp-fun}\hfill
\begin{enumerate}
\item If $z \in \lattice$ is a fixed point of $\function$, then $\Generated{\function}{z}=z${.}
\item For all $\coaxioms_1, \coaxioms_2 \in \lattice$, if $\coaxioms_1\order \coaxioms_2$, then $\Generated{\function}{\coaxioms_1} \order \Generated{\function}{\coaxioms_2}${.}
\end{enumerate}
\end{prop}
\begin{proof}\hfill
\begin{enumerate}
\item If $z$ is a fixed point, {then} it is both pre-fixed and post-fixed, hence $\closure_\function(z)=z$ and $\ker_\function(z)=z$.
Thus, we get that $\Generated{\function}{z} = \ker_\function(\closure_\function(z)) = \ker_\function(z) = z$.
\item The statement can be rephrased saying that the function $\fun{\Generated{\function}{-}}{\lattice}{\lattice}$ is monotone, and this  trivially holds since it is a composition of the monotone function $\closure_\function$ and $\ker_\function$.
\qedhere
\end{enumerate}
\end{proof}

\noindent
Therefore, by \refToProposition{ck-fp} we already know that $\Generated{\function}{\coaxioms}$ is a fixed point for any $\coaxioms \in \lattice$; the first point  of the above proposition says that all fixed point{s} of $\function$ can be generated as bounded fixed points.
In other words, considering $\Generated{\function}{-}$ as a function from $\lattice$ into itself, the first point implies that the range of this function is exactly $\Fix{\function}$.
Moreover, the second point states that $\Generated{\function}{-}$ is a monotone function on $\lattice$.

An important fact is that bounded fixed points are a generalization of both least and  greatest fixed points, since they can be obtained by taking particular generators, as stated in the following proposition.

\begin{prop}\label{prop:bfp-lfp-gfp}\hfill
\begin{enumerate}
\item $\Generated{\function}{\top}$ is the greatest fixed point of $\function$
\item $\Generated{\function}{\bot}$ is the least fixed point of $\function$
\end{enumerate}
\end{prop}
\begin{proof}\hfill
\begin{enumerate}
\item Note that $\closure_\function(\top)=\top$, since the only pre-fixed point above $\top$ is $\top$ itself, hence we get $\Generated{\function}{\top} = \ker_\function(\top) = \lub \Post{\function} = \gfp \function$,
\item As already noted  $\closure_\function(\bot) = \lfp \function$, in particular $\closure_\function(\bot)$ is post-fixed, therefore we get $\Generated{\function}{\bot} = \ker_\function(\closure_\function(\bot)) = \closure_\function(\bot)$, namely it is the least fixed point of $\function${.}
\qedhere
\end{enumerate}
\end{proof}

\noindent
{An alternative proof for the above proposition is possible by exploiting \refToProposition{bfp-fun}.
We preferred to give the above} proof, since  this follows the asymmetry of the definition of the bounded fixed point.

We now present a result that will be particularly useful to develop proof techiques {for} the bounded fixed point (see \refToSection{coaxioms-reasoning}).

We first recall some standard notions.\label{page:continuity}
A \emph{chain} $C$, is a totally ordered sequence ${(x_i)}_{i \in \N}$, we say that $C$ is \emph{ascending} if for all $i \in \N$, $x_i \order x_{i+1}$, and that $C$ is \emph{descending} if for all $i \in \N$, $x_{i+1} \order x_i$.
A function $\fun{\function}{\lattice}{\lattice}$ is said to be \emph{upward continuous} if for any chain $C$, $\function(\lub C) = \lub \function(C)$ and \emph{downward continuous} if for any chain $C$, $\function(\glb C) = \glb \function(C)$.

We will denote by $\Iterate{\function}{x}$ the set $\{ \function^n(x) \mid n \in \N\}$ where $\function^0 = \Id{\lattice}$ and $\function^{n+1} = \function \circ \function^n$.
It is  easy to check that if $x$ is either pre-fixed or post-fixed, $\Iterate{\function}{x}$ is a chain  and in particular {a descending chain if $x$ is pre-fixed.}

\begin{prop}\label{prop:ker-iterate-bounds}
Let $\bound \in \lattice$ be a pre-fixed point of $\function$. Then
\begin{enumerate}
\item for all $n \in \N$, $\ker_\function(\bound) = \ker_\function(\function^n(\bound))$
\item $\ker_\function(\bound) = \ker_\function(\glb \Iterate{\function}{\bound})$
\end{enumerate}
\end{prop}
\begin{proof}
Note that, since $\bound$ is pre-fixed, $\Iterate{\function}{\bound}$ is a descending chain, hence for all $n \in \N$ we have $\function^{n+1} (\bound) \order \function^n(\bound)$, that is, $\function^n(\bound)$ is a pre-fixed point for all $n \in \N$.
\begin{enumerate}
\item We prove the statement by induction on $n$. If $n=0$ there is nothing to prove.
Now, assume the thesis for $n$.
By definition, $\ker_\function(\function^n(\bound))$ is a post-fixed point, hence $\ker_\function(\function^n(\bound)) \order \function(\ker_\function(\function^n(\bound)))$.
Since $\ker_\function$ is a kernel operator, by \refToDefinition{ck-op}, we have $\ker_\function(\function^n(\bound)) \order \function^n(\bound)$, hence by the monotonicity of $\function$, we get $\function(\ker_\function(\function^n(\bound))) \order \function^{n+1}(\bound)$.
Now, by transitivity of $\order$ we get $\ker_\function(\function^n(\bound)) \order \function^{n+1}(\bound)$.
Therefore,  by \CoIndPrinciple we conclude $\ker_\function(\function^n(\bound)) \order \ker_\function(\function^{n+1}(\bound))$.\\
On the other hand, since $\function^n(\bound)$ is pre-fixed, we have $\function^{n+1}(\bound) \order \function^n(\bound)$.
Thus, by the monotonicity of $\ker_\function$ we get the other inequality, and this implies  $\ker_\function(\function^n(\bound)) = \ker_\function(\function^{n+1}(\bound))$.
Finally, thanks to the induction hypothesis we get the thesis.
\item By point 1 we have $\ker_\function(\bound) \order \function^n(\bound)$ for all $n \in \N$, hence $\ker_\function(\bound) \order \glb \Iterate{\function}{\bound}$.
Therefore, by \CoIndPrinciple we get $\ker_\function(\bound) \order \ker_\function(\glb \Iterate{\function}{\bound})$.
On the other hand, we have $\glb \Iterate{\function}{\bound} \order \bound$, hence, by monotonicity of $\ker_\function$, we get the other inequality, and this implies the thesis.
\qedhere
\end{enumerate}
\end{proof}

\noindent
Another way to read the above proposition is that, given a bound $\bound$ which is pre-fixed, we obtain the same greatest fixed point below $\bound$ if we take as bound
any element $\function^n(\bound)$ of the descending chain $\Iterate{\function}{\bound}$.
Moreover, \refToProposition{ker-iterate-bounds} says also that we obtain the same greatest fixed point  induced by $\bound$ if we take as bound  the greatest lower bound of that chain, namely, $\glb \Iterate{\function}{\bound}$.

We conclude this part of the section with a result that characterizes the closure and the kernel of respectively a post-fixed and a pre-fixed point using chains in analogy with the Kleene theorem~\cite{LassezKS82}

\begin{prop}\label{prop:ck-chain}
Let $\bound, \coaxioms \in \lattice$ be a pre-fixed and a post-fixed point respectively. Then
\begin{enumerate}
\item if $\function$ is downward continuous, then $\ker_\function(\bound)  = \glb \Iterate{\function}{\bound}$
\item if $\function$ is upward continuous, then $\closure_\function(\coaxioms) = \lub \Iterate{\function}{\coaxioms}$
\end{enumerate}
\end{prop}
\begin{proof}
We prove only point 1, the point 2 can by obtained by duality.\\
Note that $\LowSet{\bound}$ is a complete lattice with top element $\bound$ and
 the function $\fun{\function}{\LowSet{\bound}}{\LowSet{\bound}}$ (obtained by restricting $\function$ to $\LowSet{\bound}$)  is well-defined and monotone, since $\bound$ is a pre-fixed point.
In this case it is also downward continuous, because  so is $\function$.
Therefore, by \refToProposition{ck-fp}, $\ker_\function(\bound)$ is the greatest fixed point of $\function$ in the complete lattice $\LowSet{\bound}$, hence, since $\function$ is downward continuous, we get the thesis by the Kleene theorem.
\end{proof}

Note that the above proposition requires an additional hypothesis on $\function$, that is required to be continuous, as happens for the Kleene theorem~\cite{LassezKS82}.
Under this assumption the above result immediately applies to the bounded fixed point, providing us with an iterative characterization of it, as the following corollary shows.

\begin{cor}\label{cor:bfp-chain}
Let $\coaxioms \in \lattice$ and set $\bound = \closure_\function(\coaxioms)$.
If $\function$ is downward continuous, then $\Generated{\function}{\coaxioms} = \glb \Iterate{\function}{\bound}$.
\end{cor}
\begin{proof}
By \refToDefinition{bfp} we have $\Generated{\function}{\coaxioms} = \ker_\function(\bound)$.
Since $\function$ is downward continuous, by \refToProposition{ck-chain}  we get the thesis.
\end{proof}

\subsection{Coaxioms as generators}\label{sect:correspondence}
In this part of the section we come back to inference systems and we show that the interpretation generated by coaxioms of an inference system is indeed a fixed point of the inference operator.
In \refToSection{coaxioms} we have described two steps to construct $\Generated{\is}{\coaxioms}$, the interpretation generated by coaxioms $\coaxioms$ of an inference system $\is$:
\begin{enumerate}
\item First, we consider the inference system $\Extended{\is}{\coaxioms}$ obtained enriching $\is$ by judgements in $\coaxioms$ considered as axioms, and we take its inductive interpretation $\Ind{\Extended{\is}{\coaxioms}}$.
\item Then, we take the coinductive interpretation of the inference system obtained from $\is$ by keeping only rules with consequence in $\Ind{\Extended{\is}{\coaxioms}}$, that is, we define
\[
\Generated{\is}{\coaxioms}=\CoInd{\Restricted{\is}{\Ind{\Extended{\is}{\coaxioms}}}}
\]
\end{enumerate}
The definition of {the} bounded fixed point is the formulation of these two steps in the general setting of complete lattices.
Indeed, the inference operator $\Op{\is}$ is a monotone function on the complete lattice $\Pair{\wp(\universe)}{\subseteq}$ obtained by taking set inclusion as order,
and specifying the coaxioms $\coaxioms$ corresponds to fixing an arbitrary element of the lattice as generator.
Then:
\begin{enumerate}
\item First, we construct the closure of $\coaxioms$, that is, the best closed approximation of $\coaxioms$.
This closure plays the role of the bound for the next step.
\item Then we construct the greatest fixed point below such bound.
\end{enumerate}
To show the correspondence in a precise way, we give an alternative and equivalent characterization of both the closure and the kernel of an element in $\lattice$.

\begin{prop}\label{prop:ck-alt}
Let $\coaxioms, \bound \in \lattice$.
\begin{enumerate}
\item Consider the function $\fun{\Extended{\function}{\coaxioms}}{\lattice}{\lattice}$ defined by $\Extended{\function}{\coaxioms}(x) = \function(x) \join \coaxioms$.
Then, $\closure_\function(\coaxioms) = \lfp \Extended{\function}{\coaxioms}${.}
\item Consider the function $\fun{\Restricted{\function}{\bound}}{\lattice}{\lattice}$ defined by $\Restricted{\function}{\bound}(x) = \function(x) \meet \bound$.
Then, $\ker_\function(\bound) = \gfp \Restricted{\function}{\bound}$.
\end{enumerate}
\end{prop}
\begin{proof}
We prove only point 1, point 2 can be obtained by duality.
{First of all note that $\Extended{\function}{\coaxioms}$ is a monotone function.
By definition of $\closure_\function$, we have that $\function(\closure_\function(\coaxioms)) \order \closure_\function(\coaxioms)$ and $\coaxioms \order \closure_\function(\coaxioms)$, hence $\closure_\function(\coaxioms)$ is a pre-fixed point of $\Extended{\function}{\coaxioms}$.
Then, by \IndPrinciple, $\closure_\function(\coaxioms)$ is the least pre-fixed point of $\Extended{\function}{\coaxioms}$, hence, by Knaster-Tarski, $\closure_\function(\coaxioms) = \lfp \Extended{\function}{\coaxioms}$.  }
\end{proof}

By this alternative characterization we can formally state the correspondence with the two steps for defining $\Generated{\is}{\coaxioms}$.

\begin{thm}\label{theo:correspondence}
Let $\is$ be an inference system and $\coaxioms, \bound \in \wp(\universe)$, then the following facts hold:
\begin{enumerate}
\item $\Extended{(\Op{\is})}{\coaxioms} = \Op{(\Extended{\is}{\coaxioms})}$ (so we can safely omit brackets)
\item $\Restricted{(\Op{\is})}{\bound} = \Op{(\Restricted{\is}{\bound})}$ (so we can safely omit brackets)
\item $\closure_\Op{\is}(\coaxioms) = \Ind{\Extended{\is}{\coaxioms}}$
\item $\ker_\Op{\is}(\bound) = \CoInd{\Restricted{\is}{\bound}}$
\end{enumerate}
\end{thm}
\begin{proof}\hfill
\begin{enumerate}
\item We have to show that, for $S \subseteq \universe$, $\Extended{(\Op{\is})}{\coaxioms}(S) = \Op{(\Extended{\is}{\coaxioms})}(S)$.
If $\cons \in \Extended{(\Op{\is})}{\coaxioms}(S)$, then either $\cons \in \coaxioms$ or $\cons \in \Op{\is}(S)$;
in the former case, there exists $\frac{\phantom{xx}}{\cons} \in \Extended{\is}{\coaxioms}$ by definition of $\Extended{\is}{\coaxioms}$,
in the latter, there exists $\frac{\prem}{\cons} \in \is$ such that $\prem \subseteq S$, and this  implies $\frac{\prem}{\cons} \in \Extended{\is}{\coaxioms}$.
Therefore, in both cases $\cons \in \Op{(\Extended{\is}{\coaxioms})}(S)$. \\
Conversely, if $\cons \in \Op{(\Extended{\is}{\coaxioms})}(S)$, then there exists $\frac{\prem}{\cons} \in \Extended{\is}{\coaxioms}$ such that $\prem \subseteq S$.
By definition of $\Extended{\is}{\coaxioms}$, either $\frac{\prem}{\cons} \in \is$ or $\cons \in \coaxioms$ and $\prem = \emptyset$,
therefore, in the former case $\cons \in \Op{\is}(S)$ and in the latter $\cons \in \coaxioms$, thus in both cases $\cons \in \Extended{(\Op{\is})}{\coaxioms}(S)$.

\item We have to show that, for $S \subseteq \universe$, $\Restricted{(\Op{\is})}{\bound}(S) = \Op{(\Restricted{\is}{\bound})}(S)$.
If $\cons \in \Restricted{(\Op{\is})}{\bound}(S)$, then we have $\cons \in \bound$ and $\cons \in \Op{\is}(S)$, hence there is $\frac{\prem}{\cons} \in \is$ such that $\prem \subseteq S$;
therefore, by definition of $\Restricted{\is}{\bound}$, we get $\frac{\prem}{\cons} \in \Restricted{\is}{\bound}$, and this implies that $\cons \in \Op{(\Restricted{\is}{\bound})}(S)$.\\
Conversely, if $\cons \in \Op{(\Restricted{\is}{\bound})}(S)$, then there exists $\frac{\prem}{\cons}  \in \Restricted{\is}{\bound}$ such that $\prem \subseteq S$.
By definition of $\Restricted{\is}{\bound}$, we have that $\frac{\prem}{\cons} \in \is$ and $\cons \in \bound$,
therefore $\cons \in \Op{\is}(S)$ and  $\cons \in \bound$, thus $\cons \in \Restricted{(\Op{\is})}{\bound}(S)$.

\item By \refToProposition{ck-alt} we get that $\closure_\Op{\is}(\coaxioms) = \lfp \Extended{\Op{\is}}{\coaxioms}$, that corresponds to the \emph{inductive interpretation} of $\Extended{\is}{\coaxioms}$, $\Ind{\Extended{\is}{\coaxioms}}$, by point 1 of this theorem.

\item By \refToProposition{ck-alt} we get that $\ker_\Op{\is}(\bound) = \gfp \Restricted{\Op{\is}}{\bound}$, that corresponds to the \emph{coinductive interpretation} of $\Restricted{\is}{\bound}$, $\CoInd{\Restricted{\is}{\bound}}$, by point 2 of this theorem{.}
\qedhere
\end{enumerate}
\end{proof}

\noindent
Thanks to \refToTheorem{correspondence}, we can conclude that, given an inference system with coaxioms $\Pair{\is}{\coaxioms}$:
\[
\Generated{\is}{\coaxioms}=
\CoInd{\Restricted{\is}{\Ind{\Extended{\is}{\coaxioms}}}} =
\ker_\Op{\is} (\closure_\Op{\is}(\coaxioms)) =
\Generated{\Op{\is}}{\coaxioms}
\]
that is, the  interpretation generated by coaxioms $\coaxioms$ of the inference system $\is$ is exactly the bounded fixed point of $\Op{\is}$ generated by $\coaxioms$.

Finally, applying \refToProposition{bfp-lfp-gfp} we get that the inductive and the coinductive interpretations of $\is$ are particular cases of the interpretation generated by coaxioms.
Indeed, we get the inductive interpretation when $\coaxioms = \emptyset$ and we get the coinductive interpretation when $\coaxioms = \universe$, as shown below.
\begin{align*}
\Generated{\is}{\emptyset} &= \Generated{\Op{\is}}{\emptyset} = \lfp \Op{\is} = \Ind{\is} \\
\Generated{\is}{\universe} &= \Generated{\Op{\is}}{\universe} = \gfp \Op{\is} = \CoInd{\is}
\end{align*}


\section{Proof trees for coaxioms}\label{sect:coaxioms-trees}

In this section we formalize several proof-theoretic characterizations {of}  the semantics of inference systems with coaxioms, {and prove} their equivalence with the fixed point semantics presented in \refToSection{coaxioms-model}.
In order to discuss such proof-theoretic semantics in a rigorous way, we need a more explicit and mathematically precise notion of proof tree than the one we introduced in \refToSection{coaxioms};
therefore, we start by fixing some concepts on trees.

\subsection{A digression on graphs and trees}\label{sect:trees-graphs}

Here we report some results about trees and graphs.
{We essentially follow the approach adopted in~\cite{AczelAV01,AczelAMV03,AdamekEtAl15}, with few differences in the definition of trees: for us a tree will be labelled and unordered \Rev{as in~\cite{MoerdijkP00,BergM07}.}
The main theorem of this subsection (\refToTheorem{graph-tree}) is a weaker form of results presented in~\cite{AczelAMV03,AdamekEtAl15}, which, however, require additional conditions\footnote{These conditions are needed since they want a final coalgebra for suitable power-set functors.}  on trees, \EZ{which} we can ignore. }

Along this section we denote by $\String{A}$ the set of finite \emph{strings}  on the \emph{alphabet} $A$, which is an arbitrary set of \emph{symbols}.
We use Greek letters $\alpha, \beta, \ldots$ to range over strings and Roman letters $a, b, \ldots$  to range over symbols in $A$ and we implicitly identify strings of length one and symbols.
Moreover, we denote by juxtaposition \emph{string concatenation}, and by $\Len{\alpha}$ the \emph{length} of the string $\alpha$.
Finally, $\EString$ is the empty string.
{We also extend string concatenation to sets of {strings}, denoting, for $X, Y \subseteq \String{A}$, by $XY$ the set $\{\alpha \beta \in \String{A} \mid \alpha \in X, \beta \in Y\}$; moreover if either $X$ or $Y$ are singleton{s} we will omit curly braces, namely $\alpha Y = \{\alpha\}Y$. }

On the set $\String{A}$ we can define the \emph{prefixing relation} $\prec$ as follows: for any $\alpha, \beta \in \String{A}$, $\alpha\prec\beta$ if and only if there exists $\gamma \in \String{A}$ such that $\alpha\gamma = \beta$.
It can be shown that $\prec$ is a partial order and thus, for any $X\subseteq \String{A}${,} the restriction of $\prec$ to $X$ is well-defined and still a partial order.
We say that a subset $X\subseteq \String{A}$ is \emph{well-founded} with respect to prefixing if any chain $C \subseteq X$ is finite.

A non-empty subset $L \subseteq \String{A}$ is a \emph{tree language} if it satisfies the \emph{prefix property}, that is, if $\alpha a \in L$ then $\alpha \in L$.
In particular, $\EString \in L$ for any tree language $L\subseteq \String{A}$.
Now we are able to define trees following~\cite{Courcelle83}.

\begin{defi}\label{def:tree}
Let $A$ be an alphabet, $L \subseteq \String{A}$ a tree language and $\Labels$ a set.
A \emph{tree} labelled in $\Labels$ is a function $\fun{t}{L}{\Labels}$.
The element $t(\EString)$ is called the \emph{root} of $t$.
\end{defi}

The notion of tree in \refToDefinition{tree} is {essentially the same as standard  ones, see, e.g.,~\cite{Courcelle83,AczelAV01}.
The main difference is that we allow an arbitrary set to be taken as alphabet.
This is important because, as we will see, the branching of the tree is bounded by the cardinality of the alphabet, and,
since we have to use trees in the  context of inference systems, this cardinality cannot be bounded a priori. }

If $\fun{t}{L}{\Labels}$ is a tree, {then,} for any $\alpha \in L$, the \emph{subtree} rooted at $\alpha$ is the function $\fun{t\Restrict{\alpha}}{L\Restrict{\alpha}}{\Labels}$, where $L\Restrict{\alpha}=\{\beta \in \String{A} \mid \alpha\beta \in L\}$ and $t\Restrict{\alpha}(\beta) = t(\alpha\beta)$.
This notion is well-defined since $L\Restrict{\alpha}$ is a tree language, hence $t\Restrict{\alpha}$ is a tree.
Note that $t$ is itself a subtree, rooted at $\EString$.
Subtrees rooted at $\alpha$ with $\Len{\alpha}=1$ are called \emph{direct subtrees} of $t$.
Finally, a tree $t$ is \emph{well-founded} if $\dom(t)$ is well-founded with respect to $\prec$.

{The notion of tree introduced in \refToDefinition{tree} is mathematically precise, but  not very intuitive.
A usual, and perhaps more natural,} way to {introduce} trees {is} as particular graphs.
{Intuitively, using a graph-like terminology, that we will make precise {below}, we can see the elements in the tree language $\dom(t)$ as nodes.
Actually, thanks to the {prefix} property, a node $\alpha \in \dom(t)$  represents also all nodes (its prefixes) we have to traverse to reach $\alpha$ starting from the root $\EString$.
For instance, if $\alpha = abc$, we know that $\EString, a, ab, abc \in dom(t)$, hence they are nodes of $t$ and they form the path from the root to $\alpha$.
Therefore, requiring $t$ to be well-founded is equivalent to require that any sequence of prefixes is finite, hence it is equivalent to require that all paths in $t$ are finite. }

{To formally show that} indeed trees can be {seen} as particular graphs, {we} start {by} giving a definition of graph{.}

\begin{defi}\label{def:graph}
A \emph{graph} is a pair $\Pair{\Nodes}{\adj}$ where $\Nodes$ is the set of \emph{nodes} and $\fun{\adj}{\Nodes}{\wp(\Nodes)}$ is the \emph{adjacency function}.
\end{defi}


With this definition it is easy to assign a graph structure to (the domain of) a tree.
Let $\fun{t}{L}{\Labels}$ be a tree, we can represent it as a graph with set of nodes $L$ and adjacency function $\chl_t(\alpha) = \{\beta \in L \mid \exists a \in A. \beta = \alpha a\}$ returning the \emph{children} of a node $\alpha$. 
In the following we will omit the reference to $t$ when it is clear from the context. 
Thanks to this graph structure we justify terminology {like} node and {adjacent} for trees: a node is a string $\alpha \in \dom(t)$ and\EZ{,} given a node $\alpha$, {the set of its adjacents is $\chl(\alpha)$}.
{Furthermore, the \emph{depth} of a node $\alpha \in \dom(t)$ is its distance, in a graph-theoretic sense, from the root, that is, $\Len{\alpha}$, hence it is alway finite;
the depth of the tree $t$ is the least upper bound of the depth of all its nodes, hence it can be infinite. }

We now analyse the role of the alphabet $A$ in the definition of tree (\refToDefinition{tree}).
First, note that its elements are essentially not relevant.
What actually {matters} {is the cardinality of $A$}, that determines the maximum \emph{branching} of the tree, that is, the {maximum number of children (hence subtrees) for each node $\alpha$.
In other words, we have $|\chl(\alpha)| \le |A|$ for all $\alpha \in L$.
For instance, {we can build essentially the same trees if $A$ is either $\{1, 2, 3\}$ or $\{a, b, c\}$.}
However{,} the fact that they have both cardinality {3 is relevant, since trees built on $A$ have for each node at most 3 children.} }
{Therefore, each cardinality $\lambda$ determines a class of trees, called \emph{$\lambda$-branching trees}, that is, trees built on an alphabet of cardinality $\lambda$.

In the following, we will denote {by} $\Tree[\lambda]{\Labels}$ the set of all $\lambda$-branching trees, where $\lambda$ is a given cardinal.
We will omit $\lambda$ when not relevant.
Since, as we have noticed, the elements of the alphabet are irrelevant, we will say that two $\lambda$-branching trees are equal if, considering them as functions, they are equal up to isomorphism on the alphabet\footnote{Formally, we should define an equivalence relation on trees and then work in the quotient.
Two trees $t, t' \in \Tree[\lambda]{\Labels}$ are equivalent iff $t = t' \circ b$ where $b$ is a bijection between the alphabet of $t$ and \EZ{that of} $t'$.}.
Furthermore, a tree $t \in \Tree[\kappa]{\Labels}$ with $\kappa \le \lambda$ can be regarded as an element of $\Tree[\lambda]{\Labels}$ up to an inclusion of the alphabet in a set of cardinality $\lambda$.
We will leave implicit this inclusion and hence write $\Tree[\kappa]{\Labels} \subseteq \Tree[\lambda]{\Labels}$. }

{We now consider a special class of trees, suitable to model proof trees.
As we will see, proof trees are labelled by judgements, {notably} nodes are (labelled by) consequences of rules and their children correspond to  sets of premises, hence each child has a distinct label.
Trees of this kind can be represented in a more compact way, and enjoy an important property.

Let us introduce these trees formally. }
We say that a tree $\fun{t}{L}{\Labels}$ is \emph{children injective} if, for all $\alpha \in \dom(t)$, the restriction of $t$ to the set $\chl(\alpha)$ is injective;
more explicitly, for all $\alpha \in \dom(t)$, if $\alpha a, \alpha b \in \dom(t)$ and $t(\alpha a) = t(\alpha b)$, then $a = b$.
In other words, all children of a node must have different labels.
Note that all subtrees of a children injective tree are themselves children injective.

{The first property we observe is that a children injective tree is completely  determined by the label of its root and by the set of all paths of labels in it.
Indeed, if $\fun{t}{L}{\Labels}$ is children injective, then we can define the following function:
\[
\fun{f_t}{L}{\String{\Labels}}
\quad
\begin{cases}
f_t(\EString) &= \EString\\
f_t(\alpha a) &= f_t(\alpha) t(\alpha a)
\end{cases}
\]
Intuitively, the function $f$ maps each node $\alpha \in L$ to the string of labels encountered in the path from the root to $\alpha$.
It is easy to see that $f$ is injective and  $f_t(L)$ is a tree language.
Hence, the pair $\Pair{t(\EString)}{f_t(L)}$ conveys a complete description of $t$, that is, starting from it, we can reconstruct $t$, up to a change of the alphabet.
More precisely, we can define a tree $\fun{t_\Labels}{f_t(L)}{\Labels}$ such that $t(\alpha) = t_\Labels(f_t(\alpha))$ as follows:
\[
\begin{cases}
t_\Labels(\EString) &= t(\EString)\\
t_\Labels(\alpha a) &= a
\end{cases}
\]
As a consequence, a children injective tree $t$ can always be equivalently represented as $t_\Labels$.
Finally, note that the subtree of $t$ rooted at $\alpha$, $t\Restrict{\alpha}$, is represented by $\Pair{t(\alpha)}{f(L)\Restrict{f(\alpha)}}$. }

We denote by $\CTree{\Labels}$ the set of children injective trees labelled in $\Labels$.
{From the construction just presented, we also get that the branching of a children injective tree $t$ is bounded by the cardinality $\lambda$ of $\Labels$, hence  we have that $\CTree{\Labels} \subseteq \Tree[\lambda]{\Labels}$.}

The main result of this subsection concerning children injective trees  is \refToTheorem{graph-tree}.
Before {stating} it we need to briefly say {something} about \emph{paths} in a graph.
Let $G=\Pair{\Nodes}{\adj}$ be a graph, a path in $G$ is a {non-empty} string $\node_0 \cdots \node_n \in \String{\Nodes}$ such that{, for all} $i \in \{0, \ldots, n-1\}$, $\node_{i+1} \in \adj(\node_i)$, that is, {for all pairs of subsequent nodes the latter is adjacent to the former}.
We say that $\node_0\cdots \node_n$ is a path \emph{from $\node_0$ to $\node_n$}.
Note that the string $\node_0$ of length 1 is also a path, from $\node_0$ to $\node_0$, that does not traverse any edge.
We denote {by} $\Paths{G}$ the set of paths in $G$.

{Note that $\Paths{G}$ is closed under {non-empty} prefixes, that is, if $\alpha a $ is a path and $\alpha$ is not empty, then $\alpha$ is a path too, and more generally, if $\alpha \beta \in \Paths{G}$ and $\alpha$ and $\beta$ are not empty, then $\alpha, \beta \in \Paths{G}$.
Therefore we can easily lift $\Paths{G}$ to a tree language, by adding to it the empty string.
From these observations immediately follows that{,  for each $\alpha \in \Paths{G}$,} the set $\{\beta \in \String{\Nodes} \mid \alpha \beta \in \Paths{G}\} \subseteq \Paths{G} \cup \{\EString\}$ is a tree language. }

{It is also important to note} that the sets $\Tree[\lambda]{\Labels}$ and $\CTree{\Labels}$  both carry a graph structure with the following adjacency function{:}
\[\dsub(t) = \{ t\Restrict{\alpha} \mid \alpha \in \dom(t),\, \Len{\alpha} = 1\}\]
which returns the direct subtrees of $t$.

{Thanks to this observation, we can now prove the following theorem, that will be essential to give our proof of equivalence between the proof-theoretic and the fixed point semantics for coaxioms (\refToTheorem{ker-trees}).
Intuitively, this result allows us to associate with any node in a graph, in a canonical way,   a tree rooted in it, preserving the graph structure. }

\begin{thm}\label{theo:graph-tree}
Let $G=\Pair{\Nodes}{\adj}$ be a graph, then there exists a {function} $\fun{\Path}{\Nodes}{\CTree{\Nodes}}$ such that the  following diagram commutes:
\begin{center}
\begin{tikzcd}[column sep=huge, row sep=huge]
\Nodes \ar[d, "\adj"] \ar[r, "\Path"] & \CTree{\Nodes} \ar[d, "\dsub"] \\ 
\wp(\Nodes) \ar[r, "\wp(\Path)"] & \wp(\CTree{\Nodes}) 
\end{tikzcd}
\end{center}
\end{thm}

\begin{proof}
The function $\Path$ computes for each node the \emph{path expansion} starting from this node, that is, it maps each node $\node$ to the set of all paths starting with $\node$.
More precisely, the set of paths we compute for each node $\node$ is the following:
\[L_\node = \{\alpha \in \String{\Nodes} \mid \node\alpha\in \Paths{G}\}\]
Hence, using the representation of children injective trees as pairs $\Pair{r}{L}$ where $r$ is a label and $L$ is a tree language, using nodes as alphabet, we have that
\[\Path(\node) = \Pair{\node}{L_\node}\]
Now we have to show that, for each node $\node$, $\wp(\Path)(\adj(\node)) = \dsub(\Path(\node))$, that is, $\Pair{\anode}{L} \in \dsub(\Path(\node))$ if and only if $\anode \in \adj(\node)$ and $\Path(\anode) = \Pair{\anode}{L_\anode} = \Pair{\anode}{L}$.
The implication $\Rightarrow$ holds by construction.
On the other hand,
if $\Pair{\anode}{L} \in \dsub(\Path(\node))$ then $L = \{\alpha \in \String{\Nodes} \mid \anode\alpha \in L_\node\}$, hence $L = L_\anode$,
that is, $\Pair{\anode}{L} = \Path(\anode)$.
Moreover, for all $\alpha \in L$,  $\anode \alpha \in L_\node$ implies that $\node\anode\alpha$ is a path in $G$, hence $\anode \in \adj(\node)$\EZ{,} and this shows the equality.
\end{proof}
In the end, note that, if $t_\node = \Path(\node)$ for each $\node \in \Nodes$, then $\Path$ is the unique map making the diagram commute and such that $t_\node(\EString) = \node$.


\subsection{Combining non-well-founded and well-founded proof trees}\label{sect:trees1}
Before providing the first proof-theoretic characterization, we give a more precise definition of proof tree,
{which is a generalization of the notion of rule graph proposed in~\cite{Brotherston05}. }

\begin{defi}\label{def:proof-tree}
{Let $\is$ be an inference system, a \emph{proof tree} in $\is$ is a children injective tree $\fun{t}{L}{\universe}$ such that,
 for all $\alpha \in L$, there is a rule $\frac{\prem}{\cons} \in \is$ such that $t(\alpha) = \cons$ and $t(\chl(\alpha)) = \prem$. }
\end{defi}
In other words, a proof tree $t$ is a tree labelled in $\universe$ where each node $\alpha \in \dom(t)$ is labelled by  the conclusion of a rule $\frac{\prem}{\cons} \in \is$ and children of $\alpha$ are {bijectively labelled by judgements in} $\prem$.
{Since a proof tree $t$ is children injective by definition, we can also represent it as $\Pair{t(\EString)}{f_t(\dom(t))}$. }

In the following, we will often represent proof trees using stacks of rules, that is, if $\frac{\prem}{\cons} \in \is$ and $\T$ is a set of proof trees in bijection with $\prem$ and  such that for all $t \in \T$, $t(\EString) \in \prem$, we denote by $\frac{\T}{\cons}$ the proof tree $t_\cons$ given by
\[
\dom(t_\cons) = \{\EString\} \cup \bigcup_{t\in \T} t(\EString) f_t(\dom(t))
\quad
\begin{cases}
t_\cons(\EString) &= \cons\\
t_\cons(\alpha \judg) &= \judg
\end{cases}
\]

We say that a tree $t$ is a proof tree for a judgement $\judg \in \universe$ if it is a proof tree rooted in $\judg$.
Finally, note that all subtrees of a proof tree $t$ are proof trees themselves for their roots.
With this terminology we can define our proof-theoretic semantics.

The first {proof-theoretic characterization of the semantics of inference systems with coaxioms} is based on the following theorem which slightly generalizes  the standard result about the correspondence between the fixed point and the proof-theoretic semantics of inference systems in the coinductive case (see~\cite{LeroyGrall09}).
We choose to do the proof from scratch, even if it can be done relying on the standard equivalence (see~\cite{AnconaDZ17esop}), since we have not found in literature a detailed enough proof for the standard equivalence, and, in addition, the proof helps us to understand what happens on the proof-theoretic side, in particular when we prove that a set is consistent.
We begin proving a lemma\EZ{.}

\begin{lem}\label{lem:coind-trees}
Let $\is$ be an inference system and $\Spec$ a consistent subset of $\universe$, then for each $\judg \in \Spec$ there is a proof tree $t$ for $\judg$ such that\EZ{,} for all $\alpha \in \dom(t)$, $t(\alpha) \in \Spec$
\end{lem}
\begin{proof}
By hypothesis $\Spec$ is consistent, so for each judgement $\judg \in \Spec$ we can choose a rule $\frac{\prem_\judg}{\judg} \in \is$ such that $\prem_\judg \subseteq \Spec$.
In other words, we can define the map $\fun{\adj}{\Spec}{\wp(\Spec)}$ given by $\adj(\judg) = \prem_\judg$, that turns $\Spec$ into a graph as in \refToDefinition{graph}.
By \refToTheorem{graph-tree}, there exists a map $\fun{\Path}{\Spec}{\CTree{\Spec}}$ making the following diagram commute.
\begin{center}
\begin{tikzcd}
\Spec\ar[d, "\adj"] \ar[r, "\Path"] & \CTree{\Spec} \ar[d, "\dsub"] \\ 
\wp(\Spec) \ar[r, "\wp(\Path)"] & \wp(\CTree{\Spec}) 
\end{tikzcd}
\end{center}
Therefore{,} for each $\judg \in \Spec$, $\Path(\judg)$ is a tree rooted in $\judg$ and labelled in $\Spec$, that is\EZ{,} for all $\alpha \in \dom(\Path(\judg))$, $\Path(\judg)(\alpha) \in \Spec \subseteq \universe$.
Set $t_\judg = \Path(\judg)$ and note that by construction\EZ{,} for all $\alpha \in \dom(t_\judg)$, we have $\chl(\alpha) = \alpha \adj(t_\judg(\alpha)) = \alpha \prem_{t_\judg(\alpha)}$ and
\[
    \Rule{\prem_{t_\judg(\alpha)}}{t_\judg(\alpha)} \in \is
\]
hence $t_\judg$ is a proof tree in $\is$ as needed.
\end{proof}

This lemma essentially ensures that all judgements in a consistent set $\Spec$ have an arbitrary proof tree whose nodes are all (labelled) in $\Spec$.
The next theorem is a slight generalization \EZ{of} the standard equivalence between proof-theoretic and fixed point semantics.

\begin{thm}\label{theo:ker-trees}
Let $\is$ be an inference system and $\bound \subseteq \universe$ a set of judgements.
Then for all $\judg \in \universe$ the following are equivalent{:}
\begin{enumerate}
\item $\judg \in \ker_\Op{\is}(\bound)$
\item there exists a proof tree $t$ for $\judg$ in $\is$ such that each node of $t$ is (labelled) in $\bound$
\end{enumerate}
\end{thm}
\begin{proof}
We prove separately the two implications\EZ{.}
\begin{description}
\item[$1\Rightarrow 2$] By construction $\ker_\Op{\is}(\bound)$ is a consistent set and $\ker_\Op{\is}(\bound) \subseteq \bound$, hence by \refToLemma{coind-trees} each judgement $\judg \in \ker_\Op{\is}(\bound)$ has a proof tree in $\is$ whose nodes are all (labelled) in $\ker_\Op{\is}(\bound)$, hence they are (labelled) in $\bound$ as needed.
\item[$2\Rightarrow 1$] Let $\Spec \subseteq \universe$ be the set of all judgements having a proof tree in $\is$ whose nodes are all (labelled) in $\bound$.
As a consequence, we immediately have that $\Spec \subseteq \bound$, hence if we show that $\Spec$ is consistent, we get the thesis by \CoIndPrinciple.
Consider $\judg \in \Spec$, hence there is a proof tree $t_\judg$ for $\judg$ whose nodes are all (labelled) in $\bound$.
Note that  each $t \in \dsub(t_\judg)$ is a proof tree like $t$, hence $t(\EString) \in \Spec$ and, since $t_\judg$ is a proof tree,
\[
    \Rule{\{t(\EString) \mid t \in \dsub(t_\judg)\}}{\judg} \in \is
\]
and this shows that $\Spec$ is consistent as needed.
\qedhere
\end{description}
\end{proof}

\noindent
As a particular case we get our first proof-theoretic characterization of $\Generated{\is}{\coaxioms}$.

\begin{cor}\label{cor:proof-trees-1}
Let $\Pair{\is}{\coaxioms}$ be an inference system with coaxioms.
Then the following are equivalent\EZ{:}
\begin{enumerate}
\item $\judg \in \Generated{\is}{\coaxioms}$
\item there exists a proof tree $t$ for $\judg$ in $\is$ such that each node of $t$ has a well-founded proof tree in $\Extended{\is}{\coaxioms}$
\end{enumerate}
\end{cor}
\begin{proof}
We have that $\Generated{\is}{\coaxioms} = \ker_\Op{\is}(\Ind{\Extended{\is}{\coaxioms}})$, hence by \refToTheorem{ker-trees} we get that $\judg \in \Generated{\is}{\coaxioms}$ iff there is a proof tree $t$ for $\judg$ in $\is$ whose nodes are all (labelled) in $\Ind{\Extended{\is}{\coaxioms}}$.
Therefore, all nodes of $t$ have a well-founded proof tree in $\Extended{\is}{\coaxioms}$ by   the standard equivalence for the inductive case (see e.g.,~\cite{LeroyGrall09}).
\end{proof}

\subsection{Approximated proof trees}\label{sect:approx-trees}
For the second proof-theoretic characterization, we need to define \emph{approximated proof trees} in an inference system with coaxioms.

\begin{defi}\label{def:approx-trees}
Let $\Pair{\is}{\coaxioms}$ be an inference system with coaxioms, the sets $\T_n$ of \emph{approximated proof trees of level $n$ in $\Pair{\is}{\coaxioms}$}, for $n\in \N$, are inductively defined as follows:
\begin{center}
\begin{tabular}{ll}
$t\in\T_0$ & if $t$ well-founded proof tree in $\Extended{\is}{\coaxioms}$\\[2ex]
$\Rule{\T}{{\cons}} \in \T_{n+1}$ & if $\Rule{\prem}{\cons}\in\is$ and $\T = {(t_\judg)}_{\judg \in \prem}$ and $\forall \judg \in \prem.\, t_\judg \in \T_n \text{ and }t_\judg(\EString) = \judg$
\end{tabular}
\end{center}
\end{defi}
{Note that an approximated proof tree is a proof tree, since the set $\T$ and the set $\prem$ in the inductive step are in bijection: elements $t_\judg$ in $\T$ are indexed by judgements in $\prem$, hence there is a surjective map from $\prem$ to $\T$; moreover, if $t_\judg = t_{\judg'}$, then $\judg = t_\judg(\EString) = t_{\judg'}(\EString) = {\judg'}$, hence this map is also injective. }

In other words,  an approximated proof tree of level $n$ in $\Pair{\is}{\coaxioms}$ is a well-founded proof tree in $\Extended{\is}{\coaxioms}$ where coaxioms can only be used at depth  $\geq n$.
Therefore, if $t \in \T_n$ is an approximated proof tree of level $n$, {then,} for all $\alpha \in \dom(t)$ with $\Len{\alpha} < n$, $t(\alpha)$ is the consequence of a rule in $\is$, more precisely
\[
    \Rule{ \{t(\beta) \mid \beta \in \chl(\alpha) \} }{t(\alpha)} \in \is
\]

Another simple property of approximated proof trees is state{d} in the following proposition{.}
\begin{prop}\label{prop:approx-subtree}
If $t \in \T_n$, $\alpha \in \dom(t)$ and $\Len{\alpha} = k \le n$, then $t\Restrict{\alpha} \in \T_{n-k}$.
\end{prop}
\begin{proof}
We proceed by induction on $\Len{\alpha}$.
If $\Len{\alpha} = 0$, then $\alpha= \EString$, hence $t\Restrict{\EString} = t \in \T_n$.
Assume $\Len{\alpha} = k+1$, hence $\alpha = \beta a$, hence $\beta \in \dom(t)$ and $\Len{\beta}  = k$.
Therefore, by induction hypothesis $t\Restrict{\beta} \in \T_{n-k}$, hence  $t\Restrict{\alpha} = (t\Restrict{\beta})\Restrict{a} \in \dsub(t\Restrict{\beta})$, and this implies, by \refToDefinition{approx-trees}, that $t\Restrict{\alpha} \in \T_{n-k-1}$.
\end{proof}

The following theorem states that approximated proof trees  of level $n$ correspond to the $n$-th element of the descending chain $\Iterate{\Op{\is}}{\bound}=\{ \IterOp{\is}{n}(\bound) \mid n \in \N\}$, with  $\bound=\closure_\Op{\is}(\coaxioms)=\Ind{\Extended{\is}{\coaxioms}}$.

\begin{thm}\label{theo:approx-trees}
Let $\Pair{\is}{\coaxioms}$ be an inference system with coaxioms, and $\judg\in\universe$ a judgement.
We have that, for all $n \in \N${,} the following are equivalent{:}
\begin{enumerate}
\item $\judg\in \IterOp{\is}{n}(\closure_\Op{\is}(\coaxioms))$
\item $\judg$ has an approximated proof tree of level $n$ in $\Pair{\is}{\coaxioms}$
\end{enumerate}
\end{thm}
\begin{proof}
Let $\bound$ be $\closure_\Op{\is}(\coaxioms)$.
We prove the thesis by induction on $n$.
\begin{description}

\item[Base] If $n=0$, then, by \refToTheorem{correspondence}, $\bound = \closure_\Op{\is}(\coaxioms)$ corresponds to the inductive interpretation of $\Extended{\is}{\coaxioms}$, hence the the thesis reduces to the standard equivalence between proof-theoretic and fixed point semantics in the inductive case (see~\cite{LeroyGrall09}).

\item[Induction] We assume the equivalence for $n$ and prove it for $n+1$.
We prove separately the two implications.
\begin{description}

\item[$1\Rightarrow 2$] If $\cons \in \IterOp{\is}{n+1}(\bound)$, then there exists $\frac{\prem}{\cons} \in \is$ such that $\prem \subseteq \IterOp{\is}{n}(\bound)$.
Hence, by induction hypothesis, each judgement in $\prem$ has an approximated proof tree of level $n$, that is, for all $\judg \in \prem$ there is an approximated proof tree $t_\judg \in \T_n$ rooted in $\judg$.
Set $\T = \{t_\judg \in \T_n \mid \judg \in \prem\}$.
Hence, $t = \frac{\T}{\cons}$ is a proof tree for $\cons$, and by \refToDefinition{approx-trees}, $t\in \T_{n+1}$.

\item[$2\Rightarrow 1$] If $t \in \T_{n+1}$ is an approximated proof tree for $\cons \in \universe$, then, by definition, there exists $\frac{\prem}{\cons} \in \is$ such that $t = \frac{\T}{\cons}$, $\T = {(t_\judg)}_{\judg \in \prem}$ and for all $\judg \in \prem$, $t_\judg \in \T_n$ and $t_\judg(\EString) = \judg$.
By induction hypothesis we have $\prem\subseteq \IterOp{\is}{n} (\bound)$, and, by definition  of $\Op{\is}$, this implies $\cons \in \IterOp{\is}{n+1}(\bound)$ as needed.
\qedhere
\end{description}
\end{description}
\end{proof}

\noindent
The second proof-theoretic characterization of the interpretation generated by coaxioms  is an immediate consequence of the above theorem.

\begin{cor}\label{cor:tree2}
Let $\Pair{\is}{\coaxioms}$ be an inference system with coaxioms, and $\judg \in \universe$ a judgement.
Then the following are equivalent:
\begin{enumerate}
\item $\judg \in \Generated{\is}{\coaxioms}$
\item there exists a proof tree $t$ for $\judg$ in $\is$ such that each node of $t$ has an approximated proof tree of level $n$ in $\Pair{\is}{\coaxioms}$, for all $n \in \N$.
\end{enumerate}
\end{cor}
\begin{proof}
By \refToTheorem{correspondence}, \refToProposition{ker-iterate-bounds}, and \refToTheorem{ker-trees}, we get that,  for all $\judg \in \universe$, $\judg \in \Generated{\is}{\coaxioms}$ iff there exists a proof tree $t$ for $\judg$ in $\is$ such that each node $\judg'$ of $t$ is in $\bigcap \Iterate{\Op{\is}}{\bound}$ with $\bound = \closure_\Op{\is}(\coaxioms)$.
By \refToTheorem{approx-trees}, $\judg' \in \bigcap \Iterate{\Op{\is}}{\bound}$ iff has an approximated proof tree of level $n$, for all $n \in \N$.
\end{proof}

If the hypotheses of \refToCorollary{bfp-chain}  are satisfied, {namely, if the inference operator is downward continuous,} then we get a simpler equivalent proof-theoretic characterization.

\begin{cor}\label{cor:tree3}
Let $\Pair{\is}{\coaxioms}$ be an inference system with coaxioms, and $j \in \universe$ a judgement.
{I}f $\Op{\is}$ is downward continuous, then the following are equivalent:
\begin{enumerate}
\item $j \in \Generated{\is}{\coaxioms}$
\item $\judg$ has an approximated proof tree of level $n$ in $\Pair{\is}{\coaxioms}$, for all $n \in \N$.
\end{enumerate}
\end{cor}
\begin{proof}
Let $\bound$ be the set $\closure_\Op{\is}(\coaxioms)$.
By \refToTheorem{correspondence} and \refToCorollary{bfp-chain}, we get that $\Generated{\is}{\coaxioms} = \bigcap \Iterate{\Op{\is}}{\bound}$, therefore the thesis follows immediately from \refToTheorem{approx-trees}.
\end{proof}

\EZComm{caratterizzazione nuova, capire come si possa usare come tecnica e se migliora le cose}
\FDComm{è il caso di fare un sottosezione qui?}
In order to define the last proof{-}theoretic characterization (\refToTheorem{approx-sequence}), we need to introduce  a richer structure on trees.
In particular, we will {consider the partial order on trees defined by Courcelle in~\cite{Courcelle83}, adapted to our context. }

Consider trees $t, t' \in \Tree{\Labels}$ we define
\[
t \TOrder t' \Longleftrightarrow \dom(t) \subseteq \dom(t') \mbox{ and } \forall \alpha \in \dom(t).\ t(\alpha) = t'(\alpha)
\]
It is easy to see that $\TOrder$ is a partial order, {actually it is function inclusion. }
Indeed, reflexivity and transitivity follow from the same properties of $\subseteq$ and $=${, and a}ntisymmetry can be proved noting that if $t\TOrder t'$ and $t' \TOrder t$ we have that $\dom(t) = \dom(t')$ (by antisymmetry of $\subseteq$) and $t(\alpha) = t'(\alpha)$ for all $\alpha \in \dom(t)$, hence $t=t'$.

{Intuitively, $t \TOrder t'$ means that $t$ can be obtained from $t'$ by pruning some branches.
Alternatively, considering trees as graphs, $t\TOrder t'$ means that $t$ is a subgraph of $t'$.
In any case, $\TOrder$ {expresses} a very strong relation among trees, actually too strong for our aims, hence we need to relax it a little bit. }

{We relax {the order relation} by considering what we call its \emph{$n$-th approximation},} defined below.
Given a tree $t$, we denote by $\dom_n(t)$ the set $\{\alpha \in \dom(t) \mid \Len{\alpha} \le n\}$.
The $n$-th approximation of $\TOrder$, denoted by $\ApproxTOrder{n}$, is defined as follows:
\[
t \ApproxTOrder{n} t' \Longleftrightarrow
\dom_n(t) \subseteq \dom_n(t') \mbox{ and }
\forall \alpha \in \dom_n(t).\ t(\alpha) = t'(\alpha)
\]
Intuitively $\ApproxTOrder{n}$ is identical to $\TOrder$, but limited to nodes at level $\le n$.
We call it the $n$-th approximation of $\TOrder$ since $\ApproxTOrder{n}$ is coarser than $\TOrder$, namely, if $t\TOrder t'$ then $t \ApproxTOrder{n} t'$ for all $n \in \N$.
Actually we can say even more: $t \TOrder t'$ if and only if $t \ApproxTOrder{n} t'$ for all $n \in \N$.
Moreover, if $t \ApproxTOrder{n} t'$ then for all $k \le n$ we have $t \ApproxTOrder{k} t'$, that is, $\ApproxTOrder{n}$ is a finer approximation than $\ApproxTOrder{k}$.
{Finally, note} that $\ApproxTOrder{n}$ is reflexive and transitive, but it fails to be antisymmetric, because we compare only nodes until level $n$, hence we cannot conclude an equality between the whole trees.

We now state a result that is crucial for our proof{-}theoretic characterization (\refToTheorem{approx-sequence}).
{Indeed, the following theorem shows that a collection of trees, behaving like a sequence of more and more precise approximations, uniquely determines a tree, which can be regarded as the limit of such sequence. }

\begin{thm}\label{theo:limit-tree-sequence}
Let ${(t_n)}_{n \in \N}$ be a sequence of trees,
such that{,} for all $n \in \N$, $t_n \ApproxTOrder{n} t_{n+1}$.
Then, there exists a tree $t$ such that $\forall n \in \N.\ t_n \ApproxTOrder{n} t$, and{,} for any other tree $t'$ such that $\forall n \in \N.\ t_n \ApproxTOrder{n} t'${,} we have $t \TOrder t'$.
\end{thm}
\begin{proof}
We define the function $\fun{t}{L}{\Labels}$ where $L = \bigcup_{n \in \N} \dom_n(t_n)$  and for all $\alpha \in L$, $t(\alpha) = t_{k(\alpha)}(\alpha)$, where $ k(\alpha) = \min D_{\alpha}$ with $D_{\alpha} = \{ n \in \N \mid \alpha \in \dom_n(t_n) \}$.
Note that $k(\alpha)$ is well-defined, because $D_{\alpha} \ne \emptyset$, since $\alpha \in L$ and, by construction of $L$, there is at least an index $n \in \N$ such that $\alpha \in \dom_n(t_n)$.
Moreover, $L$ is a tree language, since if $\alpha a \in L$, then $\alpha a \in \dom_n(t_n)$ for some $n \in \N$, that is a tree language, hence $\alpha \in \dom_n(t_n) \subseteq L$.
Therefore, $t$ is a tree.\\
Fix now $n \in \N$, we have to show that $t_n \ApproxTOrder{n} t$.
By construction we have $\dom_n(t_n) \subseteq \dom_n(t)$, and if $\alpha \in \dom_n(t_n)$, then by construction $k(\alpha) \le n$.
Therefore, we have that $t_{k(\alpha)} \ApproxTOrder{k(\alpha)} t_n$, hence $t_{k(\alpha)}(\alpha) = t_n(\alpha)$, thus $t(\alpha) = t_n(\alpha)$ and this implies $t_n \ApproxTOrder{n} t$.

Consider now a tree $t'$ such that $\forall n \in \N.\ t_n \ApproxTOrder{n} t'$.
Therefore we have that for all $n \in \N$, $\dom_n(t_n) \subseteq \dom_n(t') \subseteq \dom(t')$, hence $\dom(t) \subseteq \dom(t')$.
Then, if $\alpha \in \dom(t)$, there is $n \in \N$ such that $\alpha \in \dom_n(t_n)$ and $t(\alpha) = t_n(\alpha)$.
Since $t_n \ApproxTOrder{n} t'$, we have that $t_n(\alpha) = t'(\alpha)$, hence $t(\alpha) = t'(\alpha)$ and this implies $t \TOrder t'$.
\end{proof}

{It is easy to see that a tree $t$ having the property expressed in  the above theorem is unique.
Indeed, if $t$ and $t'$ have that property for a sequence ${(t_n)}_{n \in \N}$, then  we have  both $t \TOrder t'$ and $t' \TOrder t$, hence $t = t'$.
Therefore we denote such a tree by $\TLub_{n \in \N} t_n$. }

The above theorem ensures the existence of a sort of least upper bound of an ascending chain of trees: {$\TLub_{n \in \N} t_n$  behaves like a least upper bound, but for approximations of a partial order.
However, since $\ApproxTOrder{n}$ is an approximation of $\TOrder$, it can be shown that if ${(t_n)}_{n \in \N}$ is a chain with respect to $\TOrder$, then $\TLub_{n \in \N} t_n$ is indeed the least upper bound of the chain, as {stated} in the following corollary. }

\begin{cor}
Let ${(t_n)}_{n \in \N}$ be a sequence of trees,
such that for all $n \in \N$, $t_n \TOrder t_{n+1}$.
Then, $\TLub_{n \in \N} t_n$ is the least upper bound of the sequence ${(t_n)}_{n \in \N}$ {with respect to $\TOrder$.  }
\end{cor}
\begin{proof}
{Since $\ApproxTOrder{n}$ is an approximation of $\TOrder$ we have that $t_n \ApproxTOrder{n} t_{n+1}$ for all $n \in \N$. Setting $t = \TLub_{n \in \N} t_n$, by \refToTheorem{limit-tree-sequence}, we get $t_n \ApproxTOrder{n} t$ for all $n \in \N$.
We have to show that $t$ is an upper bound of ${(t_n)}_{n \in \N}$, hence consider $\alpha \in \dom(t_n)$ and suppose $\Len{\alpha} = k$.
We have two cases{:}
\begin{itemize}
\item if $k \le n$, then $\alpha \in \dom_n(t_n)$, hence $\alpha \in \dom_n(t) \subseteq \dom(t)$ and $t_n(\alpha) = t(\alpha)$
\item if $k > n$, then, since $t_n \TOrder t_k$, $\alpha \in \dom_k(t_k) \subseteq \dom_k(t) \subseteq \dom(t)$, and $t_n(\alpha) = t_k(\alpha) = t(\alpha)$
\end{itemize}
Therefore we get $t_n \TOrder t$. }

To show that $t$ is the least upper bound, consider  an upper bound $t'$, hence $t_n \TOrder t'$ for all $n  \in \N$, and this implies that $t_n \ApproxTOrder{n} t'$ for all $n \in \N$.
Therefore, by \refToTheorem{limit-tree-sequence}, we get $t \TOrder t'$.\qedhere
\end{proof}

We now consider the equivalence relations induced by each $\ApproxTOrder{n}$, defined as follows:
\[
t \ApproxEq{n} t' \Longleftrightarrow t \ApproxTOrder{n} t' \mbox{ and } t' \ApproxTOrder{n} t
\]
{or, more explicitly:}
\[
t \ApproxEq{n} t' \Longleftrightarrow \dom_n(t) = \dom_n(t') \mbox{ and } \forall \alpha \in \dom_n(t).\ t(\alpha) = t'(\alpha)
\]
These equivalence relations are an approximation of the equality relation, indeed $t = t'$ if and only if $t \ApproxEq{n} t'$ for all $n \in \N$.

{The relations $\ApproxTOrder{n}$ and $\ApproxEq{n}$ look very similar: they are both an approximation of another relation, they are both reflexive and transitive and they both do not care about levels higher than $n$.
However, the fact that $\ApproxEq{n}$ is an equivalence relation {makes} it different.
Indeed, if $t \ApproxEq{n} t'$, then the first $n$ levels of $t$ and $t'$ are forced to be equal, while{,} if $t \ApproxTOrder{n} t'$, then the first $n$ levels of $t'$ must contain also those of $t$, but can have also  additional branches.
In other words, with $\ApproxEq{n}$ we can change only the depth of the trees, while with $\ApproxTOrder{n}$ we can change both the depth and the breadth. }

As a consequence of \refToTheorem{limit-tree-sequence}, we get the following theorem.

\begin{thm}\label{theo:eq-tree-sequence}
Let ${(t_n)}_{n \in \N}$ be a sequence of trees,
such that, for all $n \in \N$, $t_n \ApproxEq{n} t_{n+1}$.
Then, there exists a unique tree $t$ such that $\forall n \in \N.\ t_n \ApproxEq{n} t$.
\end{thm}
\begin{proof}
By definition of $\ApproxEq{n}$ we have that $t_n \ApproxTOrder{n} t_{n+1}$ and $t_{n+1} \ApproxTOrder{n} t_n$ for all $n \in \N$.
Therefore, by \refToTheorem{limit-tree-sequence}, we get that there is a tree $t = \TLub_{n \in \N} t_n$ such that, for all $n \in \N$, $t_n \ApproxTOrder{n} t$, hence we have only to prove that $t \ApproxTOrder{n} t_n$ {and that $t$ is unique. }
Since we know that $\dom_n(t_n) \subseteq \dom_n(t)$ and for all $\alpha \in \dom_n(t_n)$, $t_n(\alpha) = t(\alpha)$, it is enough to show that $\dom_n(t) \subseteq \dom_n(t_n)$.
Thus, consider $\alpha \in \dom_n(t)$.
By construction of $t$ (see the proof of \refToTheorem{limit-tree-sequence}), there is an  index $k\in \N$ such that $\alpha \in \dom_{k}(t_k)$.
We have two cases:
\begin{itemize}
\item If $k\le n$, then by hypothesis $\dom_k(t_k) \subseteq\dom_k(t_n) \subseteq \dom_n(t_n)$, hence $\alpha \in \dom_n(t_n)$.
\item Otherwise, {that is,} if $n<k$, since $\alpha\in \dom_n(t)$, we have $\Len{\alpha} \le n<k$, hence $\alpha \in \dom_n(t_k) \subseteq \dom_n(t_n)$, because, by hypothesis $t_n \ApproxEq{n} t_k$.
\end{itemize}
Therefore we have $\dom_n(t) \subseteq \dom_n(t_n)$ as needed.

{To prove that $t$ is unique,} consider a tree $t'$ such that $t_n \ApproxEq{n} t'$ for all $n\in \N$, by transitivity we get $t \ApproxEq{n} t'$ for all $n \in \N$, and this implies $t=t'$.
\end{proof}

{It is well known that trees carry a complete metric space structure~\cite{ArnoldNivat80, Courcelle83} and, even if our notion of tree is more general than that adopted in these works, we can recover the same metric on our trees, using the equivalence relations introduced earlier.
The metric is defined as follows:
\begin{mathpar}
d(t, t') = 2^{-h} \and h = \min\{n \in \N \mid t \not\ApproxEq{n} t'\}
\end{mathpar}
with assumptions $\min \emptyset = \infty$ and $2^{-\infty} = 0$.
It is easy to see that a sequence ${(t_n)}_{n \in \N}$ such that $t_n \ApproxEq{n} t_{n+1}$, like that considered in \refToTheorem{eq-tree-sequence}, is a Cauchy sequence in the metric space; indeed $d(t_n, t_{n+1}) \le 2^{-n}$.
Therefore, such sequences converge also in the metric space, and the limit is the same.
However, our notion of convergence seems to be more general: sequences like those considered in \refToTheorem{limit-tree-sequence} are not necessarily Cauchy sequences, but they admit a limit in our framework.
For instance, consider the sequence ${(t_n)}_{n \in \N}$ of children injective trees labelled on $\N$ and rooted in $0$ defined\footnote{It is enough to provide a definition for the domains since trees are children injective.} by
\begin{mathpar}
\dom(t_0) = \{\EString\} \and \dom(t_{n+1}) = \dom(t_n) \cup  \{n\}
\end{mathpar}
It is easy to check that $t_n \ApproxTOrder{n} t_{n+1}$, hence, by \refToTheorem{limit-tree-sequence}, it converges to $\TLub_{n\in \N} t_n$.
However, it is not a Cauchy sequence, since $d(t_n, t_{n+1}) = 2^{-1}$ for all $n \in \N$, and $\TLub_{n \in \N} t_n$ is not a limit of the sequence in the metric space.
A deeper comparison between these relation{s} and the standard metric structure on trees will be matter of further work. }

We can now introduce the concept that will allow the last proof-theoretic characterization.

\begin{defi}\label{def:approx-sequence}
Let $\Pair{\is}{\coaxioms}$ be an inference system with coaxioms and $\judg \in \universe$ a judgement. Then{:}
\begin{enumerate}
\item An \emph{approximating proof sequence} for $\judg$ is a sequence of proof trees ${(t_n)}_{n \in \N}$ for $\judg$ such that $t_n \in \T_n$ and  $t_n \ApproxTOrder{n} t_{n+1}$ for all $n \in \N$.
\item A \emph{{strongly approximating} proof sequence} for $\judg$ is a sequence of proof trees ${(t_n)}_{n \in \N}$ for $\judg$ such that $t_n \in \T_n$ and $t_n \ApproxEq{n} t_{n+1}$ for all $n \in \N$.
\end{enumerate}
\end{defi}

\noindent
Obviously every {strongly approximating} proof  sequence is also an approximating proof sequence.
Note also that all trees in these sequences are well-founded proof trees in $\Extended{\is}{\coaxioms}$.
Intuitively, both notions represent the growth of a proof for $\judg$ in $\is$ approximated using coaxioms.
The difference is that trees in an approximating proof sequence can grow both in depth {and} in breadth, while in a {strongly approximating} proof sequence they can grow only in depth.
We now prove our last theorem, characterizing $\Generated{\is}{\coaxioms}$ in terms of (strongly) approximating proof sequences. 

\begin{thm}\label{theo:approx-sequence}
Let $\Pair{\is}{\coaxioms}$ be an inference system with coaxioms and $\judg\in \universe$ a judgement.
Then the following are equivalent
\begin{enumerate}
\item $\judg \in \Generated{\is}{\coaxioms}$
\item $\judg$ has a {strongly approximating} proof sequence
\item $\judg$ has an approximating proof sequence $(t_n)_{n \in \N}$ such that $\TLub_{n \in \N} t_n$ is a proof tree in $\is$. 
\end{enumerate}
\end{thm}
\begin{proof}
{We assume the canonical representation for children injective trees. }
\begin{description}

\item[$1\Rightarrow 2$]
{We define trees $t_{\judg, n}$ for $\judg \in \Generated{\is}{\coaxioms}$ and $n \in \N$ such that $t_{\judg, n}(\EString) = \judg$ by induction on $n$.
By \refToCorollary{proof-trees-1}, we know that every judgement $\judg \in \Generated{\is}{\coaxioms}$ has a well-founded proof tree in $\Extended{\is}{\coaxioms}$,  that is, a proof tree in $\T_0$ rooted in $\judg$: we select one of these trees and call it $t_{\judg, 0}$.
Furthermore, since $\Generated{\is}{\coaxioms}$ is a post-fixed point, for any $\judg \in \Generated{\is}{\coaxioms}$ we can select a rule $\frac{\prem_\judg}{\judg} \in \is$ with $\prem_\judg \subseteq \Generated{\is}{\coaxioms}$; hence $t_{\judg, n+1}$ {can be} defined as follows{:} }
\[
t_{\judg, n+1} = \Rule {
	\{ t_{\judg', n} \mid \judg' \in \prem_\judg \}
}{\judg}
\]
Clearly, by construction for all $\judg \in \Generated{\is}{\coaxioms}$ and for all $n \in \N$, $t_{\judg, n} \in \T_n$.
We show by induction on $n$ that for all $n \in \N$ and for all $\judg \in \Generated{\is}{\coaxioms}$,  $t_{\judg, n} \ApproxEq{n} t_{\judg, n+1}$.
\EZComm{controllare: usare lo stile sotto con Base e Induction ovunque nelle prove per induzione}
\begin{description}
\item[Base] If $n=0$, then $\dom_0(t_{\judg, 0}) = \dom_0(t_{\judg, 1}) = \{\EString\}$ and by construction $ t_{\judg, 0}(\EString) = t_{\judg, 1}(\EString) = \judg$, hence $t_{\judg, 0} \ApproxEq{0} t_{\judg, 1}$.
\item[Induction] We assume the thesis for $n-1$ and prove it for $n$, hence we have to show that $t_{\judg, n} \ApproxEq{n} t_{\judg, n+1}$.
By construction, we have
\begin{mathpar}
t_{\judg, n} = \Rule{ \{ t_{\judg', n-1} \mid \judg' \in \prem_\judg \} }{\judg}
\and
t_{\judg, n+1} = \Rule{ \{ t_{\judg', n} \mid \judg' \in \prem_\judg \} }{\judg}
\end{mathpar}
By induction hypothesis, we get $t_{\judg', n-1} \ApproxEq{n-1} t_{\judg', n}$ for all $\judg' \in \prem_\judg$.
Therefore we have
\[\begin{split}
\dom_n(t_{\judg, n}) &= \{\EString\} \cup \bigcup_{\judg' \in \prem_\judg} \judg' \dom_{n-1}(t_{\judg', n-1})  \\
					&= \{\EString\} \cup \bigcup_{\judg' \in \prem_\judg} \judg' \dom_{n-1}(t_{\judg', n}) \\
					&= \dom_n(t_{\judg, n+1})
\end{split}\]
Consider now $\alpha \in \dom_n(t_{\judg, n})$, we have two cases{:}
\begin{align*}
\alpha &= \EString 								& & t_{\judg, n}(\alpha) = \judg = t_{\judg, n+1}(\alpha) \\
\alpha &= \judg' \beta,\ \judg' \in \prem_\judg	& & t_{\judg, n}(\alpha) = t_{\judg', n-1}(\beta) = t_{\judg', n}(\beta) = t_{\judg, n+1}(\alpha)
\end{align*}
and this shows $t_{\judg, n} \ApproxEq{n} t_{\judg, n+1}$ as needed.
\end{description}

\item[$2\Rightarrow 3$] By hypothesis, $\judg$ has a strongly approximating proof sequence $(t_n)_{n\in\N}$ and by \refToDefinition{approx-sequence} it is also an approximating proof sequence. 
We set $t = \TLub_{n \in \N} t_n$ and prove that $t$ is a proof tree in $\is$ for $\judg$. 
By \refToTheorem{eq-tree-sequence}, we have that $t_n \ApproxEq{n} t$ for all $n \in \N$, 
hence, we get $\judg = t_0(\EString) = t(\EString)$.
Consider $\alpha \in \dom(t)$ and set $n = \Len{\alpha} + 1$. 
By construction of $t$, we have that $\alpha \in \dom_n(t_n)$ and $\chl_t(\alpha) = \chl_{t_n}(\alpha) \subseteq \dom_n(t_n)$, as $t_n \ApproxEq{n} t$. 
Then, since $t_n \in \T_n$ and $\Len{\alpha} < n$ the rule
\[
    \Rule{ \{t(\beta) \mid \beta \in \chl_t(\alpha)\} }{t(\alpha)}
\]
is a rule in $\is$ by \refToDefinition{approx-trees}, thus $t$ is a proof tree in $\is$.

\item[$3\Rightarrow 1$] Set $t = \TLub_{n \in \N} t_n$, which, by hypothesis, is a proof tree in $\is$ for $\judg$. 
Consider a node $\alpha \in \dom(t)$, then there is $k \in \N$ such that $\alpha \in \dom_k(t_k)$, and so $\Len{\alpha} \le k$ and $\alpha \in \dom_m(t_m)$ for all $m \ge k$, by \refToDefinition{approx-sequence}. 
We define the sequence ${(t_n^\alpha)}_{n \in \N}$ such that $t_n^\alpha = {t_{n+k}}\Restrict{\alpha}$.
By \refToProposition{approx-subtree}, we get $t_n^\alpha \in \T_{n+k-\Len{\alpha}} \subseteq \T_n$.
This observation shows that every node in $t$ has an approximated proof tree of level $n$ for all $n \in \N$, hence by \refToTheorem{approx-trees} we get $\judg \in \Generated{\is}{\coaxioms}$.
\qedhere
\end{description}
\end{proof}

The last condition in \refToTheorem{approx-sequence} requires to build an approximating proof sequence $(t_n)_{n \in \N}$ and to check its limit $\TLub_{n \in \N} t_n$  is a proof tree. 
However, for a large class of inference systems,  this last requirement is always true, as formally proved below. 
 
We say an inference system $\is$ is \emph{bounded} if there exists a natural number $b \in \N$ such that  for each rule $\myrule \in \is$, the cardinality of $\prem$ is less than or equal to $b$. 
Note that most common inference systems are bounded, since they are specified by a finite set of finitary meta-rules. 
An inference system with coaxioms $\Pair{\is}{\coaxioms}$ with $\is$ bounded enjoys the following key property: 
\begin{prop} \label{prop:approx-bounded}
Let $\Pair{\is}{\coaxioms}$ be an inference system with coaxioms where  $\is$ is bounded, and consider an approximating proof sequence $(t_n)_{n \in \N}$ in $\Pair{\is}{\coaxioms}$. 
Then $\TLub_{n \in \N} t_n$ is a proof tree. 
\end{prop}
\begin{proof}
Set $t = \TLub_{n \in \N}$ and consider $\alpha \in \dom(t)$. 
By definition of $t$, there exists $n > \Len{\alpha}$ such that $\alpha \in \dom_n(t_n)$ and so $\chl_t(\alpha) = \bigcup_{k\ge n} \chl_{t_k}(\alpha)$. 
Since $\dom_k(t_k) \subseteq \dom_k(t_{k+1})$ we have $\chl_{t_k}(\alpha) \subseteq \chl_{t_{k+1}}(\alpha)$, and, since all $t_k$ are proof trees and $\is$ is bounded, we get that $\chl_t(\alpha)$ is finite. 
As a consequence, we have that $\chl_t(\alpha) = \chl_{t_k}(\alpha)$ for some $k\ge n$, hence, since $\Len{\alpha} < k$, the rule 
\[ \Rule{ \{ t(\beta) \mid \beta \in \chl_t(\alpha) \} }{ t(\alpha)} \] 
is a rule in $\is$, as $t_k \in \T_k$, thus $t$ is a proof tree. 
\end{proof}

Therefore, thanks to this property, we can rephrase  \refToTheorem{approx-sequence} under the boundedness hypothesis as follows: 

\begin{cor}\label{cor:bounded-approx-sequence}
Let $\Pair{\is}{\coaxioms}$ be an inference system with coaxioms where $\is$ is bounded,  and $\judg\in \universe$ a judgement.
Then the following are equivalent
\begin{enumerate}
\item $\judg \in \Generated{\is}{\coaxioms}$
\item $\judg$ has a {strongly approximating} proof sequence
\item $\judg$ has an approximating proof sequence  
\end{enumerate}
\end{cor}


\section{Reasoning with coaxioms}\label{sect:coaxioms-reasoning}

In this section we discuss proof techniques for inference systems with coaxioms.

Assume that $\DefSet=\Generated{\is}{\coaxioms}$ \FD{(for ``defined'')}  is the interpretation generated by coaxioms for some $\Pair{\is}{\coaxioms}$, and that $\Spec$ (for ``specification'') is the intended set of judgements, called \emph{valid} in the following.

Typically, we are interested in proving $\Spec\subseteq\DefSet$ (\emph{completeness}, that is, each valid judgement can be derived) and/or  $\DefSet\subseteq\Spec$ (\emph{soundness}, that is, each derivable judgement is valid). Proving both properties amounts to say that the inference system with coaxioms actually defines the intended set of judgements.

For what follows, recall that, given an inference system with coaxioms $\Pair{\is}{\coaxioms}$, the following identities hold:
\[
\Generated{\is}{\coaxioms} = \CoInd{\Restricted{\is}{\Ind{\Extended{\is}{\coaxioms}}}} =
\ker_{\Op{\is}}(\Ind{\Extended{\is}{\coaxioms}})
\]

\subsubsection*{Completeness proofs}
To show completeness, we can use \CoIndPrinciple.
Indeed, since $\DefSet= \ker_\Op{\is} (\Ind{\Extended{\is}{\coaxioms}})$, if $\Spec \subseteq  \Ind{\Extended{\is}{\coaxioms}}$ and $\Spec$ is a post-fixed point of $\Op{\is}$, by \CoIndPrinciple we get that $\Spec \subseteq \DefSet$.
That is, using the notations of inference systems, to prove completeness it is enough to show that:
\begin{itemize}
\item $\Spec \subseteq \Ind{\Extended{\is}{\coaxioms}}$ {and}
\item $\Spec \subseteq \Op{\is}(\Spec)$
\end{itemize}
We call this principle the \emph{bounded coinduction principle}.

We illustrate the technique on the inference system with coaxioms $\Pair{\is}{\coaxioms}$ which defines the judgement $\allPos{l}{b}$.
{We report here the definition from \refToSection{coaxioms}, for the reader's convenience.
\begin{small}
\begin{mathpar}
\Rule{}{ \allPos{\EList}{\True} }
\and
\Rule{}{ \allPos{\Cons{x}{l}}{\False} } x\le 0
\and
\Rule{ \allPos{l}{b} }{ \allPos{\Cons{x}{l}}{b} }x>0
\and
\CoAxiom{ \allPos{l}{\True} }
\end{mathpar}
\end{small}  }
Let $\SpecAllPos$ be the set of judgements $\allPos{l}{b}$ where $b$ is $\True$ if all the elements in $l$ are positive, $\False$ otherwise.
Completeness means that the judgement  $\allPos{l}{b}$ can be proved, for all $\allPos{l}{b}\in\SpecAllPos$. By the {bounded coinduction} principle, it is enough to show that
\begin{itemize}
\item $\SpecAllPos \subseteq \Ind{\Extended{\is}{\coaxioms}}$
\item $\SpecAllPos \subseteq \Op{\is}(\SpecAllPos)$
\end{itemize}
To prove the first condition, we have to show that, for each $\allPos{l}{b}\in\SpecAllPos$, $\allPos{l}{b}$ has a \emph{finite} proof tree in $\Extended{\is}{\coaxioms}$.
This can be easily shown, indeed:
\begin{itemize}
\item If $l$ contains a (first) non-positive element, hence\\ $l=\Cons{x_1}{ \Cons{\ldots}{ \Cons{x_n}{ \Cons{x}{l'} } } }$ with $x_i>0$ for $i\in [1..n]$, $x \leq 0$, and $b=F$\\
then we can reason by arithmetic induction on $n$. Indeed, for $n=0$, $\allPos{l}{b}$ is the consequence of the second rule with no premises, and for $n>0$ it is the consequence of the third rule where we can apply the induction hypothesis to the premise.
\item If $l$ contains only positive elements, hence $b=T$, then $\allPos{l}{b}$ is a coaxiom, hence it is the consequence of a rule with no premises in $\Extended{\is}{\coaxioms}$.
\end{itemize}
To prove the second condition, we have to show that, for each $\allPos{l}{b}\in\SpecAllPos$, $\allPos{l}{b}$ is the consequence of a rule with premises in $\SpecAllPos$. This can be easily shown, indeed:
\begin{itemize}
\item If $l=\Lambda$, hence $b=T$, then $\allPos{\Lambda}{\True}$ is the consequence of the first rule with no premises.
\item If $l=\Cons{x}{l'}$ with $x\leq 0$, hence $b=F$, then $\allPos{l}{\False}$ is the consequence of the  second rule with no premises.
\item If $l=\Cons{x}{l'}$ with $x > 0$, and $b=T$, hence $\allPos{l'}{\True}\in\SpecAllPos$, then $\allPos{l}{\True}$ is the consequence of the  third rule with premise $\allPos{l'}{\True}$, and analogously if $b=F$.
\end{itemize}

\subsubsection*{Soundness proofs}
To show soundness, it is convenient to use the alternative characterization in terms of approximated proof trees given in \refToSection{coaxioms-trees}, as detailed below.
First of all, from \refToProposition{ker-iterate-bounds}, $\DefSet\subseteq\bigcap \{\IterOp{\is}{n}(\Ind{\Extended{\is}{\coaxioms}})\mid n \geq 0\}$.
Hence, to prove $\DefSet\subseteq\Spec$, it is enough to show that $\bigcap \{\IterOp{\is}{n}(\Ind{\Extended{\is}{\coaxioms}})\mid n \geq 0\}\subseteq\Spec$.
Moreover, by \refToTheorem{approx-trees}, for all $n\in \N$, judgements in $\IterOp{\is}{n}(\Ind{\Extended{\is}{\coaxioms}})$ are those which have an approximated proof tree of level  $n$.
Hence, to prove {the above inclusion}, we can show that all judgements, which have an approximated proof tree of level $n$ for each $n$, are in $\Spec$ or equivalently, by {contraposition},  that
judgements, which are not in $\Spec$, that is, non-valid judgements, fail to have an approximated proof tree of level $n$ for some $n$.

We illustrate the technique again on the example of \allPosName.
We have to show that, for each $\allPos{l}{b}\not\in\SpecAllPos$, there exists $n\geq 0$ such that $\allPos{l}{b}$ cannot be proved by using coaxioms at level greater than $n$. By cases:

\begin{itemize}
\item If $l$ contains a (first) non-positive element, hence\\ $l=\Cons{x_1}{ \Cons{\ldots}{ \Cons{x_n}{ \Cons{x}{l'} } } }$ with $x_i>0$ for $i\in [1..n]$, $x \leq 0$, then, assuming that coaxioms can only be used at a level greater  than $n+1$, $\allPos{l}{b}$ can only be derived by instantiating $n$ times the third rule, and once the second rule, hence $b$ cannot be $\True$.
\item If $l$ contains only positive elements, then it is immediate to see that there is no finite proof tree for $\allPos{l}{\False}$.
\end{itemize}


\section{Taming coaxioms: advanced examples}\label{sect:coaxioms-examples}

In this section we will present some {more} examples of situations where coaxioms can help to {define} judgements on non-well-founded structures.
These more \EZ{involved} examples will serve for explaining how to use coaxioms, which kind of problems they can cope with, {and how complex can be the interaction between coaxioms and standard rules. }

\subsection{Mutual recursion}
Circular definitions involving inductive and coinductive judgements have no semantics in standard inference systems,
because all judgements have to be interpreted either inductively, or coinductively.
The same problem arises in the context of coinductive logic programming~\cite{SimonBMG07},
where a logic program has a well-defined semantics only if inductive and coinductive predicates can be stratified:
each stratum defines only inductive or coinductive predicates (possibly defined in a mutually recursive way),
and cannot depend on predicates defined in upper strata.
Hence, it is possible to define the semantics of a logic program only if there are no mutually recursive definitions involving both inductive and coinductive predicates.

We have already seen that an inductive inference system corresponds to an inference system with coaxioms where there are no coaxioms,
while a coinductive one corresponds to the case where coaxioms consist of all judgements in $\universe$;
however, between these two extremes, coaxioms offer many other possibilities thus allowing a finer control on the
semantics of the defined judgements.

There exist cases where two or more related judgements need to be defined recursively,
but for some of them the correct interpretation is inductive, while for others is coinductive~\cite{SimonMBG06,SimonBMG07,Ancona13,BasoldK16}.
In such cases, coaxioms may be employed to
provide the correct definition in terms of a single inference system with no stratification.
{However, the interaction between coaxioms and standard rules is not that easy, hence special care is required
to get from the inference system the intended meaning of judgements. }
In order to see this, let us consider the judgement $\pathZero(t)$, where $t$ is an infinite tree\footnote{For the purpose of this example,
we only consider trees with infinite depth and branching.} over $\{0,1\}$, which
holds iff there exists a path starting from the root of $t$  and containing just $0$s.
Trees are represented as infinite terms of shape $\tree(n,l)$, where $n\in\{0,1\}$ is the
root of the tree, and $l$ is the infinite list of its direct subtrees.
For instance, if $t_1$ and $t_2$ are the trees defined by the
syntactic equations
\begin{mathpar}
t_1=\tree(0,l_1)
\and
l_1=\Cons{t_2}{ \Cons{t_1}{l_1} }
\and
t_2=\tree(0,l_2)
\and
l_2=\Cons{\tree(1,l_1)}{l_2}
\end{mathpar}
 then we expect
$\pathZero(t_1)$ to hold, but not $\pathZero(t_2)$.

To define $\pathZero$, we introduce an auxiliary judgement $\isinZero(l)$ testing whether an infinite list $l$ of trees contains a tree
$t$ such that $\pathZero(t)$ holds.
Intuitively, we expect $\pathZero$ and $\isinZero$ to be interpreted coinductively and inductively, respectively;
this reflects the fact that $\pathZero$ checks a property universally quantified over an infinite sequence
(a \emph{safety} property in the terminology of concurrent systems): all the elements of the path must {be equal to} $0$;
on the contrary, $\isinZero$ checks a property existentially
quantified over an infinite sequence (a \emph{liveness} property in the terminology of concurrent systems): the list must contain a
tree $t$ with a specific property (that is, $\pathZero(t)$ must hold).
Driven by this intuition, one could be tempted to define the following inference system with coaxioms for all judgements
of shape $\pathZero(t)$, and no coaxioms for judgements of shape $\isinZero(l)$:
\begin{mathpar}
\Rule{\isinZero(l)}{\pathZero(\tree(0,l))}{}
\and
\CoAxiom{\pathZero(t)}
\and
\Rule{\pathZero(t)}{\isinZero(\Cons{t}{l})}
\and
\Rule{\isinZero(l)}{\isinZero(\Cons{t}{l})}
\end{mathpar}
Unfortunately, because of the mutual recursion between $\isinZero$ and $\pathZero$, the inference system above does not capture the intended behaviour:
$\isinZero(l)$ is derivable for every infinite list of trees $l$, even when $l$ does not contain a tree $t$ with an infinite path starting from its root
and containing just $0$s.
{Indeed, the coaxiom we added is not really restrictive, because it allows the predicate $\pathZero$ to be coinductive, but, since $\isinZero$ directly depends on $\pathZero$, it is allowed to be coinductive as well. }

To overcome this problem, we can break the mutual dependency between judgements,
replacing the judgement $\isinZero$ with the more general one $\isin$, such that $\isin(t,l)$ holds iff the infinite list $l$ contains the tree $t$.
Consequently, we can define the following inference system with coaxioms:
\begin{mathpar}
\Rule{\isin(t,l) \Space \pathZero(t)}{\pathZero(\tree(0,l))}{}
\and
\CoAxiom{\pathZero(t)}
\and
\Rule{}{\isin(t,\Cons{t}{l})}
\and
\Rule{\isin(t,l)}{\isin(t,\Cons{t'}{l})}
\end{mathpar}

Now the semantics of the system corresponds to the intended one, {since now $\isin$ does not depend on $\pathZero$, hence the coaxioms do not influence the semantics of $\isin$, which remains inductive as expected.
Nevertheless, the semantics is well-defined without the need of stratifying the definitions into two separate inference systems. }

Following the characterization in terms of proof trees and the proof techniques provided in \refToSection{coaxioms-trees} and
\refToSection{coaxioms-reasoning}, we can sketch a proof of correctness.
Let $\Spec$ be the set  where
elements have either shape $\pathZero(t)$, where $t$ represents a tree with an infinite path of just $0$s starting from its root, or $\isin(t,l)$, where $l$ represents an infinite list containing the tree $t$; then a judgement belongs to $\Spec$ iff it can be derived in the
inference system with coaxioms defined above.

\subsubsection*{Completeness} We first show that the set $\Spec$ is a post-fixed point, that is, it is consistent w.r.t.
the inference rules, coaxioms excluded.
Indeed, if $t$ has an infinite path of $0$s, then it has necessarily  shape
$\tree(0,l)$, where $l$ must contain a tree $t'$ with an infinite path of $0$s.
Hence, the inference rule for $\pathZero$ can be applied with premises
$\isin(t',l)\in \Spec$, and $\pathZero(t')\in \Spec$.
If an infinite list contains a tree $t$, then
it has necessarily shape $\Cons{t'}{l}$ where, either $t=t'$, and hence the axiom for $\isin$ can be applied,
or $t\neq t'$ and $t$ is contained in $l$, and hence the inference rule for $\isin$ can be applied
with premise $\isin(t,l)\in \Spec$.

We then show that $\Spec$ is bounded by the closure of the coaxioms.
For the elements of shape  $\pathZero(t)$ it suffices to directly apply the corresponding coaxiom; for the elements
of shape $\isin(t,l)$ we can show that there exists a finite proof tree built possibly also with the coaxioms by induction on the
first position (where the head of the list corresponds to {$0$}) in the list where $t$ occurs.
If the position is {$0$} (base case), then $l=\Cons{t}{l'}$, and the axiom can be applied;
if the position is $n>0$ (inductive step), then $l=\Cons{t'}{l'}$ and $t$ occurs in $l'$ at position $n-1$,
therefore, by induction hypothesis, there exists a finite proof tree for $\isin(t,l')$, therefore
we can build a finite proof tree for $\isin(t,l)$ by applying the inference rule for $\isin$.

\subsubsection*{Soundness} We first observe that the
only finite proof trees that can be derived for $\isin(t,l)$ are obtained by
application of the axiom for $\isin$, hence $\isin(t,l)$ holds iff there
exists a finite proof tree for $\isin(t, l)$ built with the inference rules
for $\isin$.
Then, we can prove that, if $\isin(t,l)$ holds, then $t$ is contained in
$l$ by induction on the inference rules for $\isin$.
For the axiom (base case) the claim trivially holds, while for the other inference rule we
have that if $t$ belongs to $l$, then trivially $t$ belongs to $\Cons{t'}{l}$.

For the elements of shape $\pathZero(t)$ we first observe that by the coaxioms they all trivially belong to the closure
of the coaxioms.
Then, any proof tree for $\pathZero(t)$ must be infinite, because there are no axioms but only one inference rule
for $\pathZero$ where $\pathZero$ is referred in the premises; furthermore, such a rule is applicable only if the
root of the tree is {$0$}. We have already proved that if $\isin(t,l)$ is derivable, then $t$ belongs to $l$,
therefore we can conclude that if $\pathZero(t)$ is derivable, then $t$ contains an infinite path starting from its root,
and containing just $0$s.

\subsection{A numerical example}
It is well{-}known that real numbers in the closed interval $[0,1]$ can be represented
by infinite sequences ${(d_i)}_{i\in\N^+}$ of decimal\footnote{Of course the example can be generalized to any base $B\geq 2$.} digits,
where $\N^+$ denotes the set of all positive natural numbers.
Indeed, ${(d_i)}_{i\in\N^+}$ represents the real number which is the limit of the series $\sum_{i=1}^{\infty}10^{-i}d_i$ in the standard
complete metric space of real numbers (such a limit always exists by completeness,
because the associated sequence of partial sums is always a Cauchy sequence).
Such a representation is not unique for all rational numbers in $[0,1]$ (except for the bounds $0$ and $1$) that can be represented by a finite sequence of digits followed by an infinite sequence of $0$s;
for instance, $0.42$ can be represented either by the sequence $42\bar{0}$, or by the sequence $41\bar{9}$, where
$\bar{d}$ denotes the infinite sequence containing just the digit $d$.

For brevity, for $r={(d_i)}_{i\in\N^+}$, $\sem{r}$ denotes $\sum_{i=1}^{\infty}10^{-i}d_i$ (that is, the real number represented by $r$).
We want to define the judgement $\add(r_1,r_2,r,c)$ which holds iff $\sem{r_1}+\sem{r_2}=\sem{r}+c$ with $c$ an integer number; that is,
$\add(r_1,r_2,r,c)$ holds iff the addition of the two real numbers represented by the sequences $r_1$ and $r_2$ yields the real number
represented by the sequence $r$ with carry $c$.
We will soon discover that, to get a complete definition for $\add$,
$c$ is required to range over a proper superset of the set $\{0,1\}$, differently from what one could initially expect.

{We can define the judgement $\add$ by an  inference system with coaxioms as follows.
We represent a real number in $[0,1]$ by an infinite list of decimal digits, which, therefore,
can always be decomposed as $\Cons{d}{r}$, where $d$ is the first digit (corresponding to the exponent $-1$),
and $r$ is the rest of the list of digits.
Hence, in the definition below, $r, r_1, r_2$ range over infinite lists of digits, $d_1, d_2$ range over decimal digits (between $0$ and $9$), $c$ is an integer and
 $\div$ and $\bmod$ denote
the integer division, and the remainder operator, respectively. }

\begin{mathpar}
\Rule{\add(r_1,r_2,r,c)}{\add(\Cons{d_1}{r_1},\Cons{d_2}{r_2},\Cons{(s \bmod 10)}{r},s \div 10)}{s=d_1+d_2+c}
\and
\CoAxiom{\add(r_1,r_2,r,c)} c \in \{-1, 0, 1, 2\}
\end{mathpar}

As clearly emerges from the proof of completeness provided
below, besides the obvious values $0$ and $1$, the values $-1$ and $2$
have to be considered for the carry to ensure a complete definition of $\add$
because both $\add(\bar{0},\bar{0},\bar{9},-1)$  and $\add(\bar{9},\bar{9},\bar{0},2)$
hold, and, hence, should be derivable; these two judgements allow the derivation
of an infinite number of other valid judgements, as, for instance, $\add(1\bar{0},1\bar{0},1\bar{9},0)$ and
$\add(1\bar{9},1\bar{9},4\bar{0},0)$, respectively, {as \EZ{shown} by the following infinite derivations:
\begin{mathpar}
\Rule{
	\Rule{
		\Rule{
			\vdots
		}{ \add(\bar{0}, \bar{0}, \bar{9}, -1)  }
	}{ \add(\bar{0}, \bar{0}, \bar{9}, -1)  }
}{ \add(1\bar{0}, 1\bar{0}, 1\bar{9}, 0) }
\and
\Rule{
	\Rule{
		\Rule{
			\vdots
		}{ \add(\bar{9}, \bar{9}, \bar{0}, 2)  }
	}{ \add(\bar{9}, \bar{9}, \bar{0}, 2)  }
}{ \add(1\bar{9}, 1\bar{9}, 4\bar{0}, 0) }
\end{mathpar}
}

Also in this case we can sketch a proof of correctness: for all infinite sequences of decimal digits $r_1$, $r_2$ and
$r$, and all $c\in\{-1,0,1,2\}$, $\add(r_1,r_2,r,c)$ is derivable iff
$\sem{r_1}+\sem{r_2}=\sem{r}+c$.

\subsubsection*{Completeness}
By the coaxioms we trivially have that each element of shape
$\add(r_1,r_2,r,c)$ such that $\sem{r_1}+\sem{r_2}=\sem{r}+c$ with $c\in\{-1,0,1,2\}$ belongs
to the closure of the coaxioms.

To show that the unique inference rule of the system is consistent with the set of all valid  judgements,
let us assume that $\sem{r'_1}+\sem{r'_2}=\sem{r'}+c'$ with $r'_1=\Cons{d_1}{r_1}$, $r'_2=\Cons{d_2}{r_2}$, $r'=\Cons{d}{r}$ and $c'\in\{-1,0,1,2\}$.
Let us set $s=10c'+d$, and $c=s-d_1-d_2$, then $s\bmod 10=d$ and $s\div 10=c'$, and we get the desired conclusion of the
inference rule, and the side condition holds;
it remains to show that $\sem{r_1}+\sem{r_2}=\sem{r}+c$ with $c\in\{-1,0,1,2\}$.

We first observe that by the properties of limits w.r.t.\ the usual arithmetic operations, and
by definition of $\sem{-}$, for all infinite sequence $r$ of decimal digits, if $r=\Cons{d}{r'}$, then
$\sem{r}=10^{-1}(d+\sem{r'})$; then, from the hypotheses we get the equality $d_1+\sem{r_1}+d_2+\sem{r_2}=d+\sem{r}+10c'$,
hence $d_1+\sem{r_1}+d_2+\sem{r_2}=\sem{r}+s$,
and, therefore, $\sem{r_1}+\sem{r_2}=\sem{r}+c$; finally,
$c$ is an integer because it is an algebraic sum of integers, and, since $c=\sem{r_1}+\sem{r_2}-\sem{r}$, and $0\leq\sem{r_1},\sem{r_2},\sem{r}\leq1$, we get $c\in\{-1,0,1,2\}$.

\subsubsection*{Soundness}
Let $r'_1=\Cons{d_1}{r_1}$, $r'_2=\Cons{d_2}{r_2}$, and $r'=\Cons{d}{r}$ be infinite sequences of decimal digits, and $c'\in\{-1,0,1,2\}$;
we note that the judgement $\add(r'_1,r'_2,r',c')$ can only be derived from the unique inference rule where the premise is the judgement
$\add(r_1,r_2,r,c)$ with  $c$  equal to  $10c'+d-d_1-d_2$ and must range over $\{-1,0,1,2\}$.

To prove soundness we show that if $\sem{r'_1}+\sem{r'_2}\neq\sem{r'}+c'$, then  the judgement
$\add(r'_1,r'_2,r',c')$ cannot be derived in the inference system.
Let us set $\delta'=|\sem{r'}+c'-\sem{r'_1}-\sem{r'_2}|$; obviously, under the hypothesis
$\sem{r'_1}+\sem{r'_2}\neq\sem{r'}+c'$, we get $\delta' >0$.
In particular, the following fact holds:  if $\delta'\geq 4\cdot10^{-1}$, then
$10c'+d-d_1-d_2\not\in\{-1,0,1,2\}$.
Indeed, by the identity $\sem{r}=10^{-1}(d+\sem{r'})$ already
used for the proof of completeness, we have that
$\delta'=10^{-1}\delta$ with $\delta = |\sem{r}+c-\sem{r_1}-\sem{r_2}|$, with $c=10c'+d-d_1-d_2$;
$10^{-1}(\sem{r}+c-\sem{r_1}-\sem{r_2}) \geq 4\cdot10^{-1}$
implies $c\geq 3$
($\sem{r_1},\sem{r_2},\sem{r}\in[0,1]$), and, hence, $c=10c'+d-d_1-d_2\not\in\{-1,0,1,2\}$. On the other hand,
$10^{-1}(\sem{r}+c-\sem{r_1}-\sem{r_2})\leq -4\cdot10^{-1}$ implies
$c\leq -2$, hence $c=10c'+d-d_1-d_2\not\in\{-1,0,1,2\}$.

By virtue of this fact, and thanks to the hypotheses, we can prove by arithmetic induction over $n$ that for all $n\geq 1$, if
 $\delta'\geq 4\cdot10^{-n}$, then there exist only finite proof trees for $\add(r'_1,r'_2,r',c')$ where the coaxioms are applied at most at depth $n-1$.
The base case is directly derived from the previously proven fact. Indeed, for $n=1$, we can only derive $\add(r'_1,r'_2,r',c')$ by directly applying the coaxiom.
For the inductive step we observe that all derivation of depth greater than $1$ are built applying the inference rule as first step.
If such rule is applicable for deriving the conclusion $\add(r'_1,r'_2,r',c')$, then  we can apply the inductive
hypothesis for the premise $\add(r_1,r_2,r,c)$ since we have already shown that
$\delta'=10^{-1}\delta$, therefore
$\delta\geq 4\cdot10^{-(n-1)}$.

We can now conclude by observing that if $\sem{r'_1}+\sem{r'_2}\neq\sem{r'}+c'$, then there exists
$n$ such that $\delta'\geq 4\cdot10^{-n}$, therefore, by the previous result, we deduce that it is not possible
to build a finite tree for $\add(r'_1,r'_2,r',c')$ where the coaxioms are applied at arbitrary depth $k$
(in particular, $k$ is bounded by $n-1$); therefore $\add(r'_1,r'_2,r',c')$ cannot be derived in the inference system.

\medskip

From the proof of soundness we can also deduce that if we let $c$ range over $\Z$, then
 the inference system becomes unsound; for instance,
 it would be possible to build the following infinite proof for $\add(\bar{0},\bar{0},\bar{0},1)$
where all nodes clearly belong to the closure of the coaxioms,
and, hence,
$\add(\bar{0},\bar{0},\bar{0},1)$ would be derivable, but $\sem{\bar{0}}+\sem{\bar{0}}\neq\sem{\bar{0}}+1$:
\[
\Rule
    {
      \Rule
          {
            \vdots
          }
          {\add(\bar{0},\bar{0},\bar{0},10^1)}
    }
    {\add(\bar{0},\bar{0},\bar{0},10^0)}
\]

\subsection{Distances and shortest paths on weighted graphs}\label{sect:graph-example}
As we already said, another widespread non-well-founded structure are graphs.
In \refToSection{coaxioms}, we have shown a first examples concerning graphs, defining the judgements $\Visit{\node}{\nodeset}$, stating that  $\nodeset$ is the set of nodes reachable from $\node$ in the graph.
Essentially, the proposed definition performs a visit of the graph, keeping track of all encountered nodes.
The same pattern can be adopted to solve more complex problems.
For instance, in this section we will deal with  distances between nodes in a \emph{weighted graph}.

Let us introduce the notion of weights for graphs.
In a graph $\Pair{\Nodes}{\adj}$ the set of edges is the set $\Edges \subseteq \Nodes \times \Nodes$  defined {by} $\Edges = \{ \Pair{\node}{\anode} \in \Nodes \times \Nodes \mid \anode \in \adj(\node) \}$.
We will often write $\node\anode$ for an edge $\Pair{\node}{\anode} \in \Edges$.
A weight function is a function $\fun{w}{\Edges}{\N}$.
Here we consider natural numbers as codomain, however we could have considered any other set of non-negative numbers.
Hence, a weighted graph is a graph $\Pair{\Nodes}{\adj}$ together with a weight function $w$.

In a weighted graph $G$, the weight of a path $\alpha$ is the sum of the weights of the edges \FD{(counting repetitions)} {determined by} $\alpha$, we denote this by $w(\alpha)$.
Note that in general the weight of a path $\alpha$ is different from its length, defined as  the number of edges \FD{(counting repetitions)}  {determined by} the path and denoted by $\|\alpha\|$.
The distance between nodes $\node$ and $\anode$ is defined as the minimum weight of a path connecting $\node$ to $\anode$, it is infinite if no such path exists.
Below  we {show} the inference system with coaxioms defining the judgement $\dist{\node}{\anode}{\delta}$ on a weighted graph, where $\delta \in \N \cup \{\infty\}$.
\begin{small}
\begin{mathpar}
\Rule{}{ \dist{\node}{\node}{0} }
\and
\Rule{}{ \dist{\node}{\anode}{\infty} }
{\begin{array}{l}
\node \ne \anode \\
\adj(\node) = \emptyset
\end{array}}
\and
\CoAxiom{ \dist{\node}{\anode}{\infty} } \node\ne\anode
\\
\Rule{
	\dist{\node_1}{\anode}{\delta_1}
	\Space \ldots \Space
	\dist{\node_k}{\anode}{\delta_k}
}{ \dist{\node}{\anode}{\delta} }
{\begin{array}{l}
\node \ne \anode \\
\adj(\node) = \{\node_1, \ldots, \node_k\} \ne \emptyset \\
\delta = \min\{w(\node\node_1) + \delta_1, \ldots,  w(\node\node_k) + \delta_k\}
\end{array} }
\end{mathpar}
\end{small}
In order to show that we cannot simply consider the coinductive interpretation of the above inference system, {and therefore} we need coaxioms, let us consider the following weighted graph:
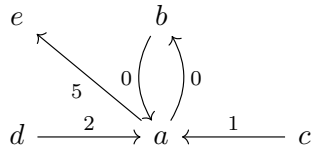
\begin{center}
\begin{tikzcd}[column sep=large, row sep=large]
e             & b \ar[d, bend right, swap, "0"]         &   \\ 
d \ar[r, "2"] & a \ar[u, bend right, swap, "0"] \ar[ul, "5"] & c \ar[l, swap, "1"] 
\end{tikzcd}
\end{center}
If we {would} interpret the inference system coinductively we can derive judgements like $\dist{c}{e}{\delta}$ for any $\delta \in [1..5]$ or $\dist{a}{d}{\delta}$ for any $\delta \in \N \cup \{\infty\}$, as shown in \refToFigure{dist-wrong-trees}.
\begin{figure}
\begin{mathpar}
\Rule{
	\Rule{
		\Rule{}{ \dist{e}{e}{0} }
		\Space
		\Rule{
			\vdots
		}{ \dist{b}{e}{\delta_1} }
	}{ \dist{a}{e}{\delta_1} } { \delta_1\le 5}
}{ \dist{c}{e}{1+\delta_1} }
\and
\Rule{
	\Rule{}{ \dist{e}{d}{\infty} }
	\Space
	\Rule{
		\vdots
	}{ \dist{b}{e}{\delta_2} }
}{ \dist{a}{d}{\delta_2}  }
\end{mathpar}
\caption{{Infinite proof trees} for $\dist{c}{e}{1+\delta_1}$ and $\dist{a}{d}{\delta_2}$}\label{fig:dist-wrong-trees}
\end{figure}
The issue here is the cycle that, having total weight equal to $0$, allows us to build cyclic proofs without increasing the value of $\delta$.
Therefore, the coaxiom is needed to filter out such proofs.
Indeed, it is easy to see that it is not possible to build a finite proof tree  for judgements proved in \refToFigure{dist-wrong-trees}  starting from the coaxiom.

Now we will sketch a proof of correctness.
We can formulate the correctness statement as follows: $\dist{\node}{\anode}{\delta}$ is derivable iff $\delta$ is the minimum of $w(\alpha)$ for all paths $\alpha$ from $\node$ to $\anode$.

\subsubsection*{Completeness}
Let us consider a judgement $\dist{\node}{\anode}{\delta}$ where $\delta$ is the minimum of $w(\alpha)$ for $\alpha$ a path from $\node$ to $\anode$.
{If $\node = \anode$, then $\delta = 0$ and so the judgement is the consequence of the first axiom.
If $\adj(\node) = \emptyset$, then $\delta = \infty$ and so the judgement is the consequence of the second axiom. }
Otherwise, note that $\alpha = \node\beta$ where $\beta$ is a path from a node $\node' \in \adj(\node)$ {to} $\anode$, hence $w(\alpha) = w(\node\node') + w(\beta)$.
Furthermore, if there were another path $\beta'$ from the node $\node'$ to $\anode$ with $w(\beta') < w(\beta)$, then the path $\node\beta'$ would be such that $w(\node\beta') < w(\alpha) = \delta${,} that is absurd, hence $\dist{\node'}{\anode}{w(\beta)}$ is a valid judgement.
Moreover, note that for any other $\node_i \in \adj(\node)$, with $\dist{\node_i}{\anode}{\delta_i}$ a valid judgement, we have $\delta \le w(\node\node_i) + \delta_i$, since, otherwise, we could build a path from $\node$ to $\anode$ with weight smaller than $\delta${,} that is absurd.
Therefore, $\dist{\node}{\anode}{\delta}$ is the consequence of {the} inference rule and {its} premises are valid judgements, and this shows that the specification  is a consistent set.

In order to show the boundedness condition, we have to build a finite proof tree for $\dist{\node}{\anode}{\delta}$ (chosen as before) using coaxioms as axioms.
{If there is no path from $\node$ to $\anode$, then $\node\ne\anode$ and $\delta = \infty$, hence we can apply the coaxiom.
Otherwise, there is a path $\alpha$ from $\node$ to $\anode$ with $w(\alpha) = \delta$.
We proceed by induction on the length of $\alpha$.
If $\|\alpha\| = 0$, then $\node = \anode$ and $\delta = 0$, hence we can apply the first axiom.
If $\|\alpha\| = n+1$, then $\alpha = \node \node' \beta$ with $\|\node'\beta\| = n$, $\node' \in \adj(\node)$, $w(\node'\beta) = \delta'$  and $\delta = w(\node\node') + \delta'$.
By induction hypothesis, we get that $\dist{\node'}{\anode}{\delta'}$ is derivable,
then we get a proof tree for $\dist{\node}{\anode}{\delta}$ by applying
the inference rule with consequence $\dist{\node}{\anode}{\delta}$ and for each $\node'' \in \adj(\node)$ a premise
$\dist{\node''}{\anode}{\infty}$ if $\node'' \ne \node'$ and $\node'' \ne \anode$, which is derivable by the coaxiom,
$\dist{\node''}{\anode}{0}$ if $\node'' = \anode$, which is derivable by the first axiom,
 and $\dist{\node''}{\anode}{\delta'}$ if $\node'' = \node'$, which is derivable by induction hypothesis. }

\subsubsection*{Soundness}
To proove soundness, we first show some useful facts.
\begin{fact}\label{fact:starone}
For all proof trees $t$ for a judgement $\dist{\node}{\anode}{\delta}$, there exists a path $\alpha$ from $\node$ to $\anode$ with $\|\alpha\| = n$ iff there exists a node in $t$ at depth  $n$ labelled by $\dist{\anode}{\anode}{0}$.
\end{fact}
\begin{proof}
Let $t$ be a proof tree for $\dist{\node}{\anode}{\delta}$.
We prove separately the two implications.
\begin{description}
\item[$\Rightarrow$] Let $\alpha$ be a path from $\node$ to $\anode$.
We proceed by induction {on} the length of $\alpha$.
If $\|\alpha\|=0$ (base case), then $\node=\anode$, hence $\dist{\node}{\anode}{\delta}$ has been derived by applying the first axiom, and this implies $\delta = 0$.
Therefore, the root of $t$ (at depth $0$) is labelled by $\dist{\anode}{\anode}{0}$.
If $\|\alpha\| = n+1$ (inductive step), then $\alpha = \node \beta$ where {$\beta$} is a path from a node $\node'$ {to} $\anode$ of length $n$.
Therefore, $\dist{\node}{\anode}{\delta}$ has been derived by applying the inference rule, hence there is a direct subtree of $t$ rooted in $\dist{\node'}{\anode}{\delta'}$, where, by induction hypothesis, $\dist{\anode}{\anode}{0}$ occurs at depth $n$.
Thus, in $t$ that judgement occurs at depth $n+1$ as needed.
\item[$\Leftarrow$] We proceed by induction on the depth $n$.
If $\dist{\anode}{\anode}{0}$ occurs at depth $0$ (base case), then it is the root of $t$, hence $\node = \anode$  and the searched path is the singleton path $\anode$.
If it occurs at depth $n+1$ (inductive step), then the depth of $t$ is greater than $0$, hence $\dist{\node}{\anode}{\delta}$ has been derived by applying the inference rule.
Therefore, $\dist{\anode}{\anode}{0}$ belongs to a direct subtree $t'$ of $t$ rooted in $\dist{\node'}{\anode}{\delta'}$ with $\node' \in \adj(\node)$, and it occurs in $t'$ at depth $n$.
Thus, by induction hypothesis, there is a path $\beta$ from $\node'$ to $\anode$ of length $n$, hence the path $\node \beta$ of length $n+1$ connects $\node$ to $\anode$.
\qedhere
\end{description}
\end{proof}

\begin{fact}\label{fact:startwo}
For all proof trees $t$, $t$ is rooted in $\dist{\node}{\anode}{\infty}$ iff all nodes in $t$ are of shape $\dist{\node'}{\anode}{\infty}$.
\end{fact}
\begin{proof}
Consider a proof tree $t$.
The implication $\Leftarrow$ is trivial.
Let us prove the other one.
{We can rephrase the thesis as follows: if the root of $t$ is $\dist{\node}{\anode}{\infty}$, then, for all $n \in \N$, all nodes  of $t$ at depth $n$ have shape $\dist{\node'}{\anode}{\infty}$.
Thus, we can proceed by induction on the depth $n$.
If the depth is $0$ (base case), then there is only one node at depth $0$, which is the root $\dist{\node}{\anode}{\infty}$, hence the thesis follows immediately by hypothesis.
If the depth is $n+1$ (inductive step), then consider a node $\dist{\node'}{\anode}{\delta}$ at depth $n+1$.
By definition, it is the child of a node at depth $n$, that, by induction hypothesis, is of shape $\dist{\node''}{\anode}{\infty}$.
Therefore, the inference rule has been applied, and, since the conclusion is $\dist{\node''}{\anode}{\infty}$, all premises $\dist{\node'_i}{\anode}{\delta_i}$ with $i \in \{1, \ldots, k\}$ are such that $\delta_i = \infty$, since $\min\{\delta_1, \ldots, \delta_k\} \ge \infty$.
Then, by construction, we have $\node' = \node'_j$ and $\delta = \delta_j$ for some $j \in \{1, \ldots, k\}$, hence we get the thesis.}
\end{proof}

\begin{fact}\label{fact:starthree}
If $\dist{\node}{\anode}{\delta}$ with $\delta \in \N$ has an approximated proof tree (of any level), then there exists a path $\alpha$ from $\node$ to $\anode$ such that $w(\alpha) = \delta$.
\end{fact}
\begin{proof}
Since approximated proof trees are well-founded by \refToDefinition{approx-trees}, we can proceed by induction on the tree structure.
If the tree has a single node (base case), then, since $\delta \in \N$, we have necessarily applied the first axiom, hence $\node = \anode$ and the searched path is the trivial one, which has weight 0.
If the tree is compound (inductive step), we have necessarily applied the inference rule, hence there is a direct subtree rooted in $\dist{\node'}{\anode}{\delta'}$ with $\node' \in \adj(\node)$ and $\delta = w(\node\node') + \delta'$.
Then, by induction hypothesis, there is a path $\alpha'$ from $\node'$ to $\anode$ with $w(\alpha') = \delta'$, hence the path $\alpha = \node \alpha'$ from $\node$ to $\anode$ is such that $w(\alpha) = w(\node\node') + w(\alpha') = \delta$, as needed.
\end{proof}

\begin{fact}\label{fact:starfour}
If $\dist{\node}{\anode}{\delta}$ has an approximated proof tree of level $n$, then $\delta \le w(\alpha)$ for all paths $\alpha$ from $\node$ to $\anode$ with $\|\alpha\| \le n$.
\end{fact}
\begin{proof}
First note that, if $\node = \anode$, then the only applicable rule is the first axiom, hence $\delta = 0$ and the thesis trivially holds, since 0 is the least possible weight.
So, let us assume $\node \ne \anode$ and proceed by induction on the level $n$.
If the level is $0$ (base case), then there is no path from $\node$ to $\anode$ with length $0$, hence we have to show  $\delta \le \infty$, which is always true.
If the level is $n+1$ (inductive step), then, since the level is greater than 0, we have applied either the second axiom or the inference rule.
In the former case, there is no path from $\node$ to $\anode$ since $\adj(\node) = \emptyset$, hence the thesis trivially holds.
In the latter case, assume that $\adj(\node) = \{\node_1, \ldots, \node_k\}$, hence the premises of the rule  are $\dist{\node_i}{\anode}{\delta_i}$ for $i \in \{1, \ldots, k\}$ and  $\delta = \min \{w(\node \node_1) + \delta_1, \ldots, w(\node\node_k) + \delta_k\}$.
Now, consider a path $\alpha$ from $\node$ to $\anode$ with $\|\alpha\| \le n+1$, hence $\alpha = \node \alpha_i$ with $\alpha_i$ a path from $\node_i$ to $\anode$ for some $\node_i \in \adj(\node)$ with $\|\alpha_i\| = n$.
Therefore, $w(\alpha) = w(\node\node_i) + w(\alpha_i)$, and, by induction hypothesis, $\delta_i \le w(\alpha_i)$, hence we get $\delta \le w(\node\node_i) + \delta_i \le w(\node\node_i) + w(\alpha_i) = w(\alpha)$ as needed.
\end{proof}

{To prove soundness,  we have to show that each derivable judgement is valid.
For judgements of shape $\dist{\node}{\anode}{\infty}$ the thesis follows immediately from \refToFact{starone} and \refToFact{startwo}.
Hence, let us assume $\delta \in \N$.
By \refToCorollary{tree2}, the judgement has an approximated proof tree for each level $n \in \N$.
Hence, by \refToFact{starfour}, $\delta \le w(\alpha)$ for all paths $\alpha$ from $\node$ to $\anode$ with $\|\alpha\| \le n$ for each $n \in \N$, that is, simply $\delta \le w(\alpha)$ for all paths $\alpha$ from $\node $ to $\anode$.
Furthermore, by \refToFact{starthree}, $\delta = w(\beta)$ for some path $\beta$ from $\node$ to $\anode$, thus $\dist{\node}{\anode}{\delta}$ is valid. }

\medskip

The notion of distance is tightly related to paths in a graph $G$.
Actually, from the above proofs, it {is} easy to see that a proof tree for a judgement $\dist{\node}{\anode}{\delta}$ explores all possible paths from $\node$ to $\anode$ in the graph in order to compute the distance.
Therefore, in some sense, it also {finds} the shortest path from $\node$ to $\anode$.
Hence, with a slight variation of the inference system  for the distance, we can get an inference system for the judgement $\minPath{\node}{\anode}{\alpha}{\delta}$ stating that $\alpha$ is the shortest path from $\node$ to $\anode$ with weight $\delta$.
We add to paths a special value $\bot$ that represents the absence of paths between two nodes, with the assumption that $\node \bot = \bot$.
The definition is reported in \refToFigure{min-path}.
\begin{figure}
\begin{small}
\begin{mathpar}
\Rule{}{ \minPath{\node}{\node}{\node}{0} }
\and
\Rule{}{ \minPath{\node}{\anode}{\bot}{\infty} }
{\begin{array}{l}
\node \ne \anode \\
\adj(\node) = \emptyset
\end{array} }
\and
\CoAxiom{ \minPath{\node}{\anode}{\bot}{\infty} }
\\
\Rule{
	\minPath{\node_1}{\anode}{\alpha_1}{\delta_1}
	\Space \ldots \Space
	\minPath{\node_k}{\anode}{\alpha_k}{\delta_k}
}{ \minPath{\node}{\anode}{\node \alpha_i}{w(\node\node_i) + \delta_i} }
{ \begin{array}{l}
\scriptstyle
\node \ne \anode \\
\scriptstyle
\adj(\node) = \{\node_1, \ldots, \node_k\} \ne \emptyset \\
\scriptstyle
i = \argmin\{w(\node\node_1) + \delta_1, \ldots,  w(\node\node_k) + \delta_k\}
\end{array} }
\end{mathpar}
\end{small}
\caption{Inference system with coaxioms for $\minPath{\node}{\anode}{\alpha}{\delta}$.}\label{fig:min-path}
\end{figure}

\subsection{Big-step operational semantics with divergence}\label{sect:op-sem}
It is well-known that divergence cannot be captured by the big-step operational semantics of a programming language when
semantic rules are interpreted inductively (that is, in the standard way)~\cite{LeroyGrall09,Ancona12,Ancona14}.
When rules are interpreted coinductively some partial result can be obtained under suitable hypotheses,
but
a practical way to capture divergence  with a big-step operational semantics
is to introduce two different forms of judgement~\cite{CousotCousot92,LeroyGrall09}: one corresponds to the standard big-step
evaluation relation, and is defined inductively, while the other one captures divergence, and is defined coinductively
in terms of the inductive judgement, thus requiring stratification.

With coaxioms a unique judgement can be defined in a more direct and compact way.
Here we show  how this is possible for the standard call-by-value operational semantics of the \LambdaCalculus, but other and more complex applications of coaxioms to model infinite behaviour of programs can be found in~\cite{AnconaDZ17oopsla, AnconaDZ18}.
{For soundness and completeness proofs of this example we refer to~\cite{AnconaDZ17oopsla}.}

\begin{figure}
\begin{center}
Syntax of terms and values
\begin{mathpar}
e ::= v \mid x \mid e\ e
\and
v ::= \lambda x.e
\and
\infv ::= v \mid \infty
\end{mathpar}
\\[2ex]
Semantic rules
\begin{mathpar}
\CoAxiomName
{coax}
{\eval{e}{\infty}}
\and
\RuleName
{val}
{}
{\eval{v}{v}}
\and
\RuleName
{app}
{\eval{e_1}{\lambda x.e}\quad\eval{e_2}{v}\quad\eval{\subs{e}{x}{v}}{\infv}}
{\eval{e_1\ e_2}{\infv}}
\\
\RuleName
{l-inf} 
{\eval{e_1}{\infty}}
{\eval{e_1\ e_2}{\infty}}
\and
\RuleName
{r-inf} 
{\eval{e_1}{v}\quad\eval{e_2}{\infty}}
{\eval{e_1\ e_2}{\infty}}
\end{mathpar}
\end{center}
\caption{Call-by-value big-step semantics of $\lambda$-calculus with divergence}\label{fig:lambda}
\end{figure}
\refToFigure{lambda} defines syntax, values, and semantic rules. The meta-variable $v$ ranges over standard values, that is,
lambda abstractions, while $\infv$ includes also divergence, represented by $\infty$.
The evaluation judgement has the general
shape $\eval{e}{\infv}$, meaning that either $e$ evaluates to a value $v$ (when $\infv\neq\infty$) or diverges (when $\infv=\infty$).

For what concerns the semantic rules, only a coaxiom is needed, stating that every expression may diverge.
This ensures that
$\infty$ is the only allowed outcome for the evaluation of an expression which diverges; this can only happen
when the corresponding derivation tree is infinite.
Rule (val) is standard. Rule (app) deals with the evaluation of application when both expressions $e_1$ and $e_2$
do not diverge; the meta-variable $v$ is required for the judgement \eval{e_2}{v} to guarantee
convergence of $e_2$, while $\infv$ is used for the result of the whole application, since
the evaluation of the body of the lambda abstraction could diverge. As usual, $\subs{e}{x}{v}$ denotes capture-avoiding substitution
modulo $\alpha$-renaming. Rules (l-inf) and (r-inf) cover the cases when either $e_1$ or $e_2$ diverges when trying to evaluate application,
assuming that a left-to-right evaluation strategy has been imposed.

As a paradigmatic example, we consider the expression  $e_\Delta=(\lambda x.x\ x) \lambda x.x\ x$ and show that the only judgement derivable for it is  \eval{e_\Delta}{\infty}. 
To this end, we first disregard the coaxiom and exhibit an infinite derivation tree
for the judgement \eval{e_\Delta}{\infv}, which is valid for all $\infv$:

\vspace{-1em}
\begin{scriptsize}
\[
\RuleName
{app}
{
  \RuleName
  {val}
  {}
  {\eval{\lambda x.x\ x}{\lambda x.x\ x}}
  \quad
  \RuleName
  {val}
  {}
  {\eval{\lambda x.x\ x}{\lambda x.x\ x}}
  \quad
  \RuleName
  {app}
  {\vdots}
  {\eval{\subs{(x\ x)}{x}{\lambda x.x\ x}}{\infv}}}
{\eval{\subs{(x\ x)}{x}{\lambda x.x\ x}=e_\Delta}{\infv}}
\]
\end{scriptsize}

In this particular case the derivation tree is also regular, but of course there are examples of divergent computations
whose derivation tree is not regular. The vertical dots indicate that the derivation continues with the same repeated pattern.
This derivation shows that the coinductive interpretation of the rules in \refToFigure{lambda} has a non-deterministic behaviour,
as happens for the coinductive interpretation of the standard big-step semantics rules~\cite{LeroyGrall09,Ancona12}.
However, here the coaxiom plays a crucial role:
it allows us to filter out all undesired values, leaving only the value $\infty$, which represents  divergence. 
Indeed, if we take into acount the coaxiom, 
we have also to construct finite derivations, where the coaxiom can be used as an axiom. 

For the expression $e_\Delta$, we can build such finite derivations only for the judgement \eval{e_\Delta}{\infty}. 
More precisely, we can easily prove by induction that, if \eval{e_\Delta}{\infv} has a finite proof tree, then $\infv = \infty$. 
Indeed, there are only two cases: 
if we have applied the coaxiom, the thesis is immediate, and, if we have applied the rule (app), then there is a premise \eval{e_\Delta}{\infv}, hence $\infv = \infty$ holds by induction hypothesis. 

As a consequence, in the inference system with the coaxiom, the only derivable judgement for $e_\Delta$ is 
\eval{e_\Delta}{\infty}.


\section{From coaxioms to corules}\label{sect:corules}

As already mentioned, the notion of coaxiom presented in this work has been inspired by operational models  for object-oriented and logic programming proposed in~\cite{AnconaZucca12, AnconaZucca13, Ancona13}.
{Intuitions behind such models lead us to develop a theory where rules added to  an inference system in order to control its semantics have no premises.
In addition, this restriction to coaxioms (rules with no premises) is also motivated by the fact that, in all examples we have considered, they are enough to get the intended semantics.}

However, as we will briefly sketch in this section, all the notions presented until now smoothly generalize to the case where  we can add to an inference system arbitrary rules, with a meaning analogous to the one of coaxioms.
For this reason such rules are named \emph{corules}, and are denoted in the same way as coaxioms.
Furthermore, this extension seems to be needed to deal with more complex examples like those we consider in~\cite{AnconaDZ18}.

Let us introduce the concept more formally.

\begin{defi}\label{def:corules}
An \emph{inference system with corules} is a pair $\Pair{\is}{\cois}$ where $\is$ and $\cois$ are inference systems.
Elements of $\cois$ are called \emph{corules}.
\end{defi}

The semantics is defined in two steps in analogy with coaxioms:
\begin{enumerate}
\item first we consider the inference system $\is \cup \cois$  and take its inductive interpretation $\Ind{\is\cup\cois}$
\item then, we take the coinductive interpretation of $\is$  restricted to rules having consequence in $\Ind{\is\cup\cois}$
\end{enumerate}
Using a notation similar to the one used for coaxioms we have that
\[
\Generated{\is}{\cois} = \CoInd{\Restricted{\is}{\Ind{\is\cup\cois}}}
\]

It is easy to see that an inference systems with coaxioms is a inference system with corules where all corules have no premises.

As we have done for coaxioms, in order to characterize $\Generated{\is}{\cois}$ as a fixed point of $\Op{\is}$, we study the analogous construction in the general framework of complete lattices.

Consider two monotone functions $\fun{\function, \afunction}{\lattice}{\lattice}$ defined on a complete lattice $\Pair{\lattice}{\order}$.
We can consider the monotone function $\function\join \afunction$ defined as the pointwise join of $\function$ and $\afunction$.
Then, we define the \emph{bounded fixed point of $\function$ generated by $\afunction$}, as $\Generated{\function}{\afunction} =  \ker_\function(\lfp(\function \join \afunction))$.
This  is a fixed point of $\function$ thanks to \refToProposition{ck-fp}, since $\lfp(\function \join \afunction)$ is  the least (pre-)fixed point of $\function \join \afunction$ and it is easy to check that all pre-fixed points of $\function \join \afunction$ are pre-fixed point of $\function$. 

Note that, in the case where $\afunction$ is  the constant function $x \mapsto \coaxioms $, we get $\function\join \afunction = \Extended{\function}{\coaxioms}$, hence we have $\Generated{\function}{\afunction} = \Generated{\function}{\coaxioms}$, that is, this construction is a generalization of the bounded fixed point generated by an element.

Then, it is easy to see that $\Generated{\is}{\cois} = \Generated{\Op{\is}}{\Op{\cois}}$, because $\Op{\is \cup \cois} = \Op{\is} \cup \Op{\cois}$.
Therefore $\Generated{\is}{\cois}$ is really a fixed point of $\Op{\is}$ as expected.

On the proof-theoretic side all notions smoothly generalize to this case, indeed, we have that $\judg \in \Generated{\is}{\cois}$ if and only if  there is an arbitrary proof tree in $\is$  for $\judg$, whose nodes have a well-founded derivation in $\is\cup\cois$.
Also the construction of approximated proof trees is the same,  only the starting point changes: this time we start from the set of well-founded proof tree in $\is \cup \cois$.

Proof techniques introduced for coaxioms can be applied also to this more general case,
in particular the bounded coinduction principle can be formulated as follows:
if $\Spec \subseteq \universe$ and
\begin{enumerate}
\item $\Spec \subseteq \Ind{\is\cup \cois}$  and
\item $\Spec \subseteq \Op{\is}(\Spec)$
\end{enumerate}
then, $\Spec \subseteq \Generated{\is}{\cois}$.

At this point a natural question arises: are corules more expressive than coaxioms?
Here more expressive means that they are able to capture more fixed points than coaxioms.
However, considering monotone functions $\fun{\function, \afunction}{\lattice}{\lattice}$, we know from \refToProposition{bfp-fun}, that all fixed points  of $\function$ can be expressed as bounded fixed points generated by themselves, that is, if $z \in \lattice$ is a fixed point, then $z = \Generated{\function}{z}$.
Therefore, since $\Generated{\function}{\afunction}$ is a fixed point of $\function$, there must be $z \in \lattice$, such that $\Generated{\function}{\afunction} = \Generated{\function}{z}$, in particular we can choose $z = \lfp(\function \join \afunction)$.

Therefore at this level, adding corules does not change the expressive power of our framework.
However, it seems that there are cases where corules are fundamental for expressing some definitions, as in~\cite{AnconaDZ18}.
We think that this apparently inconsistency is due to the fact that, in common practice, definitions are expressed through a finite set of finitary meta-rules, while the theory is developed for plain rules (with no variables), and the translation from meta-rules to rules is always left implicit.
Hence, in order to better understand the relationship between coaxioms and corules we need a formal treatment of definitions given by meta-rules, like in~\cite{MomiglianoTiu03, BrotherstonSimpson11}, that is matter of further work.


\section{Related work}\label{sect:related}

Inference systems~\cite{Aczel77,Sangiorgi11} are widely adopted to formally define operational semantics, language translations, type systems,
subtyping relations, deduction calculi, and many other relevant judgements.
Although inference systems have been introduced to deal  with inductive definitions,
in the last two decades several authors have focused on their coinductive interpretation.

Cousot and Cousot~\cite{CousotCousot92}  define divergence of programs by coinductive interpretation of an inference system that extends the big-step operational semantics.
The same approach is followed by other authors~\cite{HughesMoran95, Schmidt98, LeroyGrall09}.
Leroy and Grall~\cite{LeroyGrall09} analyse two kinds of coinductive big-step operational semantics for the call-by-value \LambdaCalculus, and study their relationships with the small-step and denotational semantics, and their suitability for compiler correctness proofs.
Coinductive big-step semantics is used as well to {reason about} cyclic objects stored
in memory~\cite{MilnerTofte91,LeroyRouaix98}, and  to prove type soundness in Java-like languages~\cite{Ancona12,Ancona14}.
Coinductive inference systems are also considered in the context of type analysis and subtyping for object-oriented languages~\cite{AnconaLagorio09, AnconaCorradi14}.

On the programming language side, coinduction is adopted to  provide primitives helping the programmer dealing with infinite objects.
Examples can be found  both in logic programming~\cite{SimonMBG06, SimonBMG07,KomendantskayaJ15} and in functional programming~\cite{Hagino87,BirdWadler88}.
Recently, other approaches have been proposed to support coinduction in a more flexible way.
We can find contributions  in all most popular paradigms:  logic~\cite{Ancona13,MantadelisRM14}, functional~\cite{JeanninKS13, JeanninKS17} and object-oriented~\cite{AnconaZucca12, AnconaZucca13}.
As a consequence, these proposals are more focused on operational aspects, and their corresponding implementation issues.
{Here we discuss those approaches most closely related to coaxioms.

The logic paradigm naturally supports coinduction.
Indeed, a logic program, like an inference system, \EZ{has an associated} monotone function on a suitable power-set lattice, and its declarative semantics is defined as a fixed point of such function~\cite{Lloyd87}; hence, considering the greatest fixed point enables coinduction.
To support non-well-founded objects, the semantics is defined in the power-set of the \emph{complete Herbrand basis}, which is the set of all ground atoms built on finite and infinite terms for the program signature~\cite{Lloyd87}.

In order to support at the same time both inductive and coinductive predicates, in~\cite{SimonMBG06,SimonBMG07} a stratified semantics is proposed: essentially the semantics is well-defined only if there is no mutual dependency beween inductive and coinductive predicates.

On the operational side, two sound resolution strategies have been proposed: \emph{CoSLD resolution}~\cite{SimonMBG06,SimonBMG07,AnconaDovier15} and \emph{structural resolution}~\cite{KomendantskayaJ15}.
The former strategy represents infinite objects through regular terms, that is, terms that can be represented, through unification, as a finite set of syntactic equations~\cite{AdamekMV06b}, hence only cyclic objects are supported.
Then, the resolution is essentially based on a loop detection mechanism and accepts all cyclic derivations.
On the other hand, the latter adopts a lazy approach, working with finite approximation of infinite objects, hence it  requires programs to be \emph{productive}, in order to ensure it is able to construct such finite approximation.
Differently from CoSLD, structural resolution can accept also non-cyclic derivations, but cannot deal with non-productive programs, while coSLD can.

Since coSLD aims to capture all coinductively derivable atoms, it accepts all cyclic proofs, but in some cases this behaviour is too rigid.
To allow  more flexible behaviours, in~\cite{Ancona13, MantadelisRM14} other operational models are provided.}
In particular, the notion of \lstinline!finally! clause, introduce by Ancona~\cite{Ancona13}, allows the programmer to specify a fact that {should be} resolved when a cycle is detected, instead of simply accepting the atom.
In this way,  predicates that are neither purely inductive nor purely coinductive can be defined and used  in a logic program.

The notion of \lstinline!finally! clause has inspired coaxioms {as described} in the introduction.
However, despite the existing strong correlation with coaxioms, the semantics of \lstinline!finally! clauses does not always coincide with a fixed point of the inference operator induced by the program. This is a {relevant} difference with coaxioms, that, instead, always generate a fixed point.

\Rev{To overcome this issue of the \lstinline!finally! clause}, we have designed an extension of coinductive logic programming supporting coaxioms, in this context called \emph{cofacts}~\cite{AnconaDZ17coalp}.
Here the declarative semantics is based on the bounded fixed point, and the resolution procedure is a combination of standard SLD and coSLD resolutions: when the latter discovers a loop, then a standard SLD resolution is triggered, which takes into account also cofacts.
We have also implemented a prototype meta-interpreter in SWI-Prolog\footnote{Available at \url{http://www.disi.unige.it/person/AnconaD/Software/co-facts.zip}}.

In the object-oriented paradigm cyclic objects are usually managed relying on imperative features, thus the language does not provide any native support for computing with such objects.
The programmer has to implement ad-hoc machinery to deal with cyclic objects in an appropriate way, and this is often involved and error-prone.

In order to overcome these difficulties, Ancona and Zucca~\cite{AnconaZucca12, AnconaZucca13} have proposed an extension of Featherweight Java (\FJ)~\cite{IgarashiPW99}: corecursive Featherweight Java (\coFJ).
This is a purely functional core calculus for Java-like languages supporting cyclic objects and corecursive methods.

Cyclic objects are represented by syntactic equations.
They cannot be directly written by the programmer,
but only built during the execution by corecursive methods.
Analogously to the coSLD, each corecursive  call is evaluated in an environment associating to already encountered calls a unique label.
If the call is in the environment, then the associated label is returned as result, otherwise a fresh label is associated to the current call, and the method body is evaluated in the extended environment;
finally, an equation for this new label is returned as result.

To make the mechanism more flexible, like in the logic paradigm,  the authors introduce a \lstinline!with! clause, which is an expression that will be evaluated when a cycle is detected, instead of simply returning the label, and this provides support for methods that are neither purely recursive nor purely corecursive.
Again like {in} the logic paradigm, this feature has inspired coaxioms and is strongly related to them, however the semantics of \lstinline!with!  clauses may not always correspond to a fixed point, while coaxioms always generate a fixed point.


\section{Conclusions}\label{sect:conclu}

Inference systems are a general and versatile framework that is well-known and widely used.
{They allow} to define several kinds of judgements{,} from operational semantics to type systems, from deduction calculi to language translations.
They can also serve as theory to reason about recursive definitions, providing a rigorous semantics in a quite simple way.

However, standard inference systems suffer from a strong rigidity: their interpretation is dichotomous, either inductive (the least one) or coinductive (the greatest one), but what can we do if we need something in the middle?
One may {wonder whether} this is a real issue, but the examples we have provided shows that there are {many interesting} cases in which we need a fixed point that is neither the least nor the greatest one, and {standard} inference system{s} are not able to provide such flexibility.
Therefore{,} in this paper we {have proposed} an extension of inference systems, aimed to  provide more flexibility  in {such} cases, without affecting standard behaviour.

The core of this paper is the concept of inference system with coaxioms, introduced in \refToSection{coaxioms}: a generalized notion of inference system, that subsumes the standard one, supporting  more flexible interpretations.

Our work {originates} from the operational models, closely related to each other,  introduced by Ancona and Zucca~\cite{AnconaZucca12, AnconaZucca13} and Ancona~\cite{Ancona13}.
As already discussed, these operational semantics {introduce} some flexibility for interpreting predicates and functions recursively defined on non-well-founded data types.
The initial {objective of our work} was {to provide} a more abstract semantics for such operational models, hence we developed a first model in~\cite{AnconaDZ16}  focused on this aim.
However{,} the result was not satisfactory, since we managed to capture the semantics of a restricted class of definitions, with a model that was quite tricky.

Then, we decided to take a more abstract perspective, considering inference systems as reference framework.
In this context we discovered the notion of coaxioms, that convinced us to be the right one.
We firstly proposed {it} in~\cite{AnconaDZ17esop} and discussed it in more detail in the master thesis~\cite{Dagnino17}, from which this paper is extracted.

In order to finely describe coaxioms, we have generalized the meta-theory of inference systems by
providing two equivalent semantics, one based on fixed points in a complete
lattice, and the other on the notion of proof tree.

{On the model-theoretic side (\refToSection{coaxioms-model}), we have defined the \emph{bounded fixed point} of a monotone function on a complete lattice generated by an element, which is the greatest (post-)fixed point of {the} corresponding one step 
inference operator, below the least pre-fixed point above the generator;
this turned out to capture the semantics of inference systems with coaxioms (\refToTheorem{correspondence}).
An important property is that the bounded fixed point can be obtained as a combination of a least and a greatest fixed point of suitable functions (\refToProposition{ck-alt}).

From the proof-theoretic perspective, we have provided three different  equivalent  semantics.
All of them essentially impose a condition on coinductive proof trees\footnote{Here we mean proof trees valid for the coinductive interpretation, hence both well-founded and non-well-founded proof trees{.}}   to be accepted, induced by coaxioms.
In other words, all these conditions {allow} us to filter out undesired derivations.
Since in literature we have not found {a rigorous enough} (for our aims) treatment of the standard proof-theoretic semantics of inference systems, we have developed our proof-theoretic model  in more detail, starting from a very precise notion of tree (see \refToSection{trees-graphs}).

The first characterization (\refToSection{trees1}) requires that each judgement in the tree is derivable with a well-founded proof tree  in the extended inference system (the inference system where coaxioms are considered as axioms).
The other two characterizations (\refToSection{approx-trees})  are based on the notion of \emph{approximated proof trees of level $n$}, that are well-founded proof trees in the extended inference system where coaxioms can only be used at depth greater than $n$.
We have proved that all these proof-theoretic semantics are equivalent to each other and to the fixed point semantics.

%

We have also developed proof techniques to reason with coaxioms (\refToSection{coaxioms-reasoning}).
For completeness proofs we have generalized the standard coinduction principle, taking into account also coaxioms,
while for soundness proofs we have described a technique based on approximated proof trees and reasoning by contraposition.  }

Finally, in \refToSection{corules} we have defined a further extension of our framework, allowing also corules, that is, rules used in the same way as coaxioms, but that can have non-empty premises.

\subsubsection*{Further work}
Starting from this work, we identify three main directions for further investigations:
\begin{enumerate}
\item {deepening the theory of coaxioms,}
\item defining language constructs to support flexible (co)inductive  definitions of data types, predicates and functions, 
\item {using coaxioms/corules to model and reason about}  infinite behaviours of programs and systems.
\end{enumerate}

\noindent
In the first direction, a compelling \EZ{topic} for further developments is exploring other proof techniques {for coaxioms, trying to extend proof techniques known for coinduction to this generalized framework.
Possible examples are techniques based on parametrization~\cite{HurNDV13}, or up-to techniques~\cite{Pous07}.


Another important goal we would like to pursue is to provide the support for coaxioms in a proof assistant, such as Agda~\cite{Agda} or Coq~\cite{Coq}, to have a tool to mechanize proofs.
In type theories supporting inductive and coinductive types~\cite{Hagino87,AbelPTS13,AbelPientka13,Basold18}, like the one at the basis of Agda, we can implement inference systems with coaxioms, representing proof trees as a coinductive type, where each node is annotated  by a finite proof tree (given by an inductive type).
What would be interesting is to hide this construction, in such a way that the programmer has only to care about specifying  rules and coaxioms, leaving everything else to the engine. }

An open problem concerning the interpretation generated by coaxioms is its computability.
It is quite obvious that in general this set is not decidable, however it could be interesting {to study} conditions and/or restrictions that ensure at least  that {it} is semi-decidable.
To this aim{,} it could be useful trying to provide another proof-theoretic characterization based on partial proof trees, that are proof trees with assumptions, and form a complete partial order.

Another question concerns the expressive power of this framework.
Here for expressive power we mean how many fixed point{s} of the inference operator we manage to capture using coaxioms.
As we noticed in \refToSection{coaxioms-model}, all fixed points can be generated by a set of coaxioms: it is enough to take as generator the fixed point itself (\refToProposition{bfp-fun}).
However, this sounds not very relevant, since we get something that we already have.
Actually inference systems, and hence inference systems with coaxioms, are never used in the form they are regarded in the development of the meta-theory, but, rather, they are expressed using  a \emph{finite set of meta-rules}, leaving implicit the step from meta-rules to plain rules, which, instead, are considered in the meta-theory.
At this level, the above question becomes more interesting,
however, to deal with this problem, we first need to clarify what is an inference system in terms of meta-rules, filling in the gap between them and plain rule.
{To this aim,  interesting starting points could be~\cite{MomiglianoTiu03, BrotherstonSimpson11}, which discuss proof systems for first-order logics with a notion of inductive and/or coinductive definition}.
Then, in that setting, we would be able to reason about the expressive power of the resulting framework.


Another interesting development is to investigate a variant of the model able to directly capture the definition of functions, rather than representing them as functional relations.
This would be relevant to more appropriately model language constructs to support flexible corecursion in  functional languages.
This variant could also imply a change of framework, moving from lattice theory to domain or category theory, where the semantics of (co)recursive definitions of functions is better supported. 
Therefore a deeper comparison  between coaxioms and category-theoretic or type-theoretic models could be useful.

In the second direction, considering language support for flexible coinduction, we have already taken the first steps by prividing a support for coaxioms in  the logic paradigm~\cite{AnconaDZ17coalp}.

{Extending the notion of coaxioms to support more flexible semantics for recursively defined functions in the object-oriented and functional paradigms is more challenging, due to the gap between the underlying theories.
The simplest idea would be to view functions as relations, which are the {entities} managed by inference systems with coaxioms, however we have always to ensure that the generated fixed point is actually a function, and this is not always guaranteed. }

For the object-oriented paradigm  a starting point could be the revision of the operational semantics of \coFJ~\cite{AnconaZucca12, AnconaZucca13} on the basis of the abstract model provided by coaxioms;
in particular{,} to guarantee that the function denoted by a function definition in \coFJ is actually a fixed point of  the induced monotone operator.
{For the functional paradigm the situation is even more challenging, since we have to deal with more complex constructs such as higher order functions and pattern matching.
A similar problem is addressed in~\cite{JeanninKS13,JeanninKS17}, which could be an interesting starting point. }

In the last direction, starting from the example in \refToSection{op-sem}, it could be interesting to better study the capabilities of coaxioms to model  non-termination.
We have already done a first step in this direction in~\cite{AnconaDZ17oopsla}, where we apply the approach sketched in \refToSection{op-sem} to an imperative \FJ-like language, studying in particular application of proof techniques for coaxioms to prove the soundness of predicates (such as typing relations) with respect to the operational semantics.

A further extension in this direction would be applying coaxioms to define trace-based operational semantics~\cite{NakataUustalu09}, that allow to capture finer characterizations of the behaviour of non-terminating programs.
In this context, corules seem to be needed to properly define the semantics, as we started studying in~\cite{AnconaDZ18}.


\section*{Acknowledgements}
Coaxioms have been initially designed together with professors Davide Ancona and Elena Zucca from University of Genova, who have supervised the \EZ{master} thesis from which this paper is extracted; their suggestions have been essential to better structure this work and to make it clearer and more readable.
Special thanks go also to professor Eugenio Moggi, who has reviewed the thesis, providing several useful comments.
I also would like to thank all the reviewers for their remarks, which have been really important to improve this work.

\bibliographystyle{alpha}
\bibliography{refs}

\end{document}